\newif\iffull
\newif\ifnotes
\newif\iflater

\fulltrue

\notesfalse

\latertrue

\iffull
\documentclass[11pt,letter]{article}
\usepackage[\ifnotes
draft,
notes=true,
\iflater later=true \else later=false \fi
\else
notes=false
\fi]{dtrt}
\usepackage{fullpage}
\else
\documentclass[runningheads]{llncs}
\usepackage[notes=xxx]{dtrt}
\fi

\iffull
\usepackage{amsthm}
\fi
\usepackage[utf8]{inputenc}
\usepackage{times}
\usepackage{amsmath}
\usepackage{comment}
\usepackage{xspace}
\usepackage{paralist}
\usepackage{enumitem}
\setlist[itemize]{leftmargin=*}
\setlist[enumerate]{leftmargin=*}
\usepackage{makecell}
\usepackage{tikz}
  \usetikzlibrary{matrix,shapes,arrows,positioning,chains,calc,shapes.arrows}
\usepackage{floatrow}
\usepackage{mleftright}
\iffull
\usepackage[margin=1cm,footnotesize]{caption}
\else
\usepackage[footnotesize]{caption}
\fi
\usepackage{bbm}
\usepackage{tabu}
\usepackage{multirow}
\usepackage{mathrsfs}
\usepackage{adjustbox}
\usepackage{microtype}
\usepackage{etoolbox}
\usepackage{graphicx}
\usepackage{breakcites}
\usepackage{mdframed}
\usepackage{soul}
\usepackage{amssymb}
\usepackage{changepage}
\usepackage{braket}
\usepackage{enumitem}
\usepackage[capitalize,nameinlink]{cleveref}
\usepackage[backend=bibtex8,style=alphabetic,maxnames=15,sorting=anyt]{biblatex}
\usepackage{bookmark}
\usepackage{setspace}
\bibliography{references}

\makeatletter
\def\grd@save@target#1{%
  \def\grd@target{#1}}
\def\grd@save@start#1{%
  \def\grd@start{#1}}
\tikzset{
  grid with coordinates/.style={
    to path={%
      \pgfextra{%
        \edef\grd@@target{(\tikztotarget)}%
        \tikz@scan@one@point\grd@save@target\grd@@target\relax
        \edef\grd@@start{(\tikztostart)}%
        \tikz@scan@one@point\grd@save@start\grd@@start\relax
        \draw[minor help lines] (\tikztostart) grid (\tikztotarget);
        \draw[major help lines] (\tikztostart) grid (\tikztotarget);
        \grd@start
        \pgfmathsetmacro{\grd@xa}{\the\pgf@x/1cm}
        \pgfmathsetmacro{\grd@ya}{\the\pgf@y/1cm}
        \grd@target
        \pgfmathsetmacro{\grd@xb}{\the\pgf@x/1cm}
        \pgfmathsetmacro{\grd@yb}{\the\pgf@y/1cm}
        \pgfmathsetmacro{\grd@xc}{\grd@xa + \pgfkeysvalueof{/tikz/grid with coordinates/major step}}
        \pgfmathsetmacro{\grd@yc}{\grd@ya + \pgfkeysvalueof{/tikz/grid with coordinates/major step}}
        \foreach \x in {\grd@xa,\grd@xc,...,\grd@xb}
        \node[anchor=north] at (\x,\grd@ya) {\pgfmathprintnumber{\x}};
        \foreach \y in {\grd@ya,\grd@yc,...,\grd@yb}
        \node[anchor=east] at (\grd@xa,\y) {\pgfmathprintnumber{\y}};
      }
    }
  },
  minor help lines/.style={
    help lines,
    step=\pgfkeysvalueof{/tikz/grid with coordinates/minor step}
  },
  major help lines/.style={
    help lines,
    line width=\pgfkeysvalueof{/tikz/grid with coordinates/major line width},
    step=\pgfkeysvalueof{/tikz/grid with coordinates/major step}
  },
  grid with coordinates/.cd,
  minor step/.initial=.2,
  major step/.initial=1,
  major line width/.initial=2pt,
}
\makeatother

\newcommand{\doclearpage}{%
	\iffull
	\clearpage
	\else
	\fi
}

\iffull
\newcommand{\keywords}[1]{\bigskip\par\noindent{\small\textbf{Keywords\/}: #1}}
\fi

\newcommand{\defemph}[1]{\textbf{\emph{#1}}}

\newcommand{\FormatAuthor}[3]{
	\begin{tabular}{c}
		#1 \\ {\small\texttt{#2}} \\ {\small #3}
	\end{tabular}
}

\theoremstyle{plain} 

\newtheorem{theorem}{Theorem}[section]
\newtheorem{lemma}[theorem]{Lemma}
\newtheorem{proposition}[theorem]{Proposition}
\newtheorem{claim}[theorem]{Claim}
\newtheorem{corollary}[theorem]{Corollary}
\newtheorem{definition}[theorem]{Definition}

\theoremstyle{definition} 

\newtheorem{construction}[theorem]{Construction}

\newtheorem{remark}[theorem]{Remark}

\newtheorem{openproblem}{Open Problem}
\newtheorem{introtheorem}{Theorem}

\theoremstyle{remark} 

\crefname{step}{Step}{Steps}

\newcommand{\CB}{\allowbreak}

\newcommand{\DoQuote}[1]{``#1''}

\newcommand{\pair}[2]{(#1 ,\CB #2)}

\DeclareMathOperator{\poly}{poly}
\DeclareMathOperator{\polylog}{polylog}

\DeclareMathOperator{\Expectation}{\mathbb{E}}

\newcommand{\Bits}{\{0,1\}}

\newcommand{\Naturals}{\mathbb{N}}

\newcommand{\DefineEqual}{:=}

\renewcommand{\Set}[1]{\{#1\}}
\newcommand{\SetCardinality}[1]{|#1|}

\newcommand{\BitSize}[1]{|#1|}

\newcommand{\Proof}{\pi\xspace}

\newcommand{\FormatComplexityClass}[1]{\mathbf{#1}}
\newcommand{\NTIME}{\FormatComplexityClass{NTIME}}
\newcommand{\BPP}{\FormatComplexityClass{BPP}}
\newcommand{\NP}{\FormatComplexityClass{NP}}

\newcommand{\sharpP}{\FormatComplexityClass{\#P}}
\newcommand{\PSPACE}{\FormatComplexityClass{PSPACE}}
\newcommand{\NEXP}{\FormatComplexityClass{NEXP}}

\newcommand{\IPCP}{\FormatComplexityClass{IPCP}}
\newcommand{\MIPS}{\FormatComplexityClass{MIP^*}}

\newcommand{\QMA}{\FormatComplexityClass{QMA}}
\newcommand{\QIP}{\FormatComplexityClass{QIP}}

\newcommand{\QSZKHV}{\FormatComplexityClass{QSZK_{HV}}}
\newcommand{\QSZK}{\FormatComplexityClass{QSZK}}
\newcommand{\PZKIPCP}{\FormatComplexityClass{PZK\mbox{-}IPCP}}

\newcommand{\PZKMIP}{\FormatComplexityClass{PZK\mbox{-}MIP}}
\newcommand{\PZKMIPS}{\FormatComplexityClass{PZK\mbox{-}\MIPStar{}}}
\newcommand{\ZKMIPS}{\FormatComplexityClass{ZK\mbox{-}\MIPStar{}}}

\newcommand{\Domain}{D}
\newcommand{\Range}{R}
\newcommand{\SubDomain}{\tilde{\Domain}}

\newcommand{\Restrict}[2]{#1|_{#2}}

\newcommand{\Field}{\mathbb{F}}
\newcommand{\SubField}{\mathbb{K}}
\newcommand{\FieldSize}{q}

\newcommand{\Multiplicative}[1]{#1^{\times}}

\newcommand{\VariableX}{X}
\newcommand{\VariableY}{Y}

\newcommand{\Poly}{Q}

\newcommand{\PolyA}{P}
\newcommand{\Lagrange}[1]{I_{#1}}

\NewDocumentCommand{\IndividualDegree}{m o}{\IfValueTF{#2}{\mathrm{deg}_{#2}(#1)}{\mathrm{deg}(#1)}}

\newcommand{\PolynomialRing}[3]{#1[#3_{1,\dots,#2}]}
\newcommand{\PolynomialRingIndOne}[4]{#1[#3_{1,\dots,#2}^{\leq #4}]}
\newcommand{\PolynomialRingIndOneXY}[7]{#1[#3_{1,\dots,#2}^{\leq #6},#5_{1,\dots,#4}^{\leq #7}]}

\newcommand{\pST}{\; \middle\vert \;}
\newcommand{\MakeDistribution}[1]{\mathcal{#1}}
\newcommand{\Distribution}{\MakeDistribution{D}}

\newcommand{\Relation}{\mathscr{R}}
\newcommand{\Language}{\mathscr{L}}
\newcommand{\Witnesses}[2]{#1\vert_{#2}}

\newcommand{\GetLanguage}[1]{\mathrm{Lan}(#1)}

\newcommand{\Instance}{x}
\newcommand{\Witness}{w}
\newcommand{\InstanceSize}{n}

\newcommand{\DeciderMachine}{M}
\newcommand{\DeciderTime}{T}

\newcommand{\Prover}{P\xspace}
\newcommand{\Verifier}{V\xspace}
\newcommand{\Simulator}{S\xspace}

\newcommand{\ProverStrategy}{P}
\newcommand{\ProverStrategyLD}{P_{\textrm{LD}}}

\newcommand{\SCSubset}{H}
\newcommand{\SCVars}{m}
\newcommand{\SCDegree}{d}
\newcommand{\SCPoly}{F}
\newcommand{\SCSum}{a}

\newcommand{\RandPoly}{R}
\newcommand{\MaskedPoly}{Q}
\newcommand{\AnsTable}[1]{\mathsf{ans}_{#1}}

\newcommand{\SCStrength}{\QueryBound}
\newcommand{\CodeSimAlgorithm}{\mathcal{A}}

\newcommand{\ListSize}{\ell}
\newcommand{\SlowSimulator}{\Simulator_{\mathrm{slow}}}

\newcommand{\SubsetSize}{\lambda}
\newcommand{\SSCVars}{k}
\newcommand{\SSCSubset}{G}
\newcommand{\StrongRandPoly}{Z}
\newcommand{\AuxRandPoly}{A}
\newcommand{\SoundnessSet}{I}

\newcommand{\Interact}[2]{\langle #1,#2 \rangle}

\newcommand{\InnerProduct}[2]{\langle #1,#2 \rangle}

\DeclareMathOperator{\rank}{rank}

\newcommand{\Malicious}[1]{\expandafter\tilde#1}
\newcommand{\Simulated}[1]{#1_{\mathrm{sim}}}

\newcommand{\View}{\mathrm{View}}

\newcommand{\IPCPView}[2]{\View\;\Interact{#1}{#2}}
\newcommand{\MIPView}[2]{\View\;\Interact{#1}{#2}}

\newcommand{\Distance}[2]{\Delta(#1,#2)}

\newcommand{\Characteristic}[1]{\mathrm{char}(#1)}

\newcommand{\Plane}{s}
\newcommand{\Line}{\ell}
\newcommand{\Planes}[1]{\mathrm{Planes}(#1)}
\newcommand{\Lines}[1]{\mathrm{Lines}(#1)}
\newcommand{\RMSubDomain}{H}
\newcommand{\RMVars}{m}
\newcommand{\RMDegree}{d}
\newcommand{\OracleSATRelation}{\Relation_{\mathrm{O3SAT}}}

\newcommand{\FormatComplexityFunction}[1]{\mathsf{#1}}

\newcommand{\SoundnessError}{\FormatComplexityFunction{\varepsilon}}
\newcommand{\ProofLength}{\FormatComplexityFunction{l}}
\newcommand{\QueryComplexity}{\FormatComplexityFunction{q}}
\newcommand{\CommunicationComplexity}{\FormatComplexityFunction{c}}

\newcommand{\ProximityParameter}{\FormatComplexityFunction{\delta}}

\newcommand{\RoundComplexity}{\FormatComplexityFunction{r}}
\newcommand{\QueryBound}{\FormatComplexityFunction{b}}

\newcommand{\NumProvers}{\FormatComplexityFunction{k}}

\newcommand{\eps}{\ensuremath{\epsilon}\xspace}
\newcommand{\F}{\ensuremath{\mathbb{F}}\xspace}
\newcommand{\E}{{\mathbb{E}}}
\newcommand{\x}{\ensuremath{\mathbf{x}}}

\renewcommand{\deg}[1]{\ensuremath{\mathrm{deg}(#1)}}

\newcommand{\MIPStar}{MIP\textsuperscript{*}}
\newcommand{\LDIPCP}{$(\Field, \SCDegree, \SCVars)$-low-degree IPCP}
\newcommand{\qstate}{entangled state}

\newcommand{\ProverMain}{\mathcal{P}_{\textrm{main}}}
\newcommand{\ProverPlane}{\mathcal{P}_{\textrm{plane}}}
\newcommand{\ProverPoint}{\mathcal{P}_{\textrm{lookup}}}

\newcommand{\Measurement}[3]{#1_{#2}^{#3}}

\newcommand{\Id}{\mathrm{Id}}
\newcommand{\HilbertSpace}{\mathcal{H}}
\newcommand{\LinearOperators}[1]{\mathcal{L}(#1)}
\newcommand{\Trace}{\mathrm{Tr}}
\newcommand{\TraceRho}[1]{\mathrm{Tr}_{\rho}\left( #1 \right)}
\newcommand{\NumRegisters}{r}
\newcommand{\InnerProductPsi}[2]{\langle #1,#2 \rangle_{\Psi}}
\newcommand{\NormPsi}[1]{\left\| #1 \right\|_{\Psi}}
\newcommand{\NormOne}[1]{\left\| #1 \right\|_{1}}
\newcommand{\Complex}{\mathbb{C}}
\newcommand{\Oracle}{R}
\newcommand{\Input}{x}
\newcommand{\Point}{\alpha}
\newcommand{\Element}{z}
\newcommand{\State}{\sigma}
\newcommand{\Register}[1]{\mathcal{#1}}
\newcommand{\Transformation}{T}

\newcommand{\LD}[1]{\hat{#1}}

\makeatletter
\newcommand{\subalign}[1]{%
  \vcenter{%
    \Let@ \restore@math@cr \default@tag
    \baselineskip\fontdimen10 \scriptfont\tw@
    \advance\baselineskip\fontdimen12 \scriptfont\tw@
    \lineskip\thr@@\fontdimen8 \scriptfont\thr@@
    \lineskiplimit\lineskip
    \ialign{\hfil$\m@th\scriptstyle##$&$\m@th\scriptstyle{}##$\crcr
      #1\crcr
    }%
  }
}
\newcommand{\IPCPparams}[5]{{\IPCP \left[ \subalign{
		\textsf{round complexity:}&\enspace {#2} \\
		\textsf{PCP length:}&\enspace {#3} \\
		\textsf{communication complexity:}&\enspace {#4} \\
		\textsf{query complexity:}&\enspace {#5} \\
		\textsf{soundness error:}&\enspace {#1} 
		} \right]}} 
\newcommand{\PZKIPCPparams}[6]{{\PZKIPCP \left[ \subalign{
		\textsf{round complexity:}&\enspace {#2} \\
		\textsf{PCP length:}&\enspace {#3} \\
		\textsf{communication complexity:}&\enspace {#4} \\
		\textsf{query complexity:}&\enspace {#5} \\
		\textsf{query bound:}&\enspace {#6} \\
		\textsf{soundness error:}&\enspace {#1}
		} \right]}} 		
\newcommand{\MIPSparams}[4]{{\MIPS \left[ \subalign{
		\textsf{number of provers:}&\enspace {#1} \\
		\textsf{round complexity:}&\enspace {#3} \\
		\textsf{communication complexity:}&\enspace {#4} \\
		\textsf{soundness error:}&\enspace {#2}
		} \right]}} 
\newcommand{\PZKMIPSparams}[4]{{\PZKMIPS \left[ \subalign{
		\textsf{number of provers:}&\enspace {#1} \\
		\textsf{round complexity:}&\enspace {#3} \\
		\textsf{communication complexity:}&\enspace {#4} \\
		\textsf{soundness error:}&\enspace {#2} \\
		} \right]}}
\newcommand{\LDIPCPparams}[6]{{\left| \subalign{
		\textsf{round complexity:}&\enspace {#2} \\
		\textsf{PCP length:}&\enspace {#3} \\
		\textsf{communication complexity:}&\enspace {#4} \\
		\textsf{query complexity:}&\enspace {#5} \\
		\textsf{oracle }\in&\enspace {#6} \\
		\textsf{soundness error:}&\enspace {#1} \\
		} \right|}}

\begin{document}

\title{%
Spatial Isolation Implies Zero Knowledge \\
Even in a Quantum World%
$^\dagger$\footnotetext{$^\dagger$This work was supported in part by the UC Berkeley Center for Long-Term Cybersecurity. A part of the earlier technical report \cite{ChiesaFS17} was merged with this work.}
%
}
\author{
\begin{tabular}[h!]{ccc}
\FormatAuthor{Alessandro Chiesa}{alexch@berkeley.edu}{UC Berkeley}
 & \FormatAuthor{Michael A. Forbes}{miforbes@illinois.edu}{University of Illinois at Urbana--Champaign} \\
 & \\
\FormatAuthor{Tom Gur}{tom.gur@berkeley.edu}{UC Berkeley}
 & \FormatAuthor{Nicholas Spooner}{nick.spooner@berkeley.edu}{UC Berkeley}
\end{tabular}
}

\iffull
\date{\today}
\fi

\maketitle

\begin{abstract}
Zero knowledge plays a central role in cryptography and complexity.
The seminal work of Ben-Or et al.\ (STOC 1988) shows that zero knowledge can be achieved unconditionally for any language in $\NEXP$, as long as one is willing to make a suitable \emph{physical assumption}: if the provers are spatially isolated, then they can be assumed to be playing independent strategies.

Quantum mechanics, however, tells us that this assumption is unrealistic, because spatially-isolated provers could share a quantum entangled state and realize a non-local correlated strategy. The \MIPStar{} model captures this setting.

\medskip

In this work we study the following question: \emph{does spatial isolation still suffice to unconditionally achieve zero knowledge even in the presence of quantum entanglement?} \

\medskip

We answer this question in the affirmative: we prove that every language in $\NEXP$ has a $2$-prover \emph{zero knowledge} interactive proof that is sound against entangled provers; that is, $\NEXP \subseteq \ZKMIPS$.

Our proof consists of constructing a zero knowledge interactive PCP with a strong algebraic structure, and then lifting it to the \MIPStar{} model. This lifting relies on a new framework that builds on recent advances in low-degree testing against entangled strategies, and clearly separates classical and quantum tools.

Our main technical contribution is the development of new algebraic techniques for obtaining unconditional zero knowledge; this includes a zero knowledge variant of the celebrated sumcheck protocol, a key building block in many probabilistic proof systems. A core component of our sumcheck protocol is a new algebraic commitment scheme, whose analysis relies on algebraic complexity theory.

\keywords{zero knowledge; multi-prover interactive proofs; quantum entangled strategies; interactive PCPs; sumcheck protocol; algebraic complexity}
\end{abstract}

\iffull
\clearpage
\setcounter{tocdepth}{2}
\begin{spacing}{0.9}
{\footnotesize \tableofcontents}
\end{spacing}
\clearpage
\fi

\section{Introduction}
\label{sec:intro}

Zero knowledge, the ability to demonstrate the validity of a claim without revealing any information about it, is a central notion in cryptography and complexity that has received much attention in the last few decades. Introduced in the seminal work of Goldwasser, Micali, and Rackoff \cite{GoldwasserMR89}, zero knowledge was first demonstrated in the model of interactive proofs, in which a resource-unbounded prover interacts with a probabilistic polynomial-time verifier to the end of convincing it of the validity of a statement.

Goldreich, Micali, and Wigderson \cite{GoldreichMW91} showed that every language in $\NP$ has a \emph{computational} zero knowledge interactive proof, under the cryptographic assumption that (non-uniform) one-way functions exist. Ostrovsky and Wigderson \cite{OstrovskyW93} proved that this assumption is necessary.

Unfortunately, the stronger notion of \emph{statistical} zero knowledge interactive proofs, where both soundness and zero knowledge hold unconditionally, is limited. For example, if $\NP$ had such proofs then the polynomial hierarchy would collapse to its second level \cite{BoppanaHZ87,Fortnow87,AielloH91}.

The celebrated work of Ben-Or et al.\ \cite{BenOrGKW88} demonstrated that the situation is markedly different when the verifier interacts with \emph{multiple} provers, in a \emph{classical world} where by spatially isolating the provers we ensure that they are playing independent strategies --- this is the model of multi-prover interactive proofs (MIPs). They proved that every language having an MIP (i.e., every language in $\NEXP$ \cite{BabaiFL91}) also has a \emph{perfect} zero knowledge MIP. This result tells us that \emph{spatial isolation implies zero knowledge}.

In light of quantum mechanics, however, we know that spatial isolation \emph{does not} imply independence, because the provers could share an \qstate{} and realize a strategy that is beyond that of independently acting provers. For example, it is possible for entangled provers to win a game (e.g., the magic square game) with probability $1$, whereas independent provers can only win with probability at most $8/9$ \cite{CleveHTW04}.

Non-local correlations arising from local measurements on entangled particles play a fundamental role in physics, and their study goes back at least to Bell's work on the Einstein--Podolsky--Rosen paradox \cite{Bell64}. Recent years have seen a surge of interest in MIPs with \emph{entangled provers}, which correspond to the setting in which multiple non-communicating provers share an \qstate{} and wish to convince a classical verifier of some statement. This notion is captured by \MIPStar{} protocols, introduced by Cleve et al.\ \cite{CleveHTW04}. A priori it is unclear whether these systems should be less powerful than standard MIPs, because of the richer class of \emph{malicious} prover strategies, or more powerful, because of the richer class of \emph{honest} prover strategies.

Investigating proof systems with entangled adversaries not only sharpens our understanding of entanglement as a computational resource, but also contributes insights to hardness of approximation and cryptography in a post-quantum world. However, while the last three decades saw the development of powerful ideas and tools for designing and analyzing proof systems with classical adversaries, despite much effort, there are only a handful of tools available for dealing with quantum entangled adversaries, and many fundamental questions remain open.
 
\MIPStar{} protocols were studied in a long line of work, culminating in a breakthrough result of Ito and Vidick \cite{ItoV12}, who in a technical tour-de-force showed that $\NEXP \subseteq \MIPS$;\footnote{%
While this is the popular statement of the result, \cite{ItoV12} show a stronger result, namely, that $\NEXP$ is exactly the class of languages decided by MIPs \emph{sound against entangled provers}. Their honest provers are classical, and soundness holds also against entangled provers. This is also the case in our protocols. It remains unknown whether entanglement grants provers additional power: there is \emph{no} known reasonable upper bound on $\MIPS$.
}
this result was further improved in \cite{Vidick16,NatarajanV17}. However, it is unknown whether these MIP protocols can achieve zero knowledge, which is the original motivation behind the classical MIP model. In sum, in this paper we pose the following question:

\begin{center}
\emph{To what extent does spatial isolation imply unconditional zero knowledge in a quantum world?}
\end{center}
\smallskip

\subsection{Our results}
\label{sec:our-results}

Our main result is a strong positive answer to the foregoing question, namely, we show that the $\NEXP \subseteq \MIPS$ result of Ito and Vidick \cite{ItoV12} continues to hold even when we require zero knowledge.

\begin{introtheorem}
\label{thm:main-zk}
Every language in $\NEXP$ has a perfect zero knowledge $2$-prover \MIPStar{}. In more detail,
\begin{equation*}
\NEXP \subseteq \PZKMIPSparams{2}{1/2}{\poly(n)}{\poly(n)} \enspace.
\end{equation*}
\end{introtheorem}

We stress that the \MIPStar{} protocols of \cref{thm:main-zk} enjoy both unconditional soundness against entangled provers as well as unconditional (perfect) zero knowledge against \emph{any} (possibly malicious) verifier.

\subsection{Other notions of quantum zero knowledge}
\label{sec:other-zk-notions}

To the best of our knowledge, this work is the first to study the notion of zero knowledge with entangled provers, as captured by the \MIPStar{} model. Nevertheless, zero knowledge has been studied in other settings in the quantum information and computation literature; we now briefly recall these.

Watrous \cite{Watrous02} introduced \emph{honest-verifier} zero knowledge for quantum interactive proofs (interactive proofs in which the prover and verifier are quantum machines), and studied the resulting complexity class $\QSZKHV$. Kobayashi \cite{Kobayashi03} studied a non-interactive variant of this notion. Damg{\aa}rd, Fehr, and Salvail \cite{DamgardFS04} achieve zero knowledge for $\NP$ against malicious quantum verifiers, but only via \emph{arguments} (i.e., computationally sound proofs) in the common reference string model. Subsequently, Watrous \cite{Watrous09} constructed quantum interactive proofs that remain zero knowledge against malicious quantum verifiers.

Zero knowledge for quantum interactive proofs has since then remained an active area of research, and several aspects and variants of it were studied in recent works, including the power of public-coin interaction \cite{Kobayashi08}, quantum proofs of knowledge \cite{Unruh12}, zero knowledge in the quantum random oracle model \cite{Unruh15}, zero knowledge proof systems for $\QMA$ \cite{BroadbentJSW16}, and oracle separations for quantum statistical zero knowledge \cite{MendaW18}.

All the above works consider protocols between a \emph{single} quantum prover and a quantum verifier. In particular, they do not study entanglement as a shared resource between two (or more) provers.

In contrast, the \MIPStar{} protocols that we study differ from the protocols above in two main aspects:
\begin{inparaenum}[(1)]
	\item our proof systems have multiple spatially-isolated provers that share an \qstate{}, and
	\item it suffices that the honest verifier is a \emph{classical} machine.
\end{inparaenum}
Indeed, we show that, analogously to the classical setting, \MIPStar{} protocols can achieve \emph{unconditional} zero knowledge for a much larger complexity class (namely, $\NEXP$) than possible for QSZK protocols (since $\QSZK \subseteq \QIP = \PSPACE$).


\clearpage
\section{Techniques}
\label{sec:techniques}

We begin by discussing the challenge that arises when trying to prove that $\NEXP \subseteq \PZKMIPS$, by outlining a natural approach to obtaining zero knowledge \MIPStar{} protocols, and considering why it fails.

\subsection{The challenge}

We know that every language in $\NEXP$ has a (perfect) zero knowledge MIP protocol, namely, that $\NEXP \subseteq \PZKMIP$ \cite{BenOrGKW88}. We also know that every language in $\NEXP$ has an \MIPStar{} protocol, namely, that $\NEXP \subseteq \MIPS$ \cite{ItoV12}. Is it then not possible to simply combine these two facts and deduce that every language in $\NEXP$ has a (perfect) zero knowledge \MIPStar{}? 

The challenge is that the standard techniques used to construct zero knowledge MIP protocols do not seem compatible with those used to construct \MIPStar{} protocols for large classes. In fact, the former are precisely the type of techniques that prove to be very limited for obtaining soundness against entangled provers.

In more detail, while constructions of MIP (and PCP) protocols typically capitalize on an \emph{algebraic} structure, known constructions of \emph{zero knowledge} MIPs are of a \emph{combinatorial} nature. For example, the zero knowledge MIP in \cite{BenOrGKW88} is based on a multi-prover information-theoretic commitment scheme, which can be thought of as a CHSH-like game. The zero knowledge MIP in \cite{DworkFKNS92} is obtained via the standard transformation from zero knowledge PCPs, which is a form of consistency game. Unfortunately, these types of constructions do not appear resistant to entangled provers, nor is it clear how one can modify them to obtain this resistance without leveraging some algebraic structure.

Indeed, initial attempts to show that $\NEXP \subseteq \MIPS$ (e.g., \cite{ItoKPSY08,ItoKM09,KempeKMTV11}) tried to apply some black box transformation to an arbitrarily structured (classical) MIP protocol to force the provers to behave as if they are not entangled, and then appeal to standard MIP soundness. These works were only able to obtain limited protocols (e.g., with very large soundness error).

In their breakthrough paper, Ito and Vidick \cite{ItoV12} overcame this hurdle and showed that $\NEXP \subseteq \MIPS$ by taking a different route: rather than a black box transformation, they modified and reanalyzed a particular proof system, namely the MIP protocol for $\NEXP$ in \cite{BabaiFL91}, while leveraging and crucially using its algebraic structure. (Subsequent works \cite{Vidick16,NatarajanV17} improved this result by reducing the number of provers and rounds to a minimum, showing \MIPStar{} protocols for $\NEXP$ with two provers and one round.)

In sum, the challenge lies in the apparent incompatibility between techniques used for zero knowledge and those used for soundness against entangled provers.

\subsection{High-level overview}

Our strategy for proving our main result is to bridge the aforementioned gap by isolating the role of algebra in granting soundness against entangled provers, and developing new algebraic techniques for zero knowledge. Our proof of \cref{thm:main-zk} thus consists of two parts.
\begin{enumerate}[label=(\Roman*)]

	\item \textbf{Lifting lemma:}
	a black box transformation from algebraically-structured classical protocols into corresponding \MIPStar{} protocols, which preserves zero knowledge.
	
	\item \textbf{Algebraic zero knowledge:}
	a new construction of zero knowledge algebraically-structured protocols for any language in $\NEXP$.

\end{enumerate}
The first part is primarily a conceptual contribution, and it deals with quantum aspects of proof systems. The second part is our main technical contribution, and it deals with classical protocols (it does not require any background in quantum information). We briefly discuss each of the parts, and then provide an overview of the first part in \cref{sec:techniques-part-I} and of the second part in \cref{sec:techniques-part-II}.

In the first part of the proof, we build on recent advances in low-degree testing against entangled provers, and provide an abstraction of techniques in \cite{ItoV12,Vidick16,NatarajanV17}. We prove a \emph{lifting lemma} (\cref{lem:lifting}) that transforms a class of algebraically-structured classical protocols into \MIPStar{} protocols, while preserving zero knowledge. This provides a generic framework for constructing \MIPStar{} protocols, while decoupling the mechanisms responsible for soundness against entangled provers from other classical components.

In the second part of the proof, we construct an algebraically-structured zero knowledge classical protocol, which we refer to as a \emph{low-degree interactive PCP}, to which we apply the lifting lemma, completing the proof. At the heart of our techniques is a strong zero knowledge variant of the sumcheck protocol \cite{LundFKN92} (a fundamental subroutine in many probabilistic proof systems), which we deem of independent interest. In turn, a key component in our zero knowledge sumcheck is a new algebraic commitment scheme, whose hiding property is guaranteed by algebraic query complexity lower bounds \cite{AaronsonW09,JumaKRS09}. These shed more light on the connection of zero knowledge to algebraic complexity theory.

We summarize the roadmap towards proving \cref{thm:main-zk} in \cref{sec:roadmap}.


\subsection{Part I: lifting classical proof systems to \MIPStar{}}
\label{sec:techniques-part-I}

The first step towards obtaining a generic framework for transforming classical protocols into corresponding \MIPStar{} protocols is making a simple, yet crucial, observation. Namely, while the result in \cite{ItoV12} is stated as a white box modification of the MIP protocol in \cite{BabaiFL91}, we observe that the techniques used there can in fact be applied more generally. That is, we observe that \emph{any} \DoQuote{low-degree interactive PCP}, a type of algebraically structured proof system that underlies (implicitly and explicitly) many constructions in the probabilistic proof systems literature, can be transformed into a corresponding \MIPStar{} protocol.

The first part of the proof of \cref{thm:main-zk} formalizes this idea, identifying sufficient conditions to apply the techniques of \cite{ItoV12,Vidick16}, and showing a lifting lemma that transforms protocols satisfying these conditions into \MIPStar{} protocols. We relate features of the original protocol to those of the resulting \MIPStar{} protocols, such as round complexity and, crucially, zero knowledge.

To make this discussion more accurate, we next define and discuss low-degree interactive PCPs.


\subsubsection{Low-degree interactive PCPs}

An \emph{Interactive PCP} (IPCP), a proof system whose systematic study was initiated by Kalai and Raz \cite{KalaiR08}, naturally extends the notions of a probabilistically checkable proof (PCP) and an interactive proof (IP). An $\RoundComplexity$-round IPCP is a two-phase protocol in which a computationally unbounded \emph{prover} $\Prover$ tries to convince a polynomial-time \emph{verifier} $\Verifier$ that an input $\Instance$, given to both parties, is in a language $\Language$. First, the prover sends to the verifier a PCP oracle (a purported proof that $x \in \Language$), which the verifier can query at any time. Second, the prover and verifier engage in an $\RoundComplexity$-round IP, at the end of which the verifier either accepts or rejects.\footnote{Alternatively, an IPCP can be viewed as a PCP that is verified \emph{interactively} (by an IP, instead of a randomized algorithm).} Completeness and soundness are defined in the usual way. 

In this work we consider a type of algebraically-structured IPCP, which we call a \emph{low-degree IPCPs}. This notion implicitly (and semi-explicitly) underlies many probabilistic proof systems in the literature. Informally, a low-degree IPCP is an IPCP satisfying the following:
\begin{inparaenum}[(1)]
  \item \emph{low-degree completeness}, which states that the PCP oracle sent by the (honest) prover is a polynomial of low (individual) degree;
  \item \emph{low-degree soundness}, which relaxes soundness to hold only against provers that send PCP oracles that are low-degree polynomials.
\end{inparaenum}

Low-degree completeness and soundness can be viewed as a promise that the PCP oracle is a low-degree polynomial. Indeed, these conditions are designed to capture \DoQuote{compatibility} with low-degree testing: only protocols with low-degree completeness will pass a low-degree test with probability $1$; moreover, adding a low-degree test to an IPCP with low-degree soundness results (roughly) in an IPCP with standard soundness.

\subsubsection{From low-degree IPCP to \MIPStar{}}

We show that any low-degree IPCP can be transformed into a corresponding \MIPStar{} protocol, in a way that preserves zero knowledge (for a sufficiently strong notion of zero knowledge IPCP). To this end, we use an entanglement-resistant low degree test, which allows us to essentially restrict the provers usage of the \qstate{} to strategies that can be approximately implemented via randomness shared among the provers. Informally, the idea is that by carefully invoking such a test, we can let one prover take on the role of the PCP oracle, and the other to take the role of the IPCP prover, and then emulate the entire IPCP protocol.

In more detail, we show a zero-knowledge-preserving transformation of low-degree IPCPs to \MIPStar{} protocols, which is captured by the following lifting lemma.

\begin{lemma}[informally stated, see \cref{lem:lifting}]
\label{lem:lifting-informal}
There exists a transformation $\Transformation$ that takes an $\RoundComplexity$-round low-degree IPCP $(\Prover',\Verifier')$ for a language $\Language$, and outputs a $2$-prover $(\RoundComplexity^* + 2)$-round \MIPStar{} $(\Prover_1, \Prover_2, \Verifier) \DefineEqual \Transformation(\Prover',\Verifier')$ for $\Language$, where $\RoundComplexity^* = \max\Set{\RoundComplexity,1}$.
Moreover, this transformation preserves zero knowledge.\footnote{More accurately, we require the given IPCP to be zero knowledge with query bound that is roughly quadratic in the degree of the PCP oracle. See \cref{sec:LDIPCP_to_MIPS} for details.}
\end{lemma}
We stress that the simplicity of the lifting lemma is a key feature since, as we describe below, it requires us to only make small structural changes to the IPCP protocol. This facilitates the preservation of various complexity measures and properties, such as zero knowledge.

To prove this lemma, a key tool that we use is a new low-degree test by Natarajan and Vidick \cite{NatarajanV17},\footnote{If we do not aim to obtain the optimal number of provers in our \MIPStar{} protocols, then it it suffices to use (an adaptation of) the low-degree test in \cite{Vidick16}.} which adapts the celebrated plane-vs-point test of Raz and Safra \cite{RazS97} to the \MIPStar{} model. A \emph{low-degree test} is a procedure used to determine if a given function $f \colon \Field^\SCVars \to \Field$ is close to a low-degree polynomial or if, instead, it is far from all low-degree polynomials, by examining $f$ at very few locations. In the plane-vs-point test, the verifier specifies a random $2$-dimensional plane in $\Field^\SCVars$ to one prover and a random point on this plane to the other prover; each prover replies with the purported value of $f$ on the received plane or point; then the verifier checks that these values are consistent.

Informally, the analysis in \cite{NatarajanV17} asserts that every entangled strategy that passes this test with high probably must satisfy an algebraic structure; more specifically, to pass this test the provers can only use their shared \qstate{} to (approximately) agree on a low-degree polynomial according to which they answer. We use the following soundness analysis of the this protocol. (See \cref{sec:quantum-information} for the standard quantum notation used in the theorem below.)

\begin{theorem}[{\cite[Theorem 2]{NatarajanV17}, informally stated}]
\label{thm:quantum_low_degree_test_informal}
There exists an absolute constant $c \in (0,1)$ such that, for every soundness parameter $\SoundnessError > 0$, number of variables $\SCVars \in \Naturals$, degree $\SCDegree \in \Naturals$, and finite field $\Field$, there exists a low-degree test $T$ for which the following holds. For every \emph{symmetric} entangled prover strategy and measurements $\Set{\Measurement{A}{\Point}{\Element}}_{\Element \in \Field,\Point \in \Field^\SCVars}$ that are accepted by $T$ with probability at least $1-\SoundnessError$, there exists a measurement $\Set{\Measurement{L}{}{\Poly}}_{\Poly}$, where $\Poly$ is an $m$-variate polynomial of degree $d$, such that:	
\begin{enumerate}[nolistsep]
	
  \item \emph{Approximate consistency with $\Set{\Measurement{A}{\Point}{\Element}}$:}
  $\E_{\Point \in \Field^\SCVars} \sum_{\Poly} \sum_{\Element \neq \Poly(\Point)} \bra{\Psi} \Measurement{A}{\Point}{\Element} \otimes \Measurement{L}{}{\Poly} \ket{\Psi} \leq \SoundnessError^c$.

  \item \emph{Self-consistency of $\Set{\Measurement{L}{}{\Poly}}$:}
  $\sum_{\Poly} \bra{\Psi} \Measurement{L}{}{\Poly} \otimes (\Id -\Measurement{L}{}{\Poly}) \ket{\Psi} \leq \SoundnessError^c$.

\end{enumerate}
\end{theorem}

In fact, we actually use a more refined version, which tests a polynomial's \emph{individual} degree rather than its \emph{total} degree. In the classical setting, such a test is implicit in \cite{GoldreichS06} via a reduction from individual-degree to total-degree testing. Informally, this reduction first invokes the test for low total degree, then performs univariate low-degree testing with respect to a random axis-parallel line in each axis. We extend this reduction and its analysis to the setting of \MIPStar{}. (See \cref{sec:quantum-individual} for details.) The analysis of the low individual degree test was communicated to us by Thomas Vidick, to whom we are grateful for allowing us to include it here.

With the foregoing low-degree test at our disposal, we are ready to outline the simple transformation from low-degree IPCPs to \MIPStar{} protocols. We begin with a preprocessing step. Note that the low individual degree test provides us with means to assert that the provers can (approximately) only use their \qstate{} to choose a low-degree polynomial $\Poly$, and answer the verifier with the evaluation of $\Poly$ on a \emph{single}, uniformly distributed point (or plane). Thus, it is important that the IPCP verifier (which we start from) only makes a single uniform query to its oracle. By adapting techniques from \cite{KalaiR08}, we can leverage the algebraic structure of the low-degree IPCP and capitalize on the interaction to ensure the IPCP verifier has this property, at essentially the cost of increasing the round complexity by $1$.\footnote{Indeed, if the original IPCP verifier makes a single uniform query to its oracle, then we can save a round in \cref{lem:lifting-informal}; that is, we obtain an \MIPStar{} with round complexity $\RoundComplexity^*+1$, rather than $\RoundComplexity^*+2$.}

Thus we have a low-degree IPCP, with prover $\Prover$ and verifier $\Verifier$, in which the verification takes place as follows. Both $\Prover$ and $\Verifier$ receive an explicit input $x$ that is allegedly in the language $\Language$. In addition, $\Verifier$ is granted oracle access to a purported low-degree polynomial $\Oracle$, whose full description is known to $\Prover$. The parties engage in an $r$-round interaction, at the end of which $\Verifier$ is allowed to make a single uniform query to $\Oracle$ and decide whether $x \in \Language$ (with high probability).

We transform this IPCP into a $2$-prover \MIPStar{} by considering the following protocol.
First, the verifier chooses uniformly at random whether to (1) invoke a low-degree test, in which it asks one prover to evaluate $\Oracle$ on a random plane or axis-parallel line and the other prover to evaluate $\Oracle$ on a random point on this plane or line, or (2) emulate the IPCP protocol, in which one prover plays the role of the IPCP prover and the other acts as lookup for $\Oracle$.

We use the approximate consistency condition of \cref{thm:quantum_low_degree_test_informal} to assert that the lookup prover approximately answers according to a low-degree polynomial, and use the self-consistency condition to ensure that both provers are consistently answering according to the \emph{same} low-degree polynomial.\footnote{Since the players are allowed the use of entanglement, we cannot hope for a single function that underlies their strategy. Indeed, the players could measure their \qstate{} to obtain shared randomness and select a random $\Oracle$ according to which they answer.}

We remark that preserving zero knowledge introduces some subtle technicalities (which we resolve), the main of which is that because the analysis of the entanglement-resistant low individual degree test requires that the provers employ \emph{symmetric} strategies, we need to perform a non-standard symmetrization (since standard symmetrization turns out to break zero knowledge in our case). See \cref{sec:LDIPCP_to_MIPS} for details.

\subsubsection{Towards zero knowledge \MIPStar{} for nondeterministic exponential time}
\label{sec:techniques-towards}

Equipped with the lifting lemma, we are left with the task of constructing classical zero knowledge low-degree IPCPs for all languages in $\NEXP$. We first explain why current constructions do \emph{not} suffice for this purpose.

The first thing to observe is that the classical protocol for the $\NEXP$-complete language \emph{Oracle 3SAT} by Babai, Fortnow, and Lund \cite{BabaiFL91} (neglecting the multilinearity test) can be viewed as low-degree IPCP. Indeed, in \cite{BabaiFL91} the protocol is stated as an \DoQuote{oracle protocol}, which is equivalent to an IPCP. The oracle is encoded as a low-degree polynomial, and so low-degree completeness is satisfied. Alas, the foregoing protocol is \emph{not} zero knowledge.
We remark that since the \MIPStar{} protocol in \cite{ItoV12} relies on the protocol in \cite{BabaiFL91}, the former inherits the lack of zero knowledge from the latter.

Proceeding to consider classical \emph{zero knowledge} proof systems, for example the protocols in \cite{DworkFKNS92,KilianPT97,GoyalIMS10}, we observe that while some of these proof systems can be viewed as IPCPs, they are not \emph{low-degree} IPCPs. This is because they achieve zero knowledge via combinatorial techniques that do \emph{not} admit the algebraic structure that we require. We stress that the natural way of endowing an IPCP with algebraic structure by taking the low-degree extension of the PCP oracle does \emph{not} necessarily preserve zero knowledge.\footnote{Intuitively, a single point in the encoded oracle can summarize a large amount of information from the original oracle (e.g., very large linear combinations).} Correspondingly, the \MIPStar{} protocols in \cite{Vidick16,NatarajanV17}, which rely on applying the low-degree extension code to a PCP, do not preserve zero knowledge for this reason.

Finally, we observe that recent advances in algebraic zero knowledge \cite{BenSassonCFGRS17} (building on techniques from \cite{BenSassonCGV16}) already provide us with a classical proof system that is compatible with our framework, and can thus be used to derive a zero knowledge \MIPStar{} protocol, albeit only for languages in $\sharpP$. 

To strengthen the aforementioned result and show that $\NEXP \subseteq \PZKMIPS$ (matching the $\NEXP \subseteq \MIPS$ containment, and showing that zero knowledge can, in a sense, be obtained for \DoQuote{free} in the setting of \MIPStar{} protocols), we need to construct a \emph{much stronger} zero knowledge low-degree IPCP. The second part of \cref{thm:main-zk}, which is our main technical contribution, provides exactly that. We proceed to provide an overview of the techniques that we use to construct such protocols.

\subsection{Part II: new algebraic techniques for zero knowledge}
\label{sec:techniques-part-II}

The techniques discussed thus far tell us that, if we wish to obtain a zero knowledge \MIPStar{} for $\NEXP$, it suffices to obtain a zero knowledge \emph{low-degree} IPCP for $\NEXP$ (an IPCP wherein the oracle is a low-degree polynomial). Doing so is the second part of our proof of \cref{thm:main-zk}, and for this we develop new algebraic techniques for obtaining zero knowledge protocols. Our techniques, which build on recent developments \cite{BenSassonCGV16,BenSassonCFGRS17}, stand in stark contrast to other known constructions of zero knowledge PCPs and interactive PCPs (such as \cite{DworkFKNS92,KilianPT97,GoyalIMS10}). We remind the reader that this part of our work only deals with classical protocols, and does not require any knowledge of quantum information.

\subsubsection{A zero knowledge low-degree IPCP for $\NEXP$}
\label{sec:techniques-nexp}

Our starting point is the protocol of Babai, Fortnow, and Lund \cite{BabaiFL91} (the \DoQuote{BFL protocol}). We first recall how the BFL protocol works, in order to explain its sources of information leakage and how one could prevent them via algebraic techniques. These are the ideas that underlie our algebraic construction of an unconditional (perfect) zero knowledge low-degree IPCP for $\NEXP$.

\parhead{The BFL protocol, and why it leaks}
\emph{Oracle 3SAT} ($\mathrm{O3SAT}$) is the following $\NEXP$-complete problem: given a boolean formula $B$, does there exist a boolean function $A$ (a witness) such that 
\begin{equation*}
	B(z, b_{1}, b_{2}, b_{3}, A(b_{1}), A(b_{2}), A(b_{3}))=0 \quad \text{for all } z \in \Bits^{r}, b_{1}, b_{2}, b_{3} \in \Bits^{s}
	\;\;\text{?}
\end{equation*}
The BFL protocol is an IPCP for $\mathrm{O3SAT}$ that is then (generically) converted to an MIP. In the BFL protocol, the honest prover first sends a PCP oracle $\LD{A} \colon \Field^{s} \to \Field$ that is the unique multilinear extension (in some finite field $\Field$) of a valid witness $A \colon \Bits^{s} \to \Bits$. The verifier must check that
\begin{inparaenum}[(a)]
	\item $\LD{A}$ is a boolean function on $\Bits^{s}$, and
	\item $\LD{A}$'s restriction to $\Bits^{s}$ is a valid witness for $B$.
\end{inparaenum}
To do these checks, the verifier arithmetizes the formula $B$ into an arithmetic circuit $\LD{B}$, and reduces the checks to conditions that involve $\LD{A}$, $\LD{B}$, and other low-degree polynomials. A technique in \cite{BabaiFLS91} allows the verifier to \DoQuote{bundle} all of these conditions into a single low-degree polynomial $f$ such that (with high probability over the choice of $f$) the conditions hold if and only if $f$ sums to $0$ on $\Bits^{r+3s+3}$. The verifier checks that this is the case by engaging in a sumcheck protocol with the prover.\footnote{The soundness of the sumcheck protocol depends on the PCP oracle being the evaluation of a low-degree polynomial, and so the verifier in \cite{BabaiFL91} checks this using a low-degree test. In our setting of \emph{low-degree} IPCPs a low-degree test is not necessary.}

We observe that the BFL protocol is \emph{not} zero knowledge for two reasons:
\begin{inparaenum}[(i)]
\item the verifier has oracle access to $\LD{A}$ and, in particular, to the witness $A$;
\item the prover's messages during the sumcheck protocol leak further information about $A$ (namely, hard-to-compute partial sums of $f$, which itself depends on $A$).
\end{inparaenum}

\parhead{A blueprint for zero knowledge}
We now describe the \DoQuote{blueprint} for an approach to achieve zero knowledge in the BFL protocol. The prover does not send $\LD{A}$ directly, but instead a \emph{commitment} to it. After this, the prover and verifier engage in a sumcheck protocol with suitable zero knowledge guarantees; at the end of this protocol, the verifier needs to evaluate $f$ at a point of its choice, which involves evaluating $\LD{A}$ at three points. Now the prover reveals the requested values of $\LD{A}$, without leaking any information beyond these, so that the verifier can perform its check. We explain how these ideas motivate the need for certain algebraic tools, which we later develop and use to instantiate our approach.

\parhead{(1) Randomized low-degree extension}
Even if the prover reveals only three values of $\LD{A}$, these may still leak information about $A$. We address this problem via a \emph{randomized low-degree extension}. Indeed, while the prover in the BFL protocol sends the \emph{unique} multilinear extension of $A$, one can verify that \emph{any} extension of $A$ of sufficiently low degree also works. We exploit this flexibility as follows: the prover randomly samples $\LD{A}$ in such a way that any three evaluations of $\LD{A}$ do not reveal any information about $A$. Of course, if any of these evaluations is within the systematic part $\Bits^{s}$, then no extension of $A$ has this property. Nevertheless, during the sumcheck protocol, the prover can ensure that the verifier chooses only evaluations outside of $\Bits^{s}$ (by aborting if the verifier deviates), which incurs only a small increase in the soundness error.\footnote{The honest verifier will be defined so that it always chooses evaluations \emph{outside} of $\Bits^{s}$, so completeness is unaffected.} With this modification in place, it suffices for the prover to let $\LD{A}$ be a random degree-$4$ extension of $A$: by a dimensionality argument, any $3$ evaluations outside of $\Bits^{s}$ are now independent and uniformly random in $\Field$. We are thus able to reduce a claim about $A$ to a claim which contains \emph{no information} about $A$.

\parhead{(2) Algebraic commitments}
As is typical in zero knowledge protocols, the prover will send a \emph{commitment} to $\LD{A}$, and then selectively reveal a limited set of evaluations of $\LD{A}$. The challenge in our setting is that this commitment must \emph{also} be a low-degree polynomial, since we require a low-degree oracle. For this, we devise a new algebraic commitment scheme based on the sumcheck protocol; we discuss this in \cref{sec:techniques-algebraic-commitment}.

\parhead{(3) Sumcheck in zero knowledge}
We need a sumcheck protocol where the prover's messages leak little information about $f$. The prior work in \cite{BenSassonCFGRS17} achieves an IPCP for sumcheck that is \DoQuote{weakly} zero knowledge: any verifier learns at most one evaluation of $f$ for each query it makes to the PCP oracle. If the verifier could evaluate $f$ by itself, as was the case in that paper, this guarantee would suffice for zero knowledge. In our setting, however, the verifier \emph{cannot} evaluate $f$ by itself because $f$ is (necessarily) hidden behind the algebraic commitment.

One approach to compensate would be to further randomize $\LD{A}$ by letting $\LD{A}$ be a random extension of $A$ of some well-chosen degree $d$. Unfortunately, this technique is incompatible with our low-degree IPCP to \MIPStar{} transformation: such a low-degree extension is at most $d$-wise independent, whereas our lifting lemma (\cref{lem:lifting}), and more generally low-degree testing, requires zero knowledge against any $\Omega(d^{2})$ queries.



We resolve this by relying on more algebraic techniques, achieving an IPCP for sumcheck with a much stronger zero knowledge guarantee: any malicious verifier that makes polynomially-many queries to the PCP oracle learns only a \emph{single} evaluation of $f$. This suffices for zero knowledge in our setting: learning one evaluation of $f$ implies learning only three evaluations of $\LD{A}$, which can be made \DoQuote{safe} if $\LD{A}$ is chosen to be a random extension of $A$ of sufficiently high degree. Our sumcheck protocol uses as building blocks both our algebraic commitment scheme and the \DoQuote{weak} zero knowledge sumcheck in \cite{BenSassonCFGRS17}; we summarize its construction in \cref{sec:techniques-strong-zksc}.

\subsubsection{Algebraic commitments from algebraic query complexity lower bounds}
\label{sec:techniques-algebraic-commitment}

We provide a high-level description of an information-theoretic commitment scheme in the low-degree IPCP model (i.e., a low-degree \emph{interactive locking scheme} \cite{GoyalIMS10}). See \cref{sec:algebraic-query-complexity} for the full details.

In this scheme, the prover commits to a message by sending to the verifier a PCP oracle that perfectly hides the message; subsequently, the prover can reveal positions of the message by engaging with the verifier in an interactive proof, whose soundness guarantees statistical binding. 

\parhead{Committing to an element}
We first consider the simple case of committing to a single element $a$ in $\Field$. Let $k$ be a security parameter, and set $N \DefineEqual 2^{k}$. Suppose that the prover samples a random $B$ in $\Field^{N}$ such that $\sum_{i=1}^{N} B_{i} = a$, and sends $B$ to the verifier as a commitment. Observe that any $N-1$ entries of $B$ do not reveal any information about $a$, and so any verifier with oracle access to $B$ that makes fewer than $N$ queries cannot learn any information about $a$. However, as $B$ is unstructured it is not clear how the prover can later convince the verifier that $\sum_{i=1}^{N} B_{i} = a$.

Instead, we can consider imbuing $B$ with additional algebraic structure. Namely, the prover views $B$ as a function from $\Bits^{k}$ to $\Field$, and sends its unique multilinear extension $\LD{B} \colon \Field^{k} \to \Field$ to the verifier. Subsequently, the prover can reveal $a$ to the verifier, and then engage in a sumcheck protocol for the claim \DoQuote{$\sum_{\vec{\beta} \in \Bits^{k}} \LD{B}(\vec{\beta}) = a$} to establish the correctness of $a$. The soundness of the sumcheck protocol protects the verifier against cheating provers and hence guarantees that this scheme is binding.

However, giving $B$ additional structure calls into question the hiding property of the scheme. Indeed, surprisingly, a result of Juma et al.~\cite{JumaKRS09} shows that this new scheme is in fact \emph{not} hiding (in fields of odd characteristic): it holds that $\LD{B}(2^{-1}, \ldots, 2^{-1}) = a \cdot 2^{-k}$ for any choice of $B$, so the verifier can learn $a$ with only a single query to $\LD{B}$!

Sending an extension of $B$ has created a new problem: querying the extension outside of $\Bits^{k}$, the verifier can learn information that may require many queries to $B$ to compute. Indeed, this additional power is precisely what underlies the soundness of the sumcheck protocol. To resolve this, we need to understand what the verifier can learn about $B$ given some low-degree extension $\LD{B}$. This is precisely the setting of \emph{algebraic query complexity} \cite{AaronsonW09}.\footnote{Interestingly, in \cite{AaronsonW09} a connection between algebra and zero knowledge is also exhibited. Namely, to show that the result $\NP \subseteq \FormatComplexityClass{CZK}$ \cite{GoldreichMW91} algebrizes, it is necessary to exploit the algebraic structure of the oracle to design a zero knowledge protocol for verifying the existence of certain sets of query answers.}

Indeed the foregoing theory suggests a natural approach for overcoming the problem created by the extension of $B$: instead of considering the multilinear extension, we can let $\LD{B}$ be chosen uniformly at random from the set of degree-$d$ extensions of $B$, for some $d > 1$. It is not hard to see that if $d$ is very large (say, $\SetCardinality{\Field}$) then $2^{k}$ queries are required to determine the summation of $\LD{B}$ on $\{0,1\}^{k}$. However, we need $d$ to be small to achieve soundness. Fortunately, a result of \cite{JumaKRS09} shows that $d = 2$ suffices: given a random multiquadratic extension $\LD{B}$ of $B$, one needs $2^{k}$ queries to $\LD{B}$ to determine $\sum_{\vec{\beta} \in \Bits^{k}} \LD{B}(\vec{\beta})$.\footnote{This is the main reason why our application to constructing \MIPStar{} protocols requires low-degree test against entangled provers, rather than just a multilinearity test, as was used in \cite{ItoV12}.}

\parhead{Committing to a polynomial}
The prover in our zero knowledge protocols needs to commit not just to a single element but rather to the evaluation of an $m$-variate polynomial $Q$ over $\Field$ of degree $d>1$. We extend our ideas to this setting. We follow a similar general approach, however, arguing the hiding property now requires a \emph{stronger} algebraic query complexity lower bound than the one proved in \cite{JumaKRS09}. Not only do we need to know that the verifier cannot determine $Q(\vec{\alpha})$ for a particular $\vec{\alpha} \in \Field^{m}$, but we need to know that the verifier cannot determine $Q(\vec{\alpha})$ for \emph{any} $\vec{\alpha} \in \Field^{m}$, or even \emph{any linear combination of any such values}. We prove that this stronger guarantee holds in the same parameter regime: if $d > 1$ then $2^{k}$ queries are both necessary and sufficient. See the discussion at the beginning of \cref{sec:algebraic-query-complexity} for a more detailed overview.

\parhead{Decommitting in zero knowledge}
To use our commitment scheme in zero knowledge protocols, we must ensure that, in the decommitment phase, the verifier cannot learn any information beyond the value $a \DefineEqual Q(\vec{\alpha})$, for a chosen $\vec{\alpha}$. To decommit, the prover sends the value $a$ and has to convince the verifier that the claim \DoQuote{$\sum_{\vec{\beta} \in \Bits^{k}} \LD{B}(\vec{\alpha}, \vec{\beta}) = a$} is true. However, if the prover and verifier simply run the sumcheck protocol on this claim, the prover leaks partial sums $\sum_{\vec{\beta} \in \Bits^{k-i}} \LD{B}(\vec{\alpha}, c_{1}, \ldots, c_{i}, \vec{\beta})$, for $c_{1}, \ldots, c_{i} \in \Field$ chosen by the verifier, which could reveal additional information about $Q$. Instead, the prover and verifier run on this claim the IPCP for sumcheck of \cite{BenSassonCFGRS17}, whose \DoQuote{weak} zero knowledge guarantee ensures that this cannot happen. (Thus, in addition to the commitment, the honest prover also sends the evaluation of a random low-degree polynomial as required by the IPCP for sumcheck of \cite{BenSassonCFGRS17}.)

\subsubsection{A zero knowledge sumcheck protocol}
\label{sec:techniques-strong-zksc}

We describe the \DoQuote{strong} zero knowledge variant of the sumcheck protocol that we use in our construction. The protocol relies on the algebraic commitment scheme described in the previous section. We first cover some necessary background, and then describe our protocol.

\parhead{Previous sumcheck protocols}
The sumcheck protocol \cite{LundFKN92} is an IP for claims of the form \DoQuote{$\sum_{\vec{\alpha} \in H^{m}} F(\vec{\alpha}) = 0$}, where $H$ is a subset of a finite field $\Field$ and $F$ is an $m$-variate polynomial over $\Field$ of small individual degree. The protocol has $m$ rounds: in round $i$, the prover sends the univariate polynomial $g_{i}(\VariableX_{i}) \DefineEqual \sum_{\vec{\alpha} \in H^{m-i}} F(c_{1}, \ldots, c_{i-1}, \VariableX_{i}, \vec{\alpha})$, where $c_{1}, \ldots, c_{i-1} \in \Field$ were sent by the verifier in previous rounds; the verifier checks that $\sum_{\alpha_{i} \in H} g_{i}(\alpha_{i}) = g_{i-1}(c_{i-1})$ and replies with a uniformly random challenge $c_{i} \in \Field$. After round $m$, the verifier outputs the claim \DoQuote{$F(c_{1}, \dots, c_{m}) = g_{m}(c_{1}, \dots, c_{m})$}. If $F$ is of sufficiently low degree and does not sum to $\SCSum$ over the space, then the output claim is false with high probability. Note that the verifier does not need access to $F$.

The \DoQuote{weak} zero knowledge IPCP for sumcheck in \cite{BenSassonCFGRS17} modifies the above protocol as follows. The prover first sends a PCP oracle that (allegedly) equals the evaluation of a random \DoQuote{masking} polynomial $R$; the verifier checks that $R$ is (close to) low degree. Subsequently, the prover and verifier conduct the following interactive proof. The prover sends $z \in \Field$ that allegedly equals $\sum_{\vec{\alpha} \in H^{m}} R(\vec{\alpha})$, and the verifier responds with a uniformly random challenge $\rho \in \Field^{*}$. The prover and verifier now run the (standard) sumcheck protocol to reduce the claim \DoQuote{$\sum_{\vec{\alpha} \in H^{m}} \rho F(\vec{\alpha}) + R(\vec{\alpha}) = \rho \SCSum + z$} to a claim \DoQuote{$\rho F(\vec{c}) + R(\vec{c})=b$}, for a random $\vec{c} \in \Field^{m}$. The verifier queries $R$ at $\vec{c}$ and then outputs the claim \DoQuote{$F(\vec{c})=\frac{b-R(\vec{c})}{\rho}$}. If $\sum_{\vec{\alpha} \in H^{m}} F(\vec{\alpha}) \neq a$, then with high probability over $\rho$ and the verifier's messages in the sumcheck protocol, this claim is false.

A key observation is that if the verifier makes no queries to $R$, then the prover's messages are identically distributed to the sumcheck protocol applied to a random polynomial $Q$. When the verifier does make queries to $R$, simulating the resulting conditional distribution involves techniques from Algebraic Complexity Theory, as shown in \cite{BenSassonCFGRS17}. Given $Q$, the verifier's queries to $R(\vec{\alpha})$, for $\vec{\alpha} \in \Field^{m}$, are identically distributed to $Q(\vec{\alpha}) - \rho F(\vec{\alpha})$. Thus, the simulator need only make at most one query to $F$ for every query to $R$; that is, any verifier making $q$ queries to $R$ learns no more than it would learn by making $q$ queries to $F$ alone.

As discussed, this zero knowledge guarantee does not suffice for the application that we consider: in the $\NEXP$ protocol, the polynomial $F$ is defined in terms of the $\NEXP$ witness. In this case the verifier can learn enough about $F$ to break zero knowledge by making only $O(\deg{F})$ queries to $R$.

\parhead{Our sumcheck protocol}
The \DoQuote{strong} zero knowledge guarantee that we aim for is the following: any polynomial-time verifier learns no more than it would by making \emph{one} query to $F$, regardless of its number of queries to the PCP oracle.

The main idea to achieve this guarantee is the following. The prover sends a PCP oracle that is an \emph{algebraic commitment} $Z$ to the aforementioned masking polynomial $R$. Then, as before, the prover and verifier run the sumcheck protocol to reduce the claim \DoQuote{$\sum_{\vec{\alpha} \in H^{m}} \rho F(\vec{\alpha}) + R(\vec{\alpha}) = \rho a + z$} to a claim \DoQuote{$\rho F(\vec{c}) + R(\vec{c})=b$} for random $\vec{c} \in \Field^{m}$.

We now face two problems. First, the verifier cannot simply query $R$ at $\vec{c}$ and then output the claim \DoQuote{$F(\vec{c})=\frac{b-R(\vec{c})}{\rho}$}, since the verifier only has oracle access to the commitment $Z$ of $R$. Second, the prover could cheat the verifier by having $Z$ be a commitment to an $R$ that is far from low degree, which allows cheating in the sumcheck protocol.

The first problem is addressed by the fact that our algebraic commitment scheme has a decommitment sub-protocol that is zero knowledge: the prover can reveal $R(\vec{c})$ in such a way that no other values about $R$ are also revealed as a side-effect. As discussed, this relies on the protocol of \cite{BenSassonCFGRS17}, used as a subroutine.

The second problem is addressed by the fact that our algebraic commitment scheme is \DoQuote{transparent} to low-degree structure; that is, the algebraic structure of the scheme implies that if the commitment $Z$ is a low-degree polynomial (as in a \emph{low-degree} IPCPs), then $R$ must also be low degree (and vice versa).

Overall, the only value that a malicious verifier can learn is $F(\vec{c})$, for $\vec{c} \in I^{m}$ of its choice (where $I$ is some sufficiently large subset of $\Field$, fixed in advance). More precisely, we prove the following theorem, which shows a strong zero knowledge sumcheck protocol.

\begin{theorem}[Informally stated, see \cref{lem:strong-zk-sumcheck-short}]
\label{lem:strong-zk-sumcheck-short-informal}
There exists a low-degree IPCP for sumcheck, with respect to a low-degree polynomial $\SCPoly$, that satisfies the following zero knowledge guarantee: the view of any probabilistic polynomial-time verifier in the protocol can be perfectly and efficiently simulated by a simulator that makes only a single query to $\SCPoly$.
\end{theorem}

Our sumcheck protocol leaks a single evaluation of $\SCPoly$. We stress that this limitation is \emph{inherent}: the honest verifier always outputs a true claim about one evaluation of $\SCPoly$, which it cannot do without learning that evaluation. Nevertheless, this guarantee is strong enough for our application, as we can ensure that learning a single evaluation of $\SCPoly$ does not harm zero knowledge.

We remark that our strong zero knowledge sumcheck protocol can be transformed into a standard IPCP, by the standard technique of adding a (classical) low-degree test to the protocol.

\clearpage
\section{Discussion and open problems}
\label{sec:open-problems}

The framework that we use to prove that $\NEXP \subseteq \PZKMIPS$ elucidates the role that algebra plays in the design of proofs systems with entangled provers. Namely, we show that a large class of algebraic protocols (low-degree IPCPs) can be transformed in a black box manner to MIPs with entangled provers. This abstraction decouples the mechanisms responsible for soundness against entangled adversaries from other classical components in the proof system. In turn, this allows us to focus our attention on designing proof systems with desirable properties (zero knowledge, in this work), without having to deal with the complications that arise from entanglement, and then derive \MIPStar{} protocols from these classical protocols.

These ideas also enable us to re-interpret prior constructions of \MIPStar{} protocols at a higher level of abstraction. For example, the protocol in \cite{ItoV12} can be viewed as applying our lifting lemma to the (low-degree) IPCP in \cite{BabaiFL91}. As another example, one can start with \emph{any} PCP for some language $\Language$, low-degree extend the PCP, and then apply our lifting lemma to obtain a corresponding \MIPStar{} protocol for $\Language$; in fact, the protocol in \cite{Vidick16} can be viewed in this perspective.

In more detail, we say that a transformation from IPCP to \MIPStar{} is \emph{black box} if it maps an IPCP protocol into an \MIPStar{} protocol whose verifier can be expressed as an algorithm that only accesses the queries and messages of the IPCP verifier, but does not access its input (apart from its length). The following corollary shows that any IPCP protocol can be transformed into an \MIPStar{} protocol via a black box transformation. While a proof of this fact is implicit in \cite{Vidick16,NatarajanV17}, the framework developed in this paper allows us to crystallize its structure and give a compellingly short proof of it. (See, also, \cref{fig:mip-diag}.)

\begin{corollary}
\label{cor:blackbox}
There is a black box transformation that maps any $\RoundComplexity$-round IPCP protocol for a language $\Language$ to a $2$-prover $(\RoundComplexity+1)$-round \MIPStar{} for $\Language$
\end{corollary}

The round complexity of $\RoundComplexity+1$ in \cref{cor:blackbox} is less than in our lifting lemma ($\RoundComplexity+2$), because now we do not require that zero knowledge is preserved. We make the foregoing discussion precise in \cref{sec:blackbox}. 

We conclude this section by discussing several open problems.

In this work we show that there exist perfect zero knowledge \MIPStar{} protocols for all languages in $\NEXP$, with \emph{polynomially-many} rounds. Since round complexity is a crucial resource in any interactive proof system,
it is essential to understand whether zero knowledge \MIPStar{} protocols with low round complexity exist. (After all, without the requirement of zero knowledge, every language in $\NEXP$ has a \MIPStar{} protocol with just \emph{one round} \cite{Vidick16,NatarajanV17}.) We remark that the \DoQuote{oracularization} technique of Ito et al.\ \cite{ItoKM09} reduces the round complexity of any \MIPStar{} to one round, but this technique does \emph{not} preserve zero knowledge.

\begin{openproblem}
Do there exist constant-round zero knowledge \MIPStar{} protocols for $\NEXP$?
\end{openproblem}


At the beginning of this section, we reflected on the fact that known results that establish the power of \MIPStar{} protocols rely on \emph{algebraic} structure, which enables classical-to-quantum black box transformations of protocols. But is algebraic structure inherently required, or does some combinatorial structure suffice?

\begin{openproblem}
Is there a richer class of classical protocols (beyond low-degree IPCPs) that can be black-box transformed into \MIPStar{} protocols?
\end{openproblem}

For instance, could we replace low-degree polynomials with, say, error correcting codes with suitable local testability and decodability properties? One place to start would be to understand whether local testers for tensor product codes \cite{BenSassonS06} are sound against entangled provers.

\begin{openproblem}
When suitably adapted to the multi-prover setting, is the random hyperplane test in \cite{BenSassonS06} for tensor product codes sound against entangled provers?
\end{openproblem}

\clearpage
\section{Roadmap}
\label{sec:roadmap}

In \cref{sec:prelims} we provide definitions needed for the technical sections, including that for a \emph{low-degree IPCP}, which is central to our work.
In \cref{part:I} we prove that any low-degree IPCP can be transformed into an \MIPStar{} protocol, while preserving zero knowledge; see \cref{lem:lifting}.
In \cref{part:II} we prove that every language in $\NEXP$ has a perfect zero knowledge low-degree IPCP; see \cref{thm:pzk-for-nexp}.
Combining the results proved in \cref{part:I,part:II} enables us to derive our main result, \cref{thm:main-zk}, which shows that every language in $\NEXP$ has a perfect zero knowledge $2$-prover \MIPStar{}.
\cref{fig:mip-diag} summarizes the roadmap towards proving \cref{thm:main-zk}.

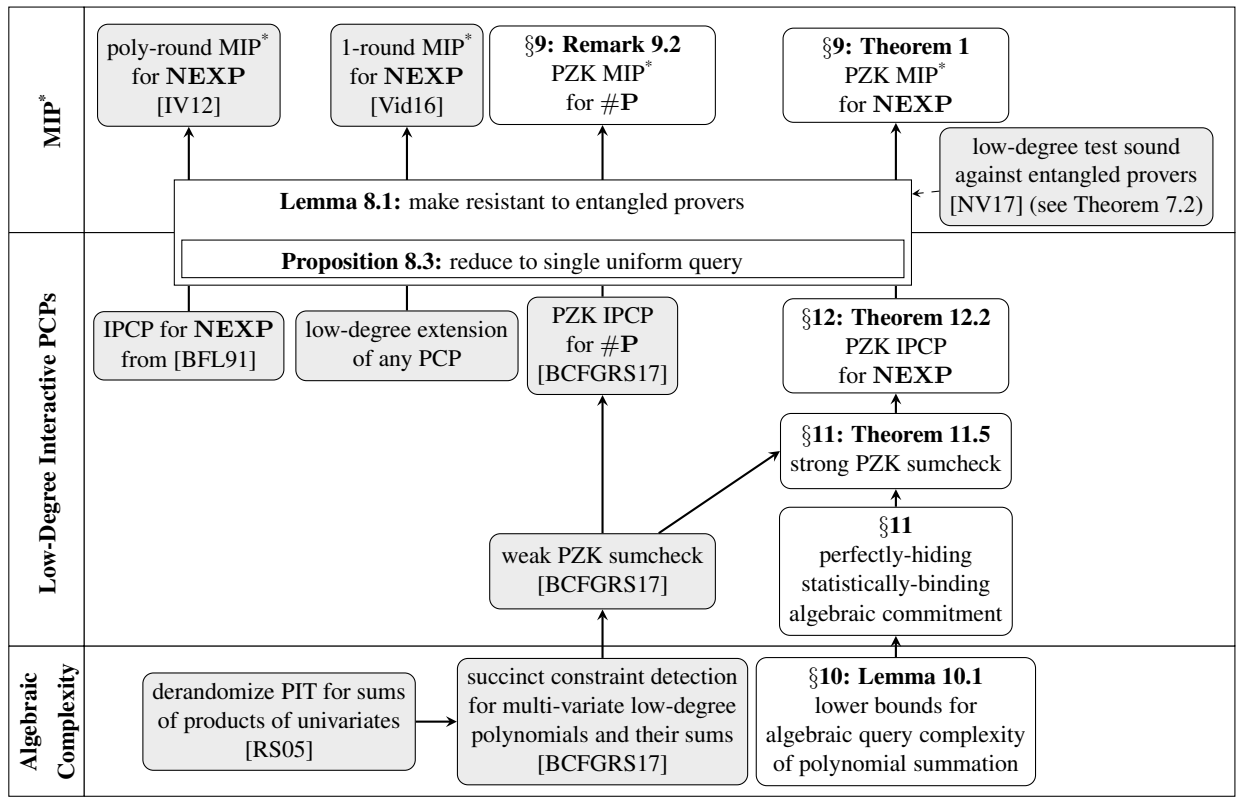
\begin{figure}[h!]
\begin{center}
\newcommand{\ResultHeader}[1]{\textbf{#1}}
\tikzstyle{headbox} = [rectangle, align=center, draw=black, fill=white, font=\footnotesize, anchor=south west]
\tikzstyle{backbox} = [rectangle, minimum width=15.3cm, align=center, draw=black, fill=white, font=\footnotesize, anchor=south west]
\tikzstyle{vtext} = [align=center, font=\footnotesize, rotate=90]
\tikzstyle{htext} = [align=center, font=\footnotesize]
\tikzstyle{newresult} = [rectangle, rounded corners, minimum width=3cm, minimum height=1cm, align=center, draw=black, fill=white, font=\footnotesize]
\tikzstyle{oldresult} = [rectangle, rounded corners, minimum width=3cm, minimum height=1cm, align=center, draw=black, fill=gray!15, font=\footnotesize]
\tikzstyle{smallresult} = [rectangle, rounded corners, align=center, draw=black, fill=gray!15, font=\footnotesize]
\tikzstyle{arrow} = [thick,->,>=stealth]
\tikzstyle{da} = [dashed,->,>=stealth]
\begin{tikzpicture}[node distance=1.5cm]
\node (backbox1) [backbox,  minimum height=2cm] at (-2.1,-1) {};
\node (headbox1) [headbox,  minimum height=2cm, minimum width=1cm] at (-3.1,-1) {};
\node (backbox2) [backbox,  minimum height=5.5cm] at (-2.1,1) {};
\node (headbox2) [headbox,  minimum height=5.5cm, minimum width=1cm] at (-3.1,1) {};
\node (backbox3) [backbox,  minimum height=3.0cm] at (-2.1,6.5) {};
\node (headbox3) [headbox,  minimum height=3.0cm, minimum width=1cm] at (-3.1,6.5) {};
\node (text1) [vtext] at (-2.55,3.7) {\textbf{Low-Degree Interactive PCPs}};
\node (text4) [vtext] at (-2.8,0) {\textbf{Algebraic}};
\node (text5) [vtext] at (-2.3,0) {\textbf{Complexity}};
\node (text6) [vtext] at (-2.55,8.0) {\textbf{\MIPStar{}}};
\node (raz) [oldresult] at (0.5,0) {derandomize PIT for sums \\ of products of univariates \\ \cite{RazS05}};
\node (scd) [oldresult] at (4.8,0) {succinct constraint detection \\ for multi-variate low-degree \\ polynomials and their sums \\ \cite{BenSassonCFGRS17}};
\node (aqc) [newresult] at (8.7,0) {\ResultHeader{\S\ref{sec:algebraic-query-complexity}: \cref{lem:sum-query-lower-bound}} \\ lower bounds for \\ algebraic query complexity \\ of polynomial summation};
\node (weak) [oldresult] at (4.8,2) {weak PZK sumcheck \\ \cite{BenSassonCFGRS17}};
\node (pzkcomm) [newresult] at (8.7,2) {\ResultHeader{\S\ref{sec:strong-zk-sumcheck}} \\ perfectly-hiding \\ statistically-binding \\ algebraic commitment};
\node (strong) [newresult] at (8.7,3.6) {\ResultHeader{\S\ref{sec:strong-zk-sumcheck}: \cref{lem:strong-zk-sumcheck}} \\ strong PZK sumcheck};
\node (bfl) [smallresult] at (-0.7,5) {IPCP for $\NEXP$ \\ from \cite{BabaiFL91}};
\node (pcp) [smallresult] at (2.2,5) {low-degree extension \\ of any PCP};
\node (sharp) [smallresult] at (4.8,5) {PZK IPCP \\ for $\sharpP$ \\ \cite{BenSassonCFGRS17}};
\node (nexp) [newresult] at (8.7,5) {\ResultHeader{\S\ref{sec:zk-nexp}: \cref{thm:pzk-for-nexp}} \\ PZK IPCP \\ for $\NEXP$};
\node (polymip) [smallresult] at (-0.7,8.6) {poly-round \MIPStar{} \\ for $\NEXP$ \\ \cite{ItoV12}};
\node (twomip) [smallresult] at (2.2,8.6) {1-round \MIPStar{} \\ for $\NEXP$ \\ \cite{Vidick16}};
\node (zkmipsharpp) [newresult] at (4.8,8.6) {\ResultHeader{\S\ref{sec:putting-it-all-together}: \cref{remark:sharpp-mip}} \\ PZK \MIPStar{} \\ for $\sharpP$};
\node (zkmipnexp) [newresult] at (8.7,8.6) {\ResultHeader{\S\ref{sec:putting-it-all-together}: \cref{thm:main-zk}} \\ PZK \MIPStar{} \\ for $\NEXP$};
\node (ldt) [smallresult] at (11.1,7.25) {low-degree test sound \\ against entangled provers  \\\cite{NatarajanV17} (see \cref{thm:quantum_low_degree_test})};
\draw [arrow] (aqc) -- (pzkcomm);
\draw [arrow] (scd) -- (weak);
\draw [arrow] (raz) -- (scd);
\draw [arrow] (pzkcomm) -- (strong);
\draw [arrow] (weak) -- (strong.west);
\draw [arrow] (weak) -- (sharp);
\draw [arrow] (strong) -- (nexp);
\draw [arrow] (pcp) -- (twomip);
\draw [arrow] (bfl) -- (polymip);
\draw [arrow] (sharp) -- (zkmipsharpp);
\draw [arrow] (nexp) -- (zkmipnexp);
\node (overlay1) [backbox,  minimum height=1.4cm, minimum width=9.8cm] at (-0.9,5.8) {};
\node (overlay2) [backbox,  minimum height=0.5cm, minimum width=9.6cm] at (-0.8,5.9) {};
\node [htext] at (3.6,6.1) {\ResultHeader{\cref{prop:round-reduction}:} reduce to single uniform query};
\node [htext] at (3.6,6.9) {\ResultHeader{\cref{lem:lifting}:} make resistant to entangled provers};
\draw [da] (ldt) -- (overlay1);
\end{tikzpicture}
\end{center}
\caption{%
Diagram of the roadmap for proving \cref{thm:main-zk}. White blocks correspond to our new contributions, while grey blocks correspond to previous works.
Building on techniques in algebraic complexity from \cite{RazS05,BenSassonCFGRS17}, we prove lower bounds on algebraic query complexity of polynomial summation (\cref{lem:sum-query-lower-bound}).\label{fig:mip-diag}
This allows us to construct the perfectly-hiding statistically-binding algebraic commitment scheme that underlies our strong perfect zero knowledge sumcheck protocol (\cref{lem:strong-zk-sumcheck}, which also relies on the weak zero knowledge sumcheck protocol in \cite{BenSassonCFGRS17}), and in turn, prove that there exists a perfect zero knowledge low-degree IPCP for any language in $\NEXP$ (\cref{thm:pzk-for-nexp}).
Finally, we show a lemma that lifts low-degree IPCPs to \MIPStar{} protocols (\cref{lem:lifting}), while \emph{preserving zero knowledge}, and use it to derive our main result (\cref{thm:main-zk}); namely, a perfect zero knowledge low-degree \MIPStar{} for any language in $\NEXP$.
Taking an alternative route, we can apply our lifting lemma to a zero knowledge low-degree IPCP in \cite{BenSassonCFGRS17} to obtain a weaker variant of our main result: a zero knowledge low-degree \MIPStar{} for any language in $\sharpP$.
We also reframe previous works \cite{ItoV12, Vidick16} via our framework.
}
\end{figure}

\clearpage
\section{Preliminaries}
\label{sec:prelims}
We cover the notation and basic definitions that are shared by both parts of this paper.

\subsection{Notation}
\label{sec:notation}

For $n \in \Naturals$ we denote by $[n]$ the set $\{1,\ldots,n\}$. For $m,n \in \Naturals$ we denote by $m+[n]$ the set $\{m+1,\ldots,m+n\}$. For a set $X$, $n \in \Naturals$, $I \subseteq [n]$, and $\vec{x} \in X^{n}$, we denote by $\vec{x}_{I}$ the vector $\big(x_{i}\big)_{i \in I}$ that is $\vec{x}$ restricted to the coordinates in $I$.

\parhead{Integrality}
All (relevant) integers stated as real numbers are implicitly rounded to the closest integer.

\parhead{Distance}
The \emph{relative Hamming distance} (or just \emph{distance}), over alphabet $\Sigma$, between two strings $x,y \in \Sigma^n$ is $\Distance{x}{y} \DefineEqual \SetCardinality{\Set{i \in[n] \text{ s.t. } x_i \neq y_i}}/n$. If $\Distance{x}{y} \leq \eps$, we say that $x$ is \emph{$\eps$-close} to $y$; otherwise we say that $x$ is \emph{$\eps$-far} from $y$. Similarly, the \emph{relative distance} of $x$ from a non-empty set $S \subseteq \Sigma^n$ is $\Distance{x}{S} \DefineEqual \min_{y \in S} \Distance{x}{y}$. If $\Distance{x}{S} \leq \eps$, we say that $x$ is \emph{$\eps$-close} to $S$; otherwise we say that $x$ is \emph{$\eps$-far} from $S$.

\parhead{Functions, distributions, fields}
We use $f \colon \Domain \to \Range$ to denote a function with domain $\Domain$ and range $\Range$; given a subset $\SubDomain$ of $\Domain$, we use $\Restrict{f}{\SubDomain}$ to denote the restriction of $f$ to $\SubDomain$. Given a distribution $\Distribution$, we write $x \gets \Distribution$ to denote that $x$ is sampled according to $\Distribution$. We denote by $\Field$ a finite field and by $\Field_{\FieldSize}$ the field of size $\FieldSize$. Arithmetic operations over $\Field_{q}$ take time $\polylog q$ and space $O(\log q)$.

\parhead{Polynomials}
We denote by $\PolynomialRing{\Field}{\SCVars}{\VariableX}$ the ring of $\SCVars$-variable polynomials over the field $\Field$. Given a polynomial $P$ in $\PolynomialRing{\Field}{m}{\VariableX}$, $\IndividualDegree{P}[\VariableX_{i}]$ is the degree of $P$ in the variable $\VariableX_{i}$. The \emph{individual degree} of a polynomial is its maximum degree in any variable, $\deg{P} \DefineEqual \max_{1 \leq i \leq m}{\IndividualDegree{P}[\VariableX_{i}]}$. Throughout, unless explicitly specified otherwise, we will exclusively work with \emph{individual} degree and often refer to it simply as degree. We denote by $\PolynomialRingIndOne{\Field}{\SCVars}{\VariableX}{\SCDegree}$ the subspace of all polynomials $P \in \PolynomialRing{\Field}{\SCVars}{\VariableX}$ such that $\deg{P} \leq \SCDegree$.

\parhead{Low-degree extensions}
Given a finite field $\Field$, subset $\SCSubset \subseteq \Field$, and number of variables $\SCVars \in\Naturals$, the \emph{low-degree extension} (LDE) of a function $f \colon \SCSubset^\SCVars \to \Field$ is the \emph{unique} polynomial of individual degree $\SetCardinality{\SCSubset}-1$ that agrees with $f$ on $\SCSubset^\SCVars$, i.e., $\LD{f}\in\PolynomialRingIndOne{\Field}{\SCVars}{\VariableX}{\SetCardinality{\SCSubset}-1}$ such that $\LD{f}(\vec{h}) = f(\vec{h})$ for all $\vec{h}\in \SCSubset^\SCVars$. In particular, $\LD{f} \colon \Field^{m} \to \Field$ is defined as follows:
\begin{equation*}
\LD{f}(\vec{\VariableX})
\DefineEqual
  \sum_{\vec{\beta} \in H^{m}}
  \Lagrange{H^{m}}(\vec{\VariableX}, \vec{\beta})
  \cdot
  f(\vec{\beta})
\enspace,
\end{equation*}
where $\Lagrange{H^{m}}(\vec{\VariableX}, \vec{\VariableY}) \DefineEqual \prod_{i=1}^{m} \sum_{\omega \in H} \prod_{\gamma \in H \setminus \{\omega\}} \frac{(\VariableX_{i} - \gamma)(\VariableY_{i} - \gamma)}{(\omega - \gamma)^{2}}$ is the unique polynomial in $\PolynomialRingIndOne{\Field}{m}{\VariableX}{\SetCardinality{H}-1}$ such that, for all $(\vec{\alpha},\vec{\beta}) \in H^{m} \times H^{m}$, $\Lagrange{H^{m}}(\vec{\alpha},\vec{\beta})$ equals $1$ when $\vec{\alpha}=\vec{\beta}$ and equals $0$ otherwise. Note that $\Lagrange{H^{m}}(\vec{\VariableX}, \vec{\VariableY})$ can be generated and evaluated in time $\poly(\SetCardinality{H}, m, \log \SetCardinality{\Field})$ and space $O(\log \SetCardinality{\Field} + \log m)$, so $\LD{f}(\vec{\alpha})$ can be evaluated in time $\SetCardinality{H}^{m} \cdot \poly(\SetCardinality{H}, m, \log \SetCardinality{\Field})$ and space $O(m \cdot \log \SetCardinality{\Field})$.

\parhead{Languages and relations}
We denote by $\Language$ a language consisting of \emph{instances} $\Instance$, and by $\Relation$ a (binary ordered) relation consisting of pairs $(\Instance,\Witness)$, where $\Instance$ is the \emph{instance} and $\Witness$ is the \emph{witness}. We denote by $\GetLanguage{\Relation}$ the language corresponding to $\Relation$, and by $\Witnesses{\Relation}{\Instance}$ the set of witnesses in $\Relation$ for $\Instance$ (if $\Instance \not\in \GetLanguage{\Relation}$ then $\Witnesses{\Relation}{\Instance} \DefineEqual \emptyset$). We assume that $\BitSize{\Witness}$ is bounded by some computable function of $\InstanceSize \DefineEqual \BitSize{\Instance}$; in fact, we are mainly interested in relations arising from nondeterministic languages: $\Relation \in \NTIME(\DeciderTime)$ if there exists a $\DeciderTime(\InstanceSize)$-time machine $\DeciderMachine$ such that $\DeciderMachine(\Instance,\Witness)$ outputs $1$ if and only if $(\Instance,\Witness) \in \Relation$. We assume that $\DeciderTime(\InstanceSize) \geq \InstanceSize$.

\parhead{Randomized algorithms and oracle access}
We denote by $A^\Oracle(x)$ the output of an algorithm $A$ when given an input $x$ (explicitly) and query access to an oracle $\Oracle$. If $A$ is probabilistic then $A^\Oracle(x)$ is a random variable and, when writing expressions such as \DoQuote{$\Pr[A^\Oracle(x) = z]$}, we mean that the probability is taken over $A$'s internal randomness (in addition to other randomness beyond it). The algorithm is said to be \emph{$\QueryBound$-query} if it makes \emph{strictly less than} $\QueryBound$ queries to its oracle. Given two interactive algorithms (protocols) $A$ and $B$, we denote by $(A^\Oracle(x),B^\Oracle(y))(z)$ the output of $A^\Oracle(x)$ when interacting with $B^\Oracle(y)$ on common input $z$.

\subsection{Low-degree interactive PCPs}
\label{sec:ipcps}

An \emph{interactive PCP} (IPCP) \cite{KalaiR08} is a probabilistically checkable proof (PCP) verifiable via an interactive proof (IP). In more detail, an IPCP protocol for a language $\Language$ is a pair of probabilistic interactive algorithms $(\Prover, \Verifier)$, where the \emph{prover} $\Prover$ is computationally unbounded and the \emph{verifier} $\Verifier$ runs in polynomial time. Both parties receive an (explicit) input $x \in \Bits^n$, allegedly in the language $\Language$, and engage in a two-phase protocol as follows. First, $\Prover$ sends to $\Verifier$ an oracle proof string $\Proof \in \Bits^{*}$. Second, $\Prover$ and $\Verifier^{\Proof}$ (i.e., $\Verifier$ with oracle access to $\Proof$) engage in an interactive protocol, at the end of which $\Verifier$ either accepts or rejects. The completeness property requires that, if $x \in \Language$, then there exists a prover $\Prover$ such that $\Pr[(\Prover,\Verifier)(x) = 1] = 1$. The soundness property requires that, if $x \not\in \Language$, then for any prover $\Malicious{\Prover}$ it holds that $\Pr[(\Malicious{\Prover},\Verifier)(x) = 1] \leq \SoundnessError(n)$, where $\SoundnessError$ is called the \emph{soundness error}. Unless specified otherwise, we define our IPCP protocols with respect to a small constant soundness error, say, $\SoundnessError = 1/2$.

A round of interaction consists of one message from each of the parties. We say that an IPCP has \emph{round complexity} $\RoundComplexity$ if the second step of the interaction (the standard IP) consists of $\RoundComplexity$ rounds. The \emph{PCP length} of an IPCP is the length of $\Proof$. The \emph{communication complexity} is the total number of bits exchanged between the parties \emph{except} for the message that contains $\Proof$. The \emph{query complexity} is the number of queries that $\Verifier$ makes to the PCP $\Proof$. We write
\begin{equation*}
	\Language \in \IPCPparams{\SoundnessError}{\RoundComplexity}{\ProofLength}{\CommunicationComplexity}{\QueryComplexity}
\end{equation*}
to indicate that a language $\Language$ has an IPCP with the specified parameters.

\parhead{Low-degree IPCPs}
A key tool that we use is \emph{low-degree IPCPs}, a class of algebraically-structured IPCPs. Informally, these are IPCPs in which the PCP oracle is \emph{promised} to be a low-degree polynomial.
In more detail, given a field $\Field$, number of variables $\SCVars$, and degree $\SCDegree$, we say that an IPCP protocol $(\Prover, \Verifier)$ is an \LDIPCP{} if the following conditions hold.
\begin{itemize}[nolistsep]
	\item \emph{Low-degree completeness}: The PCP oracle that the (honest) prover $\Prover$ sends is a polynomial $\Poly$ in $\PolynomialRingIndOne{\Field}{\SCVars}{\VariableX}{\SCDegree}$.
	\item \emph{Low-degree soundness}: soundness is merely required to hold against provers $\Malicious{\Prover}$ that send PCP oracles that are polynomials $\Malicious{\Poly}$ in $\PolynomialRingIndOne{\Field}{\SCVars}{\VariableX}{\SCDegree}$.
\end{itemize}
We remark that the notion of low-degree IPCPs is closely related to \emph{holographic} IPCPs and IPs \cite{ReingoldRR16,GurR17}. However, whereas in holographic proof systems the \emph{input} is guaranteed to be encoded as a low-degree polynomial, in low-degree IPCPs the oracle may not be related to the input in any way.

\parhead{Public-coin interaction}
Our protocols and transformations refer to \emph{public-coin} proof systems. We remark that the only part wherein we rely on public-coin interaction is in the transformation of IPCPs to low-degree IPCPs in \cref{sec:single-unif-q}. In fact, for this transformation it suffices to rely on a weaker condition that is implied by public-coin interaction; namely, all we require is that the verifier queries the PCP oracle \emph{after} the communication with the prover terminates.

\parhead{Adaptivity}
For simplicity, we assume that all (public-coin) IPCP verifiers make non-adaptive queries to their oracle. However, all of our results can be extended, in a straightforward way, to hold with respect to verifiers that make adaptive queries, at the cost of an increase in round complexity. (See \cref{remark:on-adaptivity}.)

\parhead{Zero knowledge}
We consider the standard notion of (perfect) zero knowledge for IPCPs \cite{GoyalIMS10,BenSassonCFGRS17}. Let $A,B$ be algorithms and $x,y$ strings. We denote by $\IPCPView{B(y)}{A(x)}$ the \emph{view} of $A(x)$ in an IPCP protocol with $B(y)$, i.e., the random variable $(x,r,s_{1},\dots,s_{n},t_{1},\dots,t_{m})$ where $x$ is $A$'s input, $r$ is $A$'s randomness, $s_{1},\dots,s_{n}$ are $B$'s messages, and $t_{1},\dots,t_{m}$ are the answers to $A$'s queries to the proof oracle sent by $B$.

An IPCP protocol $\pair{\Prover}{\Verifier}$ for a language $\Language$ is \emph{(perfect) zero knowledge against query bound $\QueryBound$}
if there exists a polynomial-time simulator algorithm $\Simulator$ such that for every $\QueryBound$-query algorithm $\Malicious{\Verifier}$ and input $x \in \Language$ it holds that $\Simulator^{\Malicious{\Verifier}} (x)$ and $\IPCPView{\Prover(x)}{\Malicious{\Verifier}(x)}$ are identically distributed. 
We write
\begin{equation*}
	\Language \in \PZKIPCPparams{\SoundnessError}{\RoundComplexity}{\ProofLength}{\CommunicationComplexity}{\QueryComplexity}{\QueryBound}
\end{equation*}
to indicate that a language $\Language$ has a perfect zero knowledge IPCP with the specified parameters.

\begin{remark}[straightline simulators]
The aforementioned works (\cite{GoyalIMS10,BenSassonCFGRS17}) consider a stronger notion of zero knowledge IPCPs in which the simulator is \emph{straightline}, i.e., the simulator cannot rewind the verifier. All of the simulators that we construct in this work are straightline too; even so, all of the transformations presented in this work preserve zero knowledge even for simulators that rewind the verifier.
\end{remark}

\doclearpage
\part{Low-degree IPCP to \MIPStar{}}
\label{part:I}

In this part we build on recent advances in low-degree testing against entangled provers \cite{ItoV12,Vidick16,NatarajanV17} to prove a \emph{lifting lemma} that transforms a class of algebraically-structured classical protocols, namely low-degree interactive PCPs, into \MIPStar{} protocols, while \emph{crucially}, preserving zero knowledge. 

\parhead{Organization}
We begin in \cref{sec:quantum-information} by covering the necessary preliminaries regarding quantum information and proof systems with entangled provers. In \cref{sec:quantum-individual} we discuss the main technical tool that we need: a low-degree test against entangled provers, which we refine from a total degree to an individual degree test. Then, in \cref{sec:lowdegree_to_ZK} we state and prove our transformation of low-degree IPCP to \MIPStar, while preserving zero knowledge. Finally, in \cref{sec:putting-it-all-together} we prove our main result (\cref{thm:main-zk}) by applying the foregoing transformation to a zero knowledge low-degree IPCP for $\NEXP$, which we construct in \cref{part:II}.

\section{Preliminaries: proof systems with entangled provers}
\label{sec:quantum-information}

We begin with standard preliminaries in quantum information. Let $\HilbertSpace$ be a finite-dimensional Hilbert space, and let $\NumRegisters \in \Naturals$.

\parhead{States and operators.}
We define \emph{entangled quantum states}, which for brevity, we will refer to simply as \qstate{}s. An $\NumRegisters$-register \qstate{} $\ket{\Psi}$ is a unit vector in $\HilbertSpace^{\otimes \NumRegisters}$. We say that $\ket{\Psi}$ is \emph{permutation-invariant} if $\sigma \ket{\Psi} = \ket{\Psi}$ for every linear operator that permutes the $\NumRegisters$ registers of $\HilbertSpace^{\otimes \NumRegisters}$. We denote by $\rho = \rho(\ket{\Psi})$ the reduced density of $\ket{\Psi}$ on a number of registers that will be always clear from the context, so that for an operator $A$ we have
\begin{equation*}
  \TraceRho{A} \DefineEqual \Trace(A \rho)
= \bra{\Psi} A \otimes \Id \otimes \cdots \otimes \Id \ket{\Psi}
\enspace.
\end{equation*}

Let $\LinearOperators{\HilbertSpace}$ be the set of linear operators over $\HilbertSpace$. Denote by $\Id$ the identity operator in $\LinearOperators{\HilbertSpace}$. For $\NumRegisters \ge 2$, the \qstate{} $\ket{\Psi}$ induces a bilinear form on $\LinearOperators{\HilbertSpace} \times \LinearOperators{\HilbertSpace}$, given by
\begin{equation*}
	\InnerProductPsi{A}{B} = \bra{\Psi} A \otimes B \otimes \Id^{\otimes (\NumRegisters - 2)} \ket{\Psi} \in \Complex \enspace,
\end{equation*}
as well as a semi-norm, given by
\begin{equation*}
  \NormPsi{A}
= \sqrt{\bra{\Psi} A A^{\dagger} \otimes \Id^{\otimes (\NumRegisters - 1)} \ket{\Psi}}
\enspace.
\end{equation*}

\parhead{Measurements.}
All measurement in this work are POVMs (Positive Operator Valued Measures). A \emph{measurement} on $\HilbertSpace$ is a finite set $A = \Set{\Measurement{A}{}{i}}_{i \in S}$ where $S$ is the set of measurement outcomes and the $A^i$'s are non-negative definite operators on $\HilbertSpace$ such that $\sum_{i \in S} \Measurement{A}{}{i} = \Id$. A \emph{sub-measurement} on $\HilbertSpace$ relaxes the aforementioned condition by only requiring that $\sum_{i \in S} A_i \leq \Id$. The following standard claim provides a quantitive bound on the distance between two measurements as a function of their correlation.

\begin{claim}[Approximate consistency to trace distance \cite{Vidick11,Vidick16}]
\label{clm:consistncy-to-trace}
Let $\ket{\Psi}$ be a permutation-invariant \qstate{} on $\NumRegisters \geq 2$ registers, and let $\Set{\Measurement{A}{}{\Element}}$, $\Set{\Measurement{B}{}{\Element}}$ be single-register measurements with outcomes in the same set. 
Then,
\begin{equation*}
   \sum_\Element \| \Measurement{A}{}{\Element} - \Measurement{B}{}{\Element} \|_{\Psi}^{2}
\le O \left( \max \left\{ 1- \sum_\Element \InnerProduct{\Measurement{A}{}{\Element}}{\Measurement{B}{}{\Element}}_{\Psi}, 1- \sum_\Element \InnerProduct{\Measurement{A}{}{\Element}}{\Measurement{A}{}{\Element}}_{\Psi}   \right\}\right)
\enspace.
\end{equation*}
\end{claim}

We shall also need a specific variant of Winter's gentle measurement lemma \cite{Winter99}, due to Ogawa and Nagaoka \cite{OgawaN07}, which formalizes the intuition that measurements that with high probability output a particular outcome on a certain quantum state imply that the post-measurement state is close to the original state.

\begin{lemma}[\cite{OgawaN07}]
\label{lem:gentle-measurement}
Let $\rho$ be a density operator on a Hilbert space $\HilbertSpace$, and let $A,B \in \LinearOperators{\HilbertSpace}$ be such that $A^\dagger A, B^\dagger B \le \Id$. Then,
\begin{equation*}
     \NormOne{A \rho A^\dagger - B \rho B^\dagger}
\leq 2 \sqrt{\Trace \left( (A - B) \rho (A - B)^\dagger \right)}
\enspace.
\end{equation*}
\end{lemma}

\parhead{MIPs with entangled provers}
A \emph{multi-prover interactive proof with entangled provers} (\MIPStar{}) \cite{CleveHTW04} is a multi-prover interactive proof (MIP) in which the spatially-isolated (honest and malicious) provers are allowed to use entangled strategies, i.e., any strategy obtained by measuring a shared \qstate{}.

In more detail, a $\NumProvers$-prover $\RoundComplexity$-round \MIPStar{} for a language $\Language$ is a tuple of probabilistic interactive algorithms $(\Prover_1, \ldots, \Prover_\NumProvers, \Verifier)$, where the \emph{provers} $\Prover_1, \ldots, \Prover_\NumProvers$ are computationally unbounded and the \emph{verifier} $\Verifier$ runs in polynomial time. All parties receive an input $x \in \Bits^n$, and the $\NumProvers$ provers share an \qstate{} $\ket{\Psi} \in \mathcal{H}^{\otimes \NumProvers}$ (which may depend on $\ket{\Psi}$), for a Hilbert space $\mathcal{H}$. The parties engage in a protocol of the following type. In each of the $\RoundComplexity$ rounds, each prover $\Prover_i$ receives a (classical) message from the verifier in a message register, performs a quantum operation on this register together with its share of the \qstate{} $\ket{\Psi}$, measures the message register, and sends back the (classical) outcome to the verifier.

We require perfect completeness and soundness with a given error $\SoundnessError$. If $x \in \Language$ then $(\Prover_1, \ldots, \Prover_\NumProvers)$ make $\Verifier$ accept with probability $1$; if $x \not\in \Language$ then $\Verifier$ rejects \emph{every} prover strategy with probability at least $1 - \SoundnessError(n)$, where $\SoundnessError$ is called the \emph{soundness error}. We stress that the latter condition, soundness, makes no assumptions on the computational power of the provers, nor regarding the number of entangled qubits they share. Unless specified otherwise, we define an \MIPStar{} with respect to a small constant soundness error, say, $\SoundnessError = 1/2$.\footnote{This constant is arbitrary and can be amplified via parallel repetition for entangled strategies \cite{KempeV11,Yuen16}.}

We indicate that a language $\Language$ has a $\NumProvers$-prover $\RoundComplexity$-round \MIPStar{} with soundness error $\SoundnessError$ and with communication complexity $\CommunicationComplexity$ (the total number of bits exchanged between the verifier and the provers) as follows:
\begin{equation*}
\Language \in \MIPSparams{\NumProvers}{\SoundnessError}{\RoundComplexity}{\CommunicationComplexity}
\enspace.
\end{equation*}

\parhead{Zero knowledge}
We extend the standard definition of perfect zero knowledge MIPs \cite{BenOrGKW88} to the setting of \MIPStar{} in the natural way. Denote by $\MIPView{\Prover_1, \ldots, \Prover_{\NumProvers}}{\Malicious{\Verifier}}(x)$ the \emph{view} of a (possibly malicious) verifier $\Malicious{\Verifier}$ in a $\NumProvers$-prover \MIPStar{} with provers $\Prover_1, \ldots, \Prover_\NumProvers$ and input $x$; that is, the verifier's view is the random variable consisting of the input $x$, the verifier's random string, and the provers' messages to the verifier.

An \MIPStar{} $(\Prover_1, \ldots, \Prover_\NumProvers, \Verifier)$ for a language $\Language$ is \emph{perfect zero knowledge} if there exists a probabilistic polynomial-time simulator algorithm $\Simulator$ such that for every probabilistic polynomial-time algorithm $\Malicious{\Verifier}$ and input $x \in \Language$ it holds that $\Simulator^{\Malicious{\Verifier}} (x)$ and $\MIPView{\Prover_1, \ldots, \Prover_{\NumProvers}}{\Malicious{\Verifier}}(x)$ are identically distributed. All simulators for \MIPStar{} protocols in this work achieve the stronger notion of universal \emph{straightline simulators} \cite{FeigeS89,DworkS98}, in which the simulator do \emph{not} rewind the verifier. We write
\begin{equation*}
	\Language \in \PZKMIPSparams{\NumProvers}{\SoundnessError}{\RoundComplexity}{\CommunicationComplexity}
\end{equation*}
to indicate that a language $\Language$ has an perfect zero knowledge \MIPStar{} with the specified parameters.

\parhead{Quantum malicious verifiers}
The \MIPStar{} model requires quantum entangled provers and \emph{classical} verifiers (in contrast to the Q\MIPStar{} model, in which both the provers and the verifier are quantum), and so our notion of zero knowledge is with respect to classical verifiers. Nevertheless, we remark that our results extend to hold against \emph{quantum} malicious verifiers. This is because:
\begin{inparaenum}[(1)]
	\item the honest verifier is classical, and so the provers can enforce classical communication by systematically measuring the verifier's answers in the computational basis, and
	\item all of our simulators are \emph{straightline} (i.e., they do \emph{not} rewind the verifier), and so they avoid the key hurdle for simulators of quantum verifiers, which is that quantum algorithms cannot be rewinded (as quantum information cannot be copied, and measurements are irreversible processes).
\end{inparaenum}

\parhead{Symmetric strategies}
Symmetry plays an important simplifying role in the analysis of an \MIPStar{}. The following lemma, due to Kempe et al.\ \cite{KempeKMTV11}, states that if the verifier treats provers symmetrically (in this paper this is always the case) then we can assume, without loss of generality, that the provers' optimal strategy is symmetric (all provers use the same measurement) and that any shared \qstate{} is permutation invariant.

\begin{lemma}[\cite{KempeKMTV11}]
\label{lem:asymmetry}
Let $(\Prover_1, \ldots, \Prover_\NumProvers, \Verifier)$ be a $\NumProvers$-prover \MIPStar{} for a language $\Language$ in which the verifier $\Verifier$ treats the provers symmetrically. If there exists a prover strategy $(\ProverStrategy_1, \ldots, \ProverStrategy_\NumProvers)$ with an \qstate{} $\ket{\Psi}$ that succeeds with probability $\eps$, then there also exists a \emph{symmetric} prover strategy $(\ProverStrategy', \ldots, \ProverStrategy')$ with a \emph{permutation-invariant} \qstate{} $\ket{\Psi'}$ that succeeds with probability $\eps$.
\end{lemma}

\doclearpage
\section{Low individual-degree testing against entangled quantum strategies}
\label{sec:quantum-individual}

A \emph{low-degree test} is a procedure used to determine if a given function $f \colon \Field^\SCVars \to \Field$ is close to a low-degree polynomial in $\PolynomialRing{\Field}{\SCVars}{\VariableX}$ or if, instead, it is far from all low-degree polynomials, by examining $f$ at very few locations. A test typically consists of examining $f$ at random restrictions, such as a point, line, or plane.

Low-degree tests can also be phrased in the setting of multiple non-communicating provers, where each prover is (allegedly) answering queries about the same function $f$ that is being tested. For example, the celebrated line-vs-point test of Rubinfeld and Sudan \cite{RubinfeldS96} can be viewed as a $2$-prover $1$-round MIP protocol. The verifier specifies a random line in $\Field^\SCVars$ to one prover and a random point on this line to the other prover; each prover replies with the purported value of $f$ on the received line or point; then the verifier checks that these values are consistent.


Loosely speaking, the classical analysis of this test asserts the following conditions:
\begin{inparaenum}[(1)]
  \item \emph{approximate consistency with a low-degree polynomial}, i.e., each player acts as a lookup for a function that is (close to) a low-degree polynomial; and
  \item \emph{self-consistency between the provers}, i.e., both players answer according to the \emph{same} function.
\end{inparaenum}

In this paper we rely on a similar low-degree test, the plane-vs-point test \cite{RazS97}, adapted to the setting of \MIPStar{}, whose analysis asserts the quantum analogue of the conditions above. In fact, we use a more refined version, which tests a polynomial's \emph{individual} degree rather than its \emph{total} degree. In the classical setting, such a test is implicit in \cite[Section 5.4.2]{GoldreichS06} via a reduction from individual-degree to total-degree testing. Informally, this reduction first invokes the test for low total degree, then performs univariate low-degree testing with respect to a random axis-parallel line in each axis. The extension of this reduction to the quantum setting yields an \MIPStar{} for individual-degree testing.

\parhead{Low individual degree test}
Let $\Planes{U}$ be the set of all \emph{planes} in a vector space $U$ (every $\Plane \in \Planes{U}$ is a $2$-dimensional affine subspace of $U$). The plane-vs-point test, with respect to individual degree, for \MIPStar{} is the following $2$-prover $1$-round protocol.

\begin{construction}
\label{con:plane-vs-point-individual}
Let $\Field$ be a finite field, $\SCVars \in \Naturals$ the number of variables, and $\SCDegree \in \Naturals$ the \emph{individual} degree. The quantum plane-vs-point $(\Field, \SCDegree, \SCVars)$-low-individual-degree test is an \MIPStar{} $(\Prover_1, \Prover_2, \Verifier)$ in which $\Prover_1, \Prover_2$ claim that a certain function $\Poly \colon \Field^\SCVars \to \Field$ is a polynomial of individual degree $\SCDegree$, and the (honest) interaction is as follows.

First the verifier $\Verifier$ symmetrizes the protocol: with probability $1/2$ it assigns the roles $\ProverPoint$ to the first prover and $\ProverPlane$ to the second prover, and with probability $1/2$ it assigns $\ProverPlane$ to the first prover and $\ProverPoint$ to the second prover. Then $\Verifier$ chooses uniformly at random one of the following tests.	
\begin{itemize}
	\item \emph{$(\Field, \SCVars\SCDegree, \SCVars)$-total-degree test:}
	\begin{enumerate}[nolistsep]
 		\item The verifier $\Verifier$ samples a random plane $\Plane \in \Planes{\Field^\SCVars}$ and a random point $\Point \in \Plane$ on that plane.
		\item $\Verifier$ sends the plane $\Plane$ to $\ProverPlane$, and the line $\Point$ to $\ProverPoint$.
		\item $\ProverPlane$ replies with $g \DefineEqual Q \circ \Plane$ (the bivariate polynomial obtained by restricting $\Poly$ to $\Plane$).
		\item $\ProverPoint$ replies with $\Element \DefineEqual Q(\Point)$ (the value of $Q$ at $\Point$).
		\item $\Verifier$ checks that $g$ is a polynomial of total degree $\SCVars\SCDegree$ and accepts if and only if $g(\Point) = \Element$.
	\end{enumerate}
	
	\item \emph{Axis-parallel univariate $(\Field, \SCDegree)$-degree test:}
	\begin{enumerate}[nolistsep]
		\item The verifier $\Verifier$ samples a random plane $\Plane \in \Planes{\Field^\SCVars}$, a random point $\Point \in \Plane$ on that plane, and a random axis-parallel line $\Line$ passing through the point $\Point$.
		\item $\Verifier$ sends the plane $\Plane$ to $\ProverPlane$, and the line $\Line$ to $\ProverPoint$.
		\item $\ProverPlane$ replies with $g \DefineEqual Q \circ \Plane$ (the bivariate polynomial obtained by restricting $\Poly$ to $\Plane$).
	  	\item $\ProverPoint$ replies with $h \DefineEqual Q \circ \Line$ (the univariate polynomial obtained by restricting $\Poly$ to $\Line$).
	  	\item $\Verifier$ checks that $h$ is a polynomial of individual degree $\SCDegree$ and accepts if and only if $g(\Point) = h(\Point)$.
	\end{enumerate}
\end{itemize}
\end{construction}

Perfect completeness follows since if $\Poly$ is indeed a polynomial of individual degree $\SCDegree$, then $g \DefineEqual Q \circ \Plane$ is a bivariate polynomial of total degree $\SCVars\SCDegree$ and $h \DefineEqual Q \circ \Plane$ is a univariate polynomial of degree $\SCDegree$, and so both provers can simply answer according to $\Poly$. We are grateful to Thomas Vidick for allowing to include the following theorem (and its proof), which shows that the plane-vs-point individual-degree test in \cref{con:plane-vs-point-individual} is sound against entangled strategies.

\begin{theorem}[quantum low individual degree test]
\label{thm:quantum_low_degree_test}
There exist absolute constants $c \in [0,1]$ and $C \ge 1$ such that the following holds. Let $\SoundnessError > 0$, $\SCVars, \SCDegree \in \Naturals$, and let $\Field$ be a finite field of size $\SetCardinality{\Field} = (\SCVars \SCDegree/\SoundnessError)^C$. Let $\ProverStrategy$ be a symmetric prover strategy using \qstate{} $\ket{\Psi} \in \HilbertSpace \otimes \HilbertSpace$ and projective measurements $\Set{\Measurement{A}{\Point}{\Element}}_{\Point \in \Field^\SCVars, \Element \in \Field}$.
If the strategy $(\Prover, \Prover)$ is accepted by the $( \Field, \SCDegree, \SCVars)$-low-degree test in \cref{con:plane-vs-point-individual} with probability at least $1-\SoundnessError$, then there exists a measurement $\Set{\Measurement{L}{}{\Poly}}_{\Poly \in \PolynomialRingIndOne{\Field}{\SCVars}{\VariableX}{\SCDegree}}$ that satisfies the following properties.
\begin{enumerate}
	
  \item \emph{Approximate consistency with $\Set{\Measurement{A}{\Point}{\Element}}_{\Element \in \Field, \Point \in \Field^\SCVars}$:}
\begin{equation*}
\E_{\Point \in \Field^\SCVars}
\sum_{\Poly \in \PolynomialRingIndOne{\Field}{\SCVars}{\VariableX}{\SCDegree}}
\sum_{\substack{\Element \in \Field \\ \Element \neq \Poly(\Point)}}
\bra{\Psi} \Measurement{A}{\Point}{\Element} \otimes \Measurement{L}{}{\Poly} \ket{\Psi}
\leq \SoundnessError^c
\enspace.
\end{equation*}

  \item \emph{Self-consistency of $\Set{\Measurement{L}{}{\Poly}}_{\Poly \in \PolynomialRingIndOne{\Field}{\SCVars}{\VariableX}{\SCDegree}}$:}
\begin{equation*}
\sum_{\Poly \in \PolynomialRingIndOne{\Field}{\SCVars}{\VariableX}{\SCDegree}}
\bra{\Psi} \Measurement{L}{}{\Poly} \otimes (\Id -\Measurement{L}{}{\Poly}) \ket{\Psi}
\leq \SoundnessError^c
\enspace.
\end{equation*}

\end{enumerate}	
\end{theorem}


\begin{proof}[Proof of \cref{thm:quantum_low_degree_test}.]
The proof relies on the analysis of the plane-vs-point test \cite{RazS97} for \MIPStar{}, due to Natarajan and Vidick \cite{NatarajanV17}, which asserts that the provers in an \MIPStar{} are answering according to a polynomial of low \emph{total} degree. This new analysis improves on the analysis of the multilinearity test in \cite{ItoV12} and the $3$-prover low-degree test in \cite{Vidick16}.\footnote{The extension of the analysis to a low-degree test (rather a multilinearity test as in \cite{ItoV12}) is crucial for our results (see discussion in \cref{sec:algebraic-query-complexity}). Instead, the improvement of the $3$-prover test in \cite{Vidick16} to the $2$-prover test in \cite{NatarajanV17} simply reduces the number of provers required to obtain zero knowledge.}

Throughout, we fix $\SoundnessError > 0$, $\SCVars, \SCDegree \in \Naturals$, and a finite field $\Field$ such that $\SetCardinality{\Field} = (\SCVars \SCDegree/\SoundnessError)^C$ for the absolute constant $C \ge 1$ in \cref{thm:quantum_low_degree_test}. Furthermore, we assume that all \MIPStar{} prover strategies are symmetric with respect to a permutation-invariant bipartite \qstate{} and that all measurements are projective.

Recall that the $(\Field, \SCVars\SCDegree, \SCVars)$-total-degree test (a sub-procedure in \cref{con:plane-vs-point-individual}) is an adaptation of the classical plane-vs-point test to the setting of $2$-prover $1$-round \MIPStar{}, in which the verifier specifies a random $2$-dimensional plane in $\Field^\SCVars$ to one prover and a random point on this plane to the other prover; each prover replies with the purported value of $f$ on the received plane or point; and the verifier checks that these values are consistent. The following theorem shows that this test is sound against entangled quantum provers.

\begin{theorem}[Natarajan and Vidick \cite{NatarajanV17}]
\label{thm:quantum_low_degree_test_total}
There exists an absolute constant $c \in [0,1]$ such that the following holds. Let $(\Prover, \Prover)$ be a symmetric prover strategy using an \qstate{} $\ket{\Psi} \in \HilbertSpace \otimes \HilbertSpace$ and measurements $\Set{\Measurement{A}{\Point}{\Element}}_{\Element \in \Field, \Point \in \Field^\SCVars}$. If the strategy $(\Prover,\Prover)$ is accepted by the $(\Field, \SCVars\SCDegree, \SCVars)$-total-degree test with probability at least $1-\SoundnessError$, then there exists a measurement $\Set{\Measurement{L}{}{\Poly}}_{\Poly \in \PolynomialRingIndOne{\Field}{\SCVars}{\VariableX}{\SCDegree}}$ that satisfies the following.
\begin{enumerate}
	
  \item \emph{Approximate consistency with $\Set{\Measurement{A}{\Point}{\Element}}_{\Element \in \Field, \Point \in \Field^\SCVars}$:}
\begin{equation*}
\E_{\Point \in \Field^\SCVars}
\sum_{\Poly \in \PolynomialRingIndOne{\Field}{\SCVars}{\VariableX}{\SCDegree}}
\sum_{\substack{\Element \in \Field \\ \Element \neq \Poly(\Point)}}
\bra{\Psi} \Measurement{A}{\Point}{\Element} \otimes \Measurement{L}{}{\Poly} \ket{\Psi}
\leq \SoundnessError^c
\enspace.
\end{equation*}

  \item \emph{Self-consistency of $\Set{\Measurement{L}{}{\Poly}}_{\Poly \in \PolynomialRingIndOne{\Field}{\SCVars}{\VariableX}{\SCDegree}}$:}
\begin{equation*}
\sum_{\Poly \in \PolynomialRingIndOne{\Field}{\SCVars}{\VariableX}{\SCDegree}}
\bra{\Psi} \Measurement{L}{}{\Poly} \otimes (\Id -\Measurement{L}{}{\Poly}) \ket{\Psi}
\leq \SoundnessError^c
\enspace.
\end{equation*}

\end{enumerate}	
\end{theorem}

We remark that the above result is stated in \cite{NatarajanV17} for finite fields of prime order. Nevertheless, inspection of the proof there reveals that the result in fact holds for \emph{any} finite field.

Recall that, in the $(\Field, \SCDegree, \SCVars)$-low-individual-degree test (\cref{con:plane-vs-point-individual}), the verifier flips a coin to choose whether to invoke the aforementioned \emph{$(\Field, \SCVars\SCDegree, \SCVars)$-total-degree test} or the \emph{axis-parallel univariate $(\Field, \SCDegree)$-degree test}. In the latter, the verifier samples a random plane $\Plane \in \Planes{\Field^\SCVars}$, a random point $\Point \in \Plane$ on that plane, and a random axis-parallel line $\Line$ passing through the point $\Point$; sends the line to one prover and the plane to the other; and checks that the provers reply with low-degree polynomials that agree on $\Point$.

The \emph{total}-degree test reduces the prover to performing a measurement with outcomes in the set of polynomials of total degree $\SCVars\SCDegree$; we can then argue that, given this, the total contribution of outcomes where the polynomial has \emph{individual} degree greater than $\SCDegree$ in any one variable is small. We do so by relating the probability that the prover obtains these \DoQuote{bad} outcomes to the rejection probability in the axis-parallel univariate test.

Let $T_{\SCVars\SCDegree}$ be the set of all $\SCVars$-variate polynomials over $\Field$ of \emph{total} degree $\SCVars\SCDegree$. Since the $(\Field, \SCDegree, \SCVars)$-low-individual-degree test invokes the \emph{$(\Field, \SCVars\SCDegree, \SCVars)$-total-degree test} with probability $1/2$, \cref{thm:quantum_low_degree_test_total} implies that there exist measurements $\Set{\Measurement{L}{}{\Poly}}_{\Poly \in T_{\SCVars\SCDegree}}$ such that the conclusions of \cref{thm:quantum_low_degree_test_total} hold with respect to soundness error $\SoundnessError' \DefineEqual 2\SoundnessError$.

Let $\Lines{U}$ be the set of all \emph{lines} in a vector space $U$ (every $\Line \in \Lines{U}$ is a $1$-dimensional affine subspace of $U$), and let $\Set{\Measurement{M}{\Line}{v}}_{\Line \in \Lines{\Field^\SCVars},\; v \colon \Line \to \Field}$ be the measurement applied by a prover when asked for a line $\Line$. Without loss of generality, assume that the outcomes range over univariate polynomials of degree at most $\SCDegree$, since any other outcome is rejected by the verifier.

Observe that the probability of the verifier rejecting in the axis-parallel univariate $(\Field, \SCDegree)$-degree test is at least the probability that the line obtained via the $M$-measurement disagrees with the point obtained via the $A$-measurement. More precisely,
\begin{equation*}
	\SoundnessError' \geq 
	\E_{\Line \in \Lines{\Field^\SCVars},\; \Point \in \Line}
	\sum_{\Element \in \Field,\; v \colon \Line \to \Field \text{ s.t. } v(\Point) \neq \Element} \InnerProductPsi{\Measurement{A}{\Point}{\Element}}{\Measurement{M}{\Line}{v}} \enspace.
\end{equation*}

Recall that $\deg{\Poly}$ denotes the \emph{individual degree} of a polynomial $\Poly$, and denote by $\Poly(\Line)$ the evaluations of $\Poly$ over the line $\Line$. Using the approximate consistency condition from \cref{thm:quantum_low_degree_test_total} and the fact that the marginal on $\Point$ is uniform, 
\begin{align*} 
2\SoundnessError' + O\left(\sqrt{\SoundnessError^c}\right)
&\geq \E_{\Line \in \Lines{\Field^\SCVars} ,\Point \in \Line} \sum_{\substack{\Poly \in T_{\SCVars\SCDegree},\; v \colon \Line \to \Field \\ \text{ s.t. } \Poly(\Element)\neq v(\Element)}} \,\langle \Measurement{L}{}{\Poly} ,\, \Measurement{M}{\Line}{v} \rangle_{\Psi}\\
&\geq \E_{\Line \in \Lines{\Field^\SCVars}} \sum_{\substack{\Poly \in T_{\SCVars\SCDegree} \\\text{s.t. } \Poly(\Line) \neq v}}
 \,\langle \Measurement{L}{}{\Poly} ,\, \Measurement{M}{\Line}{v} \rangle_{\Psi} - O\left(\frac{\SCVars\SCDegree}{\SetCardinality{\Field}}\right)\\
&\geq \E_{\Line \in \Lines{\Field^\SCVars}} \sum_{\Poly:\, \deg{\Poly} > \SCDegree} \,\langle \Measurement{L}{}{\Poly} ,\, \Id \rangle_{\Psi} - O\left(\frac{\SCVars\SCDegree}{\SetCardinality{\Field}}\right) \;,
\end{align*}
where the second inequality holds since distinct polynomials (of total degree $\SCVars\SCDegree$) on $\Line$ intersect in at most $\SCVars\SCDegree$ points, and the last inequality holds since any univariate (line) polynomial $v$ considered has degree at most $\SCDegree$ (polynomials with a higher degree would be immediately rejected by the verifier).

To conclude, if a $\SCVars$-variate polynomial $\Poly$ has at least one variable in which the individual degree is larger than $\SCDegree$, then its restriction to a random axis-parallel line will have degree larger than $\SCDegree$ with probability at least $O(\frac{1}{\SCVars} - \frac{\SCVars\SCDegree}{\SetCardinality{\Field}})$ over the choice of the line. This concludes the proof, by our assumption regarding the size of the field $\Field$.
\end{proof}

\doclearpage
\section{Lifting from low-degree IPCP to \MIPStar{} while preserving zero knowledge}
\label{sec:lowdegree_to_ZK}

Recall that low-degree IPCPs are IPCP protocols in which the PCP oracle is \emph{promised} to be a low-degree polynomial (see \cref{sec:ipcps}). We prove that any low-degree IPCP can be transformed into a corresponding \MIPStar{}, \emph{while preserving zero knowledge}.

\begin{lemma}[lifting lemma]
\label{lem:lifting}
Let $C \ge 1$ be the absolute constant in \cref{thm:quantum_low_degree_test}. There exists a transformation $\Transformation$ that maps any $\RoundComplexity$-round \LDIPCP{} $(\Prover',\Verifier')$ for a language $\Language$, where $\SCVars, \SCDegree \in \Naturals$ and $\Field$ is a finite field of size $\SetCardinality{\Field} > \max\Set{(2\SCVars \SCDegree)^C, 5\SCDegree\QueryComplexity}$, into a $2$-prover $(\RoundComplexity^* + 2)$-round \MIPStar{} $(\Prover_1, \Prover_2, \Verifier) \DefineEqual \Transformation(\Prover',\Verifier')$ for $\Language$, where $\RoundComplexity^* \DefineEqual \max\Set{\RoundComplexity,1}$.

Furthermore, if the IPCP $(\Prover',\Verifier')$ is zero knowledge with query bound $\QueryBound \geq 2(\SCDegree+1)^{2} + \SCDegree \QueryComplexity + 1$ ($\QueryComplexity$ denotes the query complexity of the honest verifier), then the \MIPStar{} $(\Prover_1, \Prover_2, \Verifier)$ is zero knowledge.
\end{lemma}

In the rest of this section we prove \cref{lem:lifting}. Specifically,
in \cref{sec:preprocessing} we begin with a classical preprocessing step (a query reduction);
in \cref{sec:LDIPCP_to_MIPS} we present our transformation;
in \cref{sec:main_soundness} we prove soundness against entangled provers; and
in \cref{sec:main_zk} we prove preservation of zero knowledge.
The conceptual contribution of \cref{lem:lifting} is that it provides an abstraction of techniques in \cite{ItoV12,Vidick16}.

\begin{remark}[on preserving round complexity]
\label{rem:rounds}
If we do not wish to preserve zero knowledge, then the round complexity of the 
\MIPStar{} that is obtained in \cref{lem:lifting} can be reduced by $1$ (see discussion at the end of \cref{sec:LDIPCP_to_MIPS}). In addition, if the original low-degree IPCP makes a single, uniformly distributed query to its PCP oracle, then the preprocessing step is not required, and we can save an additional round. In particular, if both conditions occur, we obtain an $\RoundComplexity^*$-round \MIPStar{}, fully preserving round complexity.
\end{remark}

\subsection{Classical preprocessing}
\label{sec:preprocessing}

Let $\SCVars, \SCDegree \in \Naturals$, and let $\Field$ be a finite field of size $\SetCardinality{\Field} > (\SCVars \SCDegree/\SoundnessError)^C$. Let $(\Prover',\Verifier')$ be an $\RoundComplexity$-round \LDIPCP{} for a language $\Language$. Denote its oracle by $\Oracle$, query complexity by $\QueryComplexity$, PCP length by $\ProofLength$, communication complexity by $\CommunicationComplexity$, and soundness error by $\SoundnessError = 1/4$.\footnote{The soundness error is reduced, via standard parallel repetition, to $1/4$, since we next apply a transformation that slightly increases the soundness error, and we wish to end up with soundness error at most $1/2$.} If $(\Prover',\Verifier')$ is zero knowledge with respect to a query bound, denote this bound by $\QueryBound$.

The preprocessing step, which is purely classical, allows us to transform any low-degree IPCP into one that makes a single uniform query, at only a small cost in parameters. Crucially, this transformation \emph{preserves zero knowledge} (with minor deterioration in the zero knowledge query bound).

\begin{proposition}
\label{prop:round-reduction}
There exists a transformation $\Transformation$ such that, for every $\SCVars, \SCDegree \in \Naturals$ and finite field $\Field$, if $(\Prover',\Verifier')$ is a public-coin\footnote{In fact, it suffices to satisfy a weaker condition that is implied by public-coin interaction. Specifically, our transformation also works for \emph{private-coin} IPCPs as long as the verifier queries the PCP oracle \emph{after} the interaction with the prover terminates.} \LDIPCP{} with parameters
	\begin{equation*}
		\LDIPCPparams
		{\SoundnessError}
		{\RoundComplexity}
		{\ProofLength}
		{\CommunicationComplexity}
		{\QueryComplexity}
		{\PolynomialRingIndOne{\Field}{\SCVars}{\VariableX}{\SCDegree}},
	\end{equation*}
then $(\Prover'',\Verifier'') \DefineEqual \Transformation(\Prover',\Verifier')$ is a low-degree IPCP for $\Language$ with parameters
	\begin{equation*}
		\LDIPCPparams
		{\SoundnessError' = \SoundnessError + \frac{\SCDegree \QueryComplexity}{|\Field|-\QueryComplexity}}
		{\RoundComplexity' = \RoundComplexity+1}
		{\ProofLength' = \ProofLength}
		{\CommunicationComplexity' = \CommunicationComplexity + \poly(\SCVars, \SCDegree, \QueryComplexity)}
		{\QueryComplexity' = 1}
		{\PolynomialRingIndOne{\Field}{\SCVars}{\VariableX}{\SCDegree}},
	\end{equation*}
	where the verifier's single query is uniformly distributed. Furthermore, if $(\Prover',\Verifier')$ is (perfect) zero knowledge with query bound $\QueryBound$, then $(\Prover'',\Verifier'')$ is (perfect) zero knowledge with query bound $\QueryBound - (\SCDegree \QueryComplexity + 1)$.
\end{proposition}
The proof of \cref{prop:round-reduction} is via a straightforward adaptation of a technique from \cite{KalaiR08}, while keeping track of its effect on zero knowledge; we defer this proof to \cref{sec:single-unif-q}.

We apply \cref{prop:round-reduction} to $(\Prover',\Verifier')$ in order to obtain the \LDIPCP{} $(\Prover'',\Verifier'')$, with parameters as stated in the lemma above, whose verifier makes a single uniformly distributed query to its PCP oracle. In particular, note that the new zero knowledge query bound is $\QueryBound' \geq 2(\SCDegree+1)^{2}$ and the soundness is $1/4 + \frac{\SCDegree \QueryComplexity}{|\Field|-\QueryComplexity} \le 1/2$. We then proceed to transform $(\Prover'',\Verifier'')$ to an \MIPStar{} in \cref{sec:LDIPCP_to_MIPS}.

\begin{remark}[prover-oblivious queries]
After the preprocessing, the verifier makes a single uniform query, which means that its queries are a random variable that is \emph{independent} of the prover messages (but may be correlated with the verifier messages). We refer to this property as \emph{prover-oblivious queries}.	
\end{remark}

\begin{remark}[on adaptivity]
\label{remark:on-adaptivity}
We assumed that all IPCP verifiers make non-adaptive queries to their oracle. However, we can extend all of our results, in a straightforward way, to hold with respect to verifiers that make adaptive queries, at the cost of an increase in round complexity. Specifically, by the public-coin property of our IPCP verifiers, we can assume without loss of generality that the verifier performs its queries \emph{after} the interaction with the prover ceases. After which, the verifier can ask the prover for the evaluation of the oracle, instead of actually querying it (at the cost of an additional round of interaction per adaptive query), and then perform all queries, non-adaptively, at the end.
\end{remark}

\subsection{The transformation}
\label{sec:LDIPCP_to_MIPS}

Recall that $(\Prover'',\Verifier'')$ is an $\RoundComplexity$-round IPCP protocol for a language $\Language$ with soundness error $\SoundnessError$, whose completeness and soundness conditions are with respect to a low-degree PCP oracle $\Oracle \in \PolynomialRingIndOne{\Field}{\SCVars}{\VariableX}{\SCDegree}$ to which the verifier makes a single uniform query. 

To construct an \MIPStar{} for $\Language$, we follow the proof overview presented in \cref{sec:techniques}. However, there is an additional complication that we need to deal with, which we discuss next.

\parhead{Symmetrization and zero knowledge}
Our high-level strategy for constructing a zero knowledge \MIPStar{} for $\Language$ is to let one entangled prover simulate the PCP oracle $\Oracle$ and the other one simulate the IPCP prover $\Prover$, while using the entanglement-resistant low-degree test to assert that the prover simulating $\Oracle$ actually answers according to a low-degree polynomial.

Recall that the analysis of the low-degree test (\cref{thm:quantum_low_degree_test}) requires that the provers employ \emph{symmetric} strategies. Typically, this is handled by letting the verifier randomly choose the roles that the provers play. However, in the setting of zero knowledge \MIPStar{} such a symmetrization causes problems.

Specifically, to prove zero knowledge we need to consider \emph{malicious} verifiers that may abuse the interaction to learn from the provers. In particular, it turns out that if the verifier asks both provers to take the role of the IPCP prover $\Prover$, then the protocol may \emph{no longer} be zero knowledge condition. Indeed, the particular zero knowledge IPCP that we construct in \cref{part:II} to the end of obtaining our zero knowledge \MIPStar{} loses its zero knowledge property if the verifier is allowed to perform two parallel interactions with the prover. 

We overcome this difficulty via the following (non-standard) symmetrization. First, the provers flip a coin (by performing a measurement on $\ket{\Psi}$) to decide which prover is \emph{primary} and which is \emph{secondary}, and send its outcome to the verifier. The secondary prover may only be assigned with the role of plane or point lookup, whereas the primary prover may also be assigned with the role of the IPCP prover. This allows the verifier to enforce that only one prover takes the role of the IPCP prover, while keeping the provers' strategy symmetric.

Below we describe how to construct an \MIPStar{} for $\Language$.

\begin{construction}
\label{con:sumcheck-mipstar}
We construct a $2$-prover \MIPStar{} $(\Prover_1, \Prover_2, \Verifier)$ for the language $\Language$. The provers $\Prover_1$ and $\Prover_2$ share an \qstate{} $\ket{\Psi}$, and all three parties receive an explicit input $\Input$. The (honest) interaction takes place as follows.

\begin{enumerate}

	\item \emph{Symmetrization.} The provers flip a coin (by performing a measurement on $\ket{\Psi}$) and send its outcome to the verifier, to decide which prover is \emph{primary} and which one is \emph{secondary}. The primary prover may be assigned a role in $\Set{\ProverMain, \ProverPoint, \ProverPlane}$, and the secondary prover only in $\Set{\ProverPoint, \ProverPlane}$.
	
	\item The verifier chooses uniformly at random between the following procedures.
	
	\begin{itemize}
		\item \emph{Low individual degree test.} The verifier $\Verifier$ performs the low individual degree test of \cref{con:plane-vs-point-individual}. Recall that with probability $1/4$ the verifier sends a random point $\alpha \in \Field^{\SCVars}$ to the secondary prover.
		
		\item \emph{IPCP emulation.}
		\begin{enumerate}
		\item The verifier $\Verifier$ assigns the primary prover the role $\ProverMain$ and the secondary prover the role $\ProverPoint$.
		\item $\Verifier$ asks $\ProverPoint$ for an evaluation of $\RandPoly$ at a uniformly chosen point $\vec{\beta} \in \Field^\SCVars$.
  		\item $\ProverMain$ and $V$ emulate the interaction of the IPCP $(\Prover'',\Verifier'')$. This generates a value $c \in \Field$ such that, with probability at least $1 - \SoundnessError$, $\Input \in \Language$ if and only if $\RandPoly(\vec{\beta})=c$.
  		\item $\ProverPoint$ replies with an element $\Malicious{\Element} \in \Field$.
  		\item $\Verifier$ accepts if and only if $c = \Malicious{\Element}$.
  		\end{enumerate}
	\end{itemize} 

\end{enumerate}
\end{construction}

The honest prover strategy in \cref{con:sumcheck-mipstar} is symmetric and so we write $\ProverStrategy \DefineEqual \Prover_1 = \Prover_2$. The round complexity of the \MIPStar in \cref{con:sumcheck-mipstar} is $\RoundComplexity^* + 2$, because the parties assign roles in the first round, then either engage in a $1$-round low-degree test protocol or an $\RoundComplexity^*+ 1$-round IPCP protocol.\footnote{This relies on the IPCP verifier satisfying the prover-oblivious queries property.}

For completeness, if $\Input \in \Language$, then there exists a low-degree polynomial $\Oracle \in \PolynomialRingIndOne{\Field}{\SCVars}{\VariableX}{\SCDegree}$ that the IPCP verifier $\Verifier''$ accepts. Hence, after the primary prover is chosen, if the verifier selects the low-degree test, then both $\ProverPoint$ and $\ProverPlane$ can simply answer according to $\Oracle$, and if the verifier choses the IPCP emulation, then the prover given role $\ProverMain$ acts according to the strategy of $\Prover''$, and $\ProverPoint$ acts as a lookup for $\Oracle$. In the case of the low-degree test, the foregoing strategy is accepted with probability $1$, whereas in the IPCP emulation, the strategy inherits its completeness directly from the IPCP $(\Prover'',\Verifier'')$.

We next argue soundness (\cref{sec:main_soundness}) and preservation of zero knowledge (\cref{sec:main_zk}).

\parhead{Preserving round complexity sans zero knowledge}
As mentioned in \cref{rem:rounds}, the round complexity of the \MIPStar{} in \cref{con:sumcheck-mipstar} can be improved by $1$. This is achieved by letting the (honest) \emph{verifier} choose at random which prover is primary and which is secondary (replacing the first step in \cref{con:sumcheck-mipstar}).

While this modification may break the zero knowledge property (as it allows the verifier to engage in protocols that abuse the interaction with the prover, e.g., by allowing the verifier to set both provers as the main prover\footnote{Indeed, the particular zero knowledge IPCP (shown in \cref{part:II}) that we use to obtain our main result (\cref{thm:main-zk} see also \cref{rem:breaking-zk}) was observed to lose its zero knowledge property if the verifier is allowed to perform two parallel interactions with the prover (see \cite[Remark 5.6]{BenSassonCFGRS17}).}), this modification has essentially no effect on the soundness analysis, which we show next.

\subsection{Soundness analysis}
\label{sec:main_soundness}
We argue soundness against entangled quantum provers for the \MIPStar{} from \cref{con:sumcheck-mipstar}. Namely, we prove soundness with respect to a large constant soundness error $1 - \SoundnessError$, and then we amplify the soundness to the desired constant via parallel repetition for entangled strategies \cite{KempeV11,Yuen16}.\footnote{While the known parallel repetition theorems for entangled strategies are much weaker than their classical counterparts (having polynomial rather than exponential decay), this difference is immaterial in our setting.} Note that this preserves the complexities required by the conclusion of \cref{lem:lifting}.

Let $x \not\in \Language$. We may assume by \cref{lem:asymmetry} that, since the protocol is symmetric, the provers employ some symmetric strategy $(\Malicious{\ProverStrategy},\Malicious{\ProverStrategy})$, using a permutation invariant \qstate{} $\ket{\Psi} \in \HilbertSpace \otimes \HilbertSpace$. Suppose towards contradiction that the verifier accepts with probability at least $1 - \SoundnessError/2$, for constant $\SoundnessError$ to be determined later. We show that this implies a strategy that fools the (classical) IPCP with constant probability.

Let $\Set{\Measurement{A}{\Point}{\Element}}_{\Element \in \Field, \Point \in \Field^\SCVars}$ be the projective measurement describing the strategy $\Malicious{\ProverStrategy}$. The verifier $\Verifier$ with probability $1/2$ performs the low individual degree test of \cref{con:plane-vs-point-individual}. Since by assumption the verifier accepts with probability at least $1 - \SoundnessError/2$, the low-degree test passes with probability at least $1-\SoundnessError$. 

Therefore, by \cref{thm:quantum_low_degree_test}, there exists an absolute constant $c \in [0,1]$ and a measurement $\Set{\Measurement{L}{}{\Poly}}_{\Poly \in \PolynomialRingIndOne{\Field}{\SCVars}{\VariableX}{\SCDegree}}$ such that for $\ProximityParameter \DefineEqual \SoundnessError^c$ it holds that
\begin{equation}
\label{eq:point-strategy}
	\E_{\Point \in \Field^\SCVars} \sum_{\Poly \in \PolynomialRingIndOne{\Field}{\SCVars}{\VariableX}{\SCDegree}} \sum_{\substack{\Element \in \Field \\ \Element \neq \Poly(\Point)}} \bra{\Psi} \Measurement{A}{\Point}{\Element} \otimes \Measurement{L}{}{\Poly} \ket{\Psi} \leq \ProximityParameter \enspace,
\end{equation}
and
\begin{equation}
\label{eq:self-consistency}
	\sum_{\Poly \in \PolynomialRingIndOne{\Field}{\SCVars}{\VariableX}{\SCDegree}} \bra{\Psi} \Measurement{L}{}{\Poly} \otimes (\Id -\Measurement{L}{}{\Poly}) \ket{\Psi} \leq \ProximityParameter \enspace.
\end{equation}

Let $\Malicious{\ProverStrategyLD}$ be the strategy derived from $\Malicious{\ProverStrategy}$ by replacing the (arbitrary) measurement $\Set{\Measurement{A}{\Point}{\Element}}_{\Element \in \Field, \Point \in \Field^\SCVars}$ with the (low-degree) measurement $\Set{\Measurement{L}{\Point}{\Element}}_{\Element \in \Field,\Point \in \Field^\SCVars}$ given by
\begin{equation*}
\Measurement{L}{\Point}{\Element}
\DefineEqual
\sum_{\substack{\Poly \in \PolynomialRingIndOne{\Field}{\SCVars}{\VariableX}{\SCDegree} \\ \Poly(\Point) = \Element}}
  \Measurement{L}{}{\Poly}
\enspace.
\end{equation*}

Without loss of generality we designate the primary prover as $\Prover_{1}$ and the secondary prover as $\Prover_{2}$. With probability at least $1/2$, the verifier performs IPCP emulation; since the success probability of the malicious prover strategy $(\Malicious{\ProverStrategy},\Malicious{\ProverStrategy})$ is $1 - \SoundnessError/2$, this succeeds with probability at least $1 - \SoundnessError$. Recall that in this case $\Prover_{1}$ is given the role of $\ProverMain$ and $\Prover_{2}$ the role of $\ProverPoint$. In this setting, $\Prover_{1}$ acts identically under strategy $\Malicious{\ProverStrategy}$ and $\Malicious{\ProverStrategyLD}$, whereas $\Prover_{2}$ measures according to $\Set{\Measurement{A}{\Point}{\Element}}_{\Element \in \Field, \Point \in \Field^\SCVars}$ in the former case and $\Set{\Measurement{L}{\Point}{\Element}}_{\Element \in \Field,\Point \in \Field^\SCVars}$ in the latter. We show that the probability that $\Verifier$ (falsely) accepts the strategy $(\Malicious{\ProverStrategy},\Malicious{\ProverStrategy})$ is close, up to an additive $\ProximityParameter$ factor, to the probability it accepts the strategy $(\Malicious{\ProverStrategyLD},\Malicious{\ProverStrategyLD})$.

We describe the system via the following four registers.
\begin{enumerate}
	\item $\Register{A}$ is the (classical) register wherein the message from $\Verifier$ to $\ProverPoint$ is stored.
	\item $\Register{B}$ is the register that corresponds to the private space of $\ProverPoint$.
	\item $\Register{C}$ is the register that consists of the rest of the system (everything but $\Register{A}$, $\Register{B}$, and the ancilla).
	\item $\Register{D}$ is the ancilla $\ProverPoint$ uses to store its answers to $\Verifier$.

\end{enumerate}

Let $\State \in \HilbertSpace_{\Register{A}} \otimes \HilbertSpace_{\Register{B}} \otimes \HilbertSpace_{\Register{C}} \otimes \HilbertSpace_{\Register{D}}$ be the global \qstate{} of the system prior to the measurement performed by $\ProverPoint$, $\State_B$ be the global state of the system after $\ProverPoint$ measures according to $\set{\Measurement{A}{\Point}{b}}_{b \in \Field, \Point \in \Field^\SCVars}$, and $\State_T$ be the global state of the system after $\ProverPoint$ measures according to $\Set{\Measurement{L}{}{\Poly}}_{\Poly \in \PolynomialRingIndOne{\Field}{\SCVars}{\VariableX}{\SCDegree}}$, where after the measurements $\ProverPoint$ discards the post-measurement state. Note that
\begin{equation*}
	\State = \E_{\Point \in \Field^\SCVars}\ket{\Point}\bra{\Point} \otimes \State_{\Point}^{\Register{B}, \Register{C}} \enspace,
\end{equation*} 
\begin{equation*}
	\State_{B} = \Trace_{\Register{B}} \left[
	\E_{\Point \in \Field^\SCVars}\ket{\Point}\bra{\Point} \otimes
	(\Measurement{B}{\Point}{b} \otimes \Id_{\Register{C}}) \State_{\Point}^{\Register{B}, \Register{C}} (\Measurement{B}{\Point}{b} \otimes \Id_{\Register{C}}) \otimes
	\sum_{b \in \Field}\ket{b}\bra{b} \right] \enspace,
\end{equation*} 
\begin{equation*}
	\State_{T} = \Trace_{\Register{B}} \left[
	\E_{\Point \in \Field^\SCVars}\ket{\Point}\bra{\Point} \otimes
	(\Measurement{L}{\Point}{b} \otimes \Id_{\Register{C}}) \State_{\Point}^{\Register{B}, \Register{C}} (\Measurement{L}{\Point}{b} \otimes \Id_{\Register{C}}) \otimes
	\sum_{b \in \Field}\ket{b}\bra{b} \right] \enspace,
\end{equation*}
where $\State_{\Point}^{\Register{B}, \Register{C}}$ denotes the \qstate{} $\State$ after $\Verifier$ asked the question $\Point$, restricted to the registers $\Register{B}, \Register{C}, \Register{D}$, and $\Id_{\Register{C}}$ denotes the identity operator over $\HilbertSpace_{\Register{D}}$.

Recall that $(\Malicious{\ProverStrategy}, \Malicious{\ProverStrategy}, \Verifier)(x)$ and $(\Malicious{\ProverStrategyLD},\Malicious{\ProverStrategyLD},\Verifier)(x)$ denote the random variables representing the output of the verifier $\Verifier$ when interacting with provers employing strategies $(\Malicious{\ProverStrategy},\Malicious{\ProverStrategy})$ and $(\Malicious{\ProverStrategyLD},\Malicious{\ProverStrategyLD})$, respectively, on input $x \notin \Language$. Observe that
\begin{align*}
	 & \left| \Pr[(\Malicious{\ProverStrategy},\Malicious{\ProverStrategy},\Verifier)(x) = 1] - \Pr[(\Malicious{\ProverStrategyLD},\Malicious{\ProverStrategyLD},\Verifier)(x) = 1] \right| \\
	 & \leq \frac{1}{2} \NormOne{\State_{B} - \State_{T}} \\
	 & = \frac{1}{2} \NormOne{\Trace_{\Register{B}}
	\E_{\Point \in \Field^\SCVars}\ket{\Point}\bra{\Point} \otimes
	\left( (\Measurement{B}{\Point}{b} \otimes \Id_{\Register{C}}) \State_{\Point}^{\Register{B},\Register{C}} (\Measurement{B}{\Point}{b} \otimes \Id_{\Register{C}}) - (\Measurement{L}{\Point}{b} \otimes \Id_{\Register{C}}) \State_{\Point}^{\Register{B},\Register{C}} (\Measurement{L}{\Point}{b} \otimes \Id_{\Register{C}}) \right) \otimes \sum_{b \in \Field}\ket{b}\bra{b} } \enspace, \\
	 & \leq \frac{1}{2}\E_{\Point \in \Field^\SCVars} \NormOne{\sum_{b \in \Field} \left( (\Measurement{B}{\Point}{b} \otimes \Id_{\Register{C}}) \State_{\Point}^{\Register{B},\Register{C}} (\Measurement{B}{\Point}{b} \otimes \Id_{\Register{C}}) - (\Measurement{L}{\Point}{b} \otimes \Id_{\Register{C}}) \State_{\Point}^{\Register{B},\Register{C}} (\Measurement{L}{\Point}{b} \otimes \Id_{\Register{C}}) \right) \otimes \ket{b}\bra{b}} \\
	 & \leq \E_{\Point \in \Field^\SCVars} \sqrt{\sum_{b \in \Field} \Trace \left( (\Measurement{B}{\Point}{b} - \Measurement{L}{\Point}{b}) \rho (\Measurement{B}{\Point}{b} - \Measurement{L}{\Point}{b})^\dagger \right)} \quad\quad \text{(by \cref{lem:gentle-measurement})} \\
	 & \leq \sqrt{ \E_{\Point \in \Field^\SCVars} \sum_{b \in \Field} \TraceRho{(\Measurement{B}{\Point}{b} - \Measurement{L}{\Point}{b})^2}} \quad\quad \text{(by Jensen's inequality)} \\
	 & \leq \delta \quad\quad \text{(by \cref{eq:point-strategy} and \cref{clm:consistncy-to-trace}).}
\end{align*}

By the triangle inequality, the total success probability of $(\Malicious{\ProverStrategyLD},\Malicious{\ProverStrategyLD})$ is at least $1 - \SoundnessError - \ProximityParameter > 1-2\ProximityParameter$. To conclude the proof, note that when using strategy $(\Malicious{\ProverStrategyLD},\Malicious{\ProverStrategyLD})$ the prover $\ProverPoint$ can measure the prior entanglement obliviously to its question, and so the strategy can be realize merely via shared randomness.

Therefore we can construct a malicious prover $\Malicious{\Prover}''$ that will fool the \LDIPCP{} verifier $\Verifier''$ for $\Language$ (which we started from) with probability at least $1-2\ProximityParameter$, by implementing the strategy $\Malicious{\ProverStrategy}''$ in the natural way. That is, $\Malicious{\Prover}''$ samples some $\Poly \in \PolynomialRingIndOne{\Field}{\SCVars}{\VariableX}{\SCDegree}$ according to the distribution induced by $\Set{\Measurement{L}{}{\Poly}}$ and $\ket{\Psi}$ and sends $\Poly$ as the oracle. It then interacts with $\Verifier$ according to the strategy $\Malicious{\ProverStrategy}$.

\subsection{Preserving zero knowledge}
\label{sec:main_zk}

We argue that \cref{con:sumcheck-mipstar} preserves zero knowledge. Suppose that the IPCP $(\Prover'',\Verifier'')$ is zero knowledge with query bound $\QueryBound' \geq 2(\SCDegree+1)^{2}$, and let $\Simulator''$ be the corresponding simulator. We explain how to construct a simulator $\Simulator$ for the \MIPStar{} $(\Prover_1, \Prover_2, \Verifier)$.

Given a malicious verifier $\Malicious{\Verifier}$ for the \MIPStar{}, we design a \DoQuote{malicious} verifier $\Malicious{\Verifier}''$ for the IPCP protocol that, when interacting with the IPCP prover $\Prover''$, outputs the view of the malicious \MIPStar{} verifier $\Malicious{\Verifier}$ when interacting with the (honest) \MIPStar{} provers $\Prover_{1}, \Prover_{2}$. The simulator $\Simulator$ is then given by running $(\Simulator'')^{\Malicious{\Verifier}''}$, and returning the output of $\Malicious{\Verifier}''$. We first describe the operation of $\Malicious{\Verifier}''$.

\begin{mdframed}
\begin{enumerate}[nolistsep]
\item Begin simulating $\Malicious{\Verifier}$.
\item Flip a coin, and send the outcome to $\Malicious{\Verifier}$. Receive from $\Malicious{\Verifier}$ the role assignments for the provers; if the secondary prover is assigned to be $\ProverMain$, we simulate as if it has aborted.
\item The remainder of the simulation is divided up with respect to prover role.
\begin{enumerate}[nolistsep]
	\item Every message $\Malicious{\Verifier}$ sends to the prover assigned to be $\ProverMain$, if any, is forwarded to $\Prover''$, and the responses of $\Prover''$ are forwarded to $\Malicious{\Verifier}$.
	\item If any prover is assigned the role of $\ProverPoint$, then if $\Malicious{\Verifier}$ sends a query point $\alpha \in \F^{\SCVars}$ to this prover, query the oracle at $\alpha$, and send the answer to $\Malicious{\Verifier}$; if $\Malicious{\Verifier}$ sends an axis-parallel line $\ell$ to this prover, query the oracle at $(\SCDegree+1)$ points on $\ell$ in order to interpolate $\RandPoly \circ \ell$, and send this polynomial to $\Malicious{\Verifier}$.
	\item If any prover is assigned the role of $\ProverPlane$, then if $\Malicious{\Verifier}$ sends a query plane $\Plane \in \Planes{\Field^\SCVars}$ to this prover, query the oracle at a set of points sufficient to interpolate $\RandPoly \circ \Line$ (of size at most $(\SCDegree+1)^{2}$), and send this polynomial to $\Malicious{\Verifier}$.
\end{enumerate}
\item Output the view of the simulated $\Malicious{\Verifier}$.
\end{enumerate}
\end{mdframed}

It is clear that the output of $\Malicious{\Verifier}''$ when interacting with $\Prover''$ is exactly the view of $\Malicious{\Verifier}$ in the \MIPStar{}. The number of queries $\Malicious{\Verifier}''$ makes is at most $(\SCDegree+1)^{2} + \SCDegree + 1 \leq 2(\SCDegree+1)^{2}$. By the zero knowledge guarantee for $\Simulator''$, provided $\QueryBound' \geq 2(\SCDegree+1)^{2}$, the view of $\Malicious{\Verifier}''$ is perfectly simulated, and so in particular its output in simulation is identically distributed to the output of $\Malicious{\Verifier}''$ when interacting with $\Prover''$.

\begin{remark}
\label{rem:breaking-zk}
Observe that if it were possible for the verifier to assign both provers to be $\ProverMain$, as is the case with the standard symmetrization, then the above argument would not go through. The reason is that the two interactions may be correlated in some way that we cannot simulate.
\end{remark}

\doclearpage
\section{Zero knowledge \MIPStar{} for nondeterministic exponential time}
\label{sec:putting-it-all-together}

Recall that our plan is to prove \cref{thm:main-zk} in two steps:
\begin{inparaenum}[(1)]
	\item construct a zero knowledge low-degree IPCP for any language in $\NEXP$;
	\item invoke the lifting lemma (\cref{lem:lifting}) on this low-degree IPCP in order to obtain a zero knowledge \MIPStar{} for $\NEXP$.
\end{inparaenum}
So far, in \cref{part:I}, we have obtained the tools for deriving a zero knowledge \MIPStar{} from a zero knowledge low-degree IPCP. The goal of \cref{part:II} is to construct such a zero knowledge low-degree IPCP for any language in $\NEXP$; that is, in \cref{part:II} we prove the following theorem.

\begin{theorem}[concisely stated; see \cref{thm:pzk-for-nexp} for the full statement]
\label{thm:pzk-ldipcp}
There exists a constant $c \in \Naturals$ such that for every query bound function $\QueryBound$ and language $\Language \in \NEXP$ the following holds. There exists an \LDIPCP{} for $\Language$, where $d, m = O(n^{c} \log \QueryBound)$ and $\Field$ is a field with $|\Field| = \Omega((n^{c} \log \QueryBound)^{4})$, that is perfect zero knowledge against all $\QueryBound$-query malicious verifiers and has the following parameters:
	\begin{equation*}
		\LDIPCPparams{1/2}{\poly(n) + O(\log \QueryBound)}{\poly(2^{n},\QueryBound)}{\poly(n, \log(\QueryBound))}{\poly(n, \log(\QueryBound))}{\PolynomialRingIndOne{\Field}{\SCVars}{\VariableX}{\SCDegree}} \enspace.
	\end{equation*}
\end{theorem}


In this section we prove \cref{thm:main-zk} by taking the zero knowledge low-degree IPCP in \cref{thm:pzk-ldipcp} and lifting it via \cref{lem:lifting} to obtain a $2$-prover zero knowledge \MIPStar{}, concluding the proof of \cref{thm:main-zk}.

Let $\Language$ be a language in $\NEXP$, and let $\pair{\Prover'}{\Verifier'}$ be the \LDIPCP{} for $\Language$ implied by \cref{thm:pzk-ldipcp}, with respect to query bound $\QueryBound$, $d, m = O(n^{c} \log \QueryBound)$, and a finite field $\Field$ of size $\SetCardinality{\Field} = \poly(n)$ such that $\SetCardinality{\Field} > \max\Set{(2\SCVars \SCDegree)^C, 5\SCDegree\QueryComplexity, \Omega((n^{c} \log \QueryBound)^{4})}$ (where $C$ is the constant from \cref{thm:quantum_low_degree_test}).

The \LDIPCP{} $\pair{\Prover'}{\Verifier'}$ satisfies the conditions of the lifting lemma (\cref{lem:lifting}), and thus we can apply the transformation $\Transformation$ in \cref{lem:lifting} to the low-degree IPCP $\pair{\Prover'}{\Verifier'}$ to obtain a perfect zero knowledge $2$-prover \MIPStar{} $(\Prover_1, \Prover_2, \Verifier) \DefineEqual \Transformation(\Prover',\Verifier')$ for $\Language$, with round complexity $\poly(n) + O(\log \QueryBound) = \poly(n)$, communication complexity $\poly(n, \log(\QueryBound), \SCDegree, \QueryComplexity, \SCVars) = \poly(n)$, and soundness error $1/2$. This concludes the proof of our main result, \cref{thm:main-zk}.


\begin{remark}[zero knowledge \MIPStar{} for $\sharpP$ via known IPCPs]
\label{remark:sharpp-mip}
As mentioned in \cref{sec:techniques-towards}, a recent work in algebraic zero knowledge \cite{BenSassonCFGRS17} (building on techniques from \cite{BenSassonCGV16}) obtains a zero knowledge low-degree IPCP for any language in $\sharpP$. By replacing our IPCP for $\NEXP$ in \cref{thm:pzk-ldipcp} with their IPCP for $\sharpP$, we can derive a zero knowledge \MIPStar{}, albeit only for languages in $\sharpP$. 
\end{remark}

\doclearpage
\part{Low-degree IPCP with zero knowledge}
\label{part:II}

The purpose of this part is to show that there exists a perfect zero knowledge low-degree IPCP for any language in $\NEXP$, which is the remaining step in our construction of perfect zero knowledge \MIPStar{} protocols for $\NEXP$ (as discussed in \cref{sec:putting-it-all-together}). To this end we build on advances in algebraic zero knowledge \cite{BenSassonCFGRS17} and ideas from algebraic complexity theory, and we develop new techniques for obtaining algebraic zero knowledge.

\parhead{Organization}
We begin in \cref{sec:algebraic-query-complexity}, where we show new algebraic query complexity lower bounds on polynomial summation. Then, in \cref{sec:strong-zk-sumcheck} we construct our strong zero knowledge sumcheck protocol, whose analysis relies on the foregoing algebraic query complexity lower bounds. Finally, in \cref{sec:zk-nexp} we use our strong zero knowledge sumcheck protocol to show a perfect zero knowledge low-degree IPCP for any language in $\NEXP$.

\section{Algebraic query complexity of polynomial summation}
\label{sec:algebraic-query-complexity}

We have outlined in \cref{sec:techniques-algebraic-commitment} an algebraic commitment scheme based on the sumcheck protocol and lower bounds on the algebraic query complexity of polynomial summation. The purpose of this section is to describe this construction in more detail, and then provide formal statements for the necessary lower bounds.

\parhead{The setting: algebraic commitment schemes}
We begin with the case of committing to a single element $a \in \Field$. The prover chooses a uniformly random string $B \in \Field^{N}$ such that $\sum_{i=1}^{N} B_{i} = a$, for some $N \in \Naturals$. Fixing some $\SCDegree \in \Naturals$, $\SSCSubset \subseteq \Field$ and $\SSCVars \in \Naturals$ such that $\SetCardinality{\SSCSubset} \leq \SCDegree+1$ and $\SetCardinality{\SSCSubset}^{\SSCVars} = N$, the prover views $B$ as a function from $\SSCSubset^{\SSCVars}$ to $\Field$ (via an arbitrary ordering on $\SSCSubset^{\SSCVars}$) and sends the evaluation of a degree-$d$ extension $\LD{B} \colon \Field^{\SSCVars} \to \Field$ of $B$, chosen uniformly at random from all such extensions. The verifier can test that $\LD{B}$ is indeed (close to) a low-degree polynomial but (ideally) cannot learn any information about $a$ without reading \emph{all} of $B$ (i.e., without making $N$ queries). Subsequently, the prover can decommit to $a$ by convincing the verifier that $\sum_{\vec{\beta} \in \SSCSubset^{\SSCVars}} \LD{B}(\vec{\beta}) = a$ via the sumcheck protocol.

To show that the above is a commitment scheme, we must show both \emph{binding} and \emph{hiding}. Both properties depend on the choice of $\SCDegree$. The binding property follows from the soundness of the sumcheck protocol, and we thus would like the degree $d$ of $\LD{B}$ to be as small as possible. A natural choice would be $\SCDegree = 1$ (so $\SetCardinality{\SSCSubset} = 2$), which makes $\LD{B}$ the unique multilinear extension of $B$. However (as discussed in \cref{sec:techniques-algebraic-commitment}) this choice of parameters does not provide any hiding: it holds that $\sum_{\beta \in \Bits^{k}} B(\beta) = \LD{B}(2^{-1}, \ldots, 2^{-1}) \cdot 2^{k}$ (as long as $\Characteristic{\Field} \neq 2$). We therefore need to understand how the choice of $d$ affects the number of queries to $\LD{B}$ required to compute $a$. This is precisely the setting of \emph{algebraic query complexity}, which we discuss next.

The algebraic query complexity (defined in \cite{AaronsonW09} to study \DoQuote{algebrization}) of a function $f$ is the (worst-case) number of queries to some low-degree extension $\LD{B}$ of a string $B$ required to compute $f(B)$. This quantity is bounded from above by the standard query complexity of $f$, but it may be the case (as above) that the low-degree extension confers additional information that helps in computing $f$ with fewer queries. The usefulness of this information depends on the parameters $\SCDegree$ and $\SSCSubset$ of the low-degree extension. Our question amounts to understanding this dependence for the function $\textsc{Sum} \colon \Field^{N} \to \Field$ given by $\textsc{Sum}(B) \DefineEqual \sum_{i=1}^{N} B_{i}$. It is known that if $\SSCSubset = \Bits$ and $\SCDegree = 2$ then the algebraic query complexity of $\textsc{Sum}$ is exactly $N$ \cite{JumaKRS09}.

For our purposes, however, it is not enough to commit to a single field element. Rather, we need to commit to the evaluation of a polynomial $Q \colon \Field^{\SCVars} \to \Field$ of degree $\SCDegree_{Q}$, which we do as follows. Let $K$ be a subset of $\Field$ of size $\SCDegree_{Q}+1$. The prover samples, for each $\vec{\alpha} \in K^{\SCVars}$, a random string $B^{\vec{\alpha}} \in \Field^{N}$ such that $\textsc{Sum}(B^{\vec{\alpha}}) = Q(\vec{\alpha})$. The prover views these strings as a function $B \colon K^{\SCVars} \times \SSCSubset^{\SSCVars} \to \Field$, and takes a low-degree extension $\LD{B} \colon \Field^{\SCVars} \times \Field^{\SSCVars} \to \Field$. The polynomial $\LD{B}(\vec{\VariableX}, \vec{\VariableY})$ has degree $d_{Q}$ in $\vec{\VariableX}$ and $d$ in $\vec{\VariableY}$; this is a commitment to $Q$ because $\sum_{\vec{\beta} \in \SSCSubset^{\SSCVars}} \LD{B}(\vec{\VariableX}, \vec{\beta})$ is a degree-$d_{Q}$ polynomial that agrees with $Q$ on $K^{\SCVars}$, and hence equals $Q$.

Once again we will decommit to $Q(\vec{\alpha})$ using the sumcheck protocol, and so for binding we need $d$ to be small. For hiding, as in the single-element case, if $d$ is too small, then a few queries to $\LD{B}$ can yield information about $Q$. Moreover, it could be the case that the verifier can leverage the fact that $\LD{B}$ is a \emph{joint} low-degree extension to learn some linear combination of evaluations of $Q$. We must exclude these possibilities in order to obtain our zero knowledge guarantees.

\parhead{New algebraic query complexity lower bounds}
The foregoing question amounts to a generalization of algebraic query complexity where, given a list of strings $B_{1}, \ldots, B_{M}$, we determine how many queries we need to make to their \emph{joint} low-degree extension $\LD{B}$ to determine any nontrivial linear combination $\sum_{i=1}^{M} c_{i} \cdot \textsc{Sum}(B_{i})$. We will show that the \DoQuote{generalized} algebraic query complexity of $\textsc{Sum}$ is exactly $N$, provided $\SCDegree \geq 2(\SetCardinality{\SSCSubset}-1)$ (which is also the case for the standard algebraic query complexity).

In the remainder of the section we state our results in a form equivalent to the above, which is more useful to us. Denote by $\PolynomialRingIndOneXY{\Field}{\SCVars}{\VariableX}{\SSCVars}{\VariableY}{\SCDegree}{\SCDegree'}$ the set of all $(m+k)$-variate polynomials of individual degree $\SCDegree$ in the variables $\VariableX_1,\ldots, \VariableX_m$ and individual degree $\SCDegree'$ in the variables $\VariableY_1,\ldots, \VariableY_k$.
Given an arbitrary polynomial $\StrongRandPoly \in \PolynomialRingIndOneXY{\Field}{\SCVars}{\VariableX}{\SSCVars}{\VariableY}{\SCDegree}{\SCDegree'}$, we ask how many queries are required to determine any nontrivial linear combination of $\sum_{\vec{y} \in \SSCSubset^{\SSCVars}} \StrongRandPoly(\vec{\alpha}, \vec{y})$ for $\vec{\alpha} \in \Field^{\SCVars}$. The following lemma is more general: it states that not only do we require many queries to determine \emph{any} linear combination, but that the number of queries grows linearly with the number of independent combinations that we wish to learn.

\begin{lemma}[algebraic query complexity of polynomial summation]
\label{lem:sum-query-lower-bound}
Let $\Field$ be a field, $\SCVars, \SSCVars, \SCDegree, \SCDegree' \in \Naturals$, and $\SSCSubset, K, L$ be finite subsets of $\Field$ such that $K \subseteq L$, $\SCDegree' \geq \SetCardinality{\SSCSubset} - 2$, and $\SetCardinality{K} = \SCDegree+1$. If $S \subseteq \Field^{\SCVars+\SSCVars}$ is such that there exist matrices $C \in \Field^{L^{m} \times \ell}$ and $D \in \Field^{S \times \ell}$ such that for all $\StrongRandPoly \in \PolynomialRingIndOneXY{\Field}{\SCVars}{\VariableX}{\SSCVars}{\VariableY}{\SCDegree}{\SCDegree'}$ and all $i \in \{1, \ldots, \ell\}$
	\begin{equation*}
	\sum_{\vec{\alpha} \in L^{\SCVars}} C_{\vec{\alpha},i} \sum_{\vec{y} \in \SSCSubset^{\SSCVars}} \StrongRandPoly(\vec{\alpha}, \vec{y}) = \sum_{\vec{q} \in S} D_{\vec{q},i} \StrongRandPoly(\vec{q}) \enspace,
	\end{equation*}
	then $\SetCardinality{S} \geq \rank(BC) \cdot (\min\{\SCDegree' - \SetCardinality{\SSCSubset} + 2, \SetCardinality{\SSCSubset}\})^{\SSCVars}$, where $B \in \Field^{K^{\SCVars} \times L^{\SCVars}}$ is such that column $\vec{\alpha}$ of $B$ represents $\StrongRandPoly(\vec{\alpha})$ in the basis $(\StrongRandPoly(\vec{\beta}))_{\vec{\beta} \in K^{\SCVars}}$.
\end{lemma}

We remark that in \cref{sec:query-complexity-appendix} we prove upper bounds showing that, in some cases, \cref{lem:sum-query-lower-bound} is tight.

\begin{proof}[Proof of \cref{lem:sum-query-lower-bound}]
	We use a rank argument. First, since $\StrongRandPoly$ has individual degree at most $\SCDegree$ in $\vec{\VariableX}$, we can rewrite any such linear combination in the following way:
	\begin{equation*}
	\sum_{\vec{\alpha} \in L^{\SCVars}} C_{\vec{\alpha},i} \sum_{\vec{y} \in \SSCSubset^{\SSCVars}} \StrongRandPoly(\vec{\alpha}, \vec{y})
	= \sum_{\vec{\alpha} \in L^{\SCVars}} C_{\vec{\alpha},i} \sum_{\vec{\beta} \in K^{\SCVars}} b_{\vec{\beta},\vec{\alpha}} \sum_{\vec{y} \in \SSCSubset^{\SSCVars}} \StrongRandPoly(\vec{\alpha}, \vec{y})
	= \sum_{\vec{\alpha} \in K^{\SCVars}} C'_{\vec{\alpha},i} \sum_{\vec{y} \in \SSCSubset^{\SSCVars}} \StrongRandPoly(\vec{\alpha}, \vec{y})
	= \sum_{\vec{q} \in S} D_{\vec{q},i} \StrongRandPoly(\vec{q}) \enspace,
	\end{equation*}
	where $C' \DefineEqual BC$. If $\SCDegree' = \SetCardinality{\SSCSubset} - 2$, then the bound is trivial. Otherwise, let $H$ be some arbitrary subset of $G$ of size $\min\{\SCDegree' - \SetCardinality{\SSCSubset} + 2, \SetCardinality{\SSCSubset}\}$. Let $P_{0} \subseteq \PolynomialRingIndOneXY{\Field}{\SCVars}{\VariableX}{\SSCVars}{\VariableY}{\SCDegree}{\SetCardinality{\SCSubset}-1}$ be such that for all $p \in P_{0}$ and for all $\vec{q} \in S$, $p(\vec{q}) = 0$. Since these are at most $S$ linear constraints, $P_{0}$ has dimension at least $(\SCDegree+1)^{\SCVars}\SetCardinality{\SCSubset}^{\SSCVars}-\SetCardinality{S}$.
	
	Let $B_{0} \in \Field^{n \times (\SCDegree+1)^{\SCVars} \SetCardinality{\SCSubset}^{\SSCVars}}$ be a matrix whose rows form a basis for the vector space $\{\big(p(\vec{\alpha},\vec{\beta})\big)_{\vec{\alpha} \in K^{\SCVars}, \vec{\beta} \in \SCSubset^{\SSCVars}} : p \in P_{0}\}$ of evaluations of polynomials in $P_{0}$ on $K^{\SCVars} \times \SCSubset^{\SSCVars}$; we have $n \geq (\SCDegree+1)^{\SCVars}\SetCardinality{\SCSubset}^{\SSCVars}-\SetCardinality{S}$. By an averaging argument there exists $\vec{\beta}_{0} \in \SCSubset^{k}$ such that the submatrix $B_{\vec{\beta}_{0}}$ consisting of columns $(\vec{\alpha},\vec{\beta}_{0})$ of $B_{0}$ for each $\vec{\alpha} \in K^{\SCVars}$ has rank at least $(\SCDegree+1)^{\SCVars}-\SetCardinality{S}/\SetCardinality{\SCSubset}^{\SSCVars}$.
	
	Let $q \in \PolynomialRingIndOne{\Field}{\SSCVars}{\VariableY}{\SetCardinality{\SSCSubset}-1}$ be the polynomial such that $q(\vec{\beta}_{0}) = 1$, and $q(\vec{y}) = 0$ for all $\vec{y} \in \SSCSubset^{\SSCVars} - \{\vec{\beta}_{0}\}$. For arbitrary $p \in P_{0}$, let $Z(\vec{\VariableX}, \vec{\VariableY}) \DefineEqual q(\vec{\VariableY}) p(\vec{\VariableX},\vec{\VariableY}) \in \PolynomialRingIndOneXY{\Field}{\SCVars}{\VariableX}{\SSCVars}{\VariableY}{\SCDegree}{\SetCardinality{\SCSubset}+\SetCardinality{\SSCSubset}-2}$. Observe that our choice of $\SCSubset$ ensures that the degree of $\StrongRandPoly$ in $\vec{\VariableY}$ is at most $\SCDegree'$. Then for all $i \in \{1, \ldots, \ell\}$, it holds that
	\begin{equation*}
	\sum_{\vec{\alpha} \in K^{\SCVars}} C'_{\vec{\alpha},i} \sum_{\vec{y} \in \SSCSubset^{\SSCVars}} Z(\vec{\alpha},\vec{y})
	= \sum_{\vec{\alpha} \in K^{\SCVars}} C'_{\vec{\alpha},i} \cdot p(\vec{\alpha}, \vec{\beta}_{0})
	= \sum_{\vec{q} \in S} D_{\vec{q},i} \cdot Z(\vec{\alpha},\vec{y}) = 0 \enspace.
	\end{equation*}
	Thus the column space of $C'$ is contained in the null space of $B_{\vec{\beta}_{0}}$, and so the null space of $B_{\vec{\beta}_{0}}$ has rank at least $\rank(C')$. Hence $(d+1)^{m} - \rank(C') \geq \rank(B_{\vec{\beta}_{0}}) \geq (d+1)^{m} - \SetCardinality{S}/\SetCardinality{\SCSubset}^{\SSCVars}$, so $\SetCardinality{S} \geq \rank(C') \cdot \SetCardinality{\SCSubset}^{\SSCVars}$, which yields the theorem.
\end{proof}

\parhead{Implications} 
We state below special cases of \cref{lem:sum-query-lower-bound} that suffice for our zero knowledge applications.

\begin{corollary}
\label{cor:special-case-query}
Let $\Field$ be a finite field, $\SSCSubset$ be a subset of $\Field$, and $\SCDegree,\SCDegree' \in \Naturals$ with $\SCDegree' \geq 2(\SetCardinality{\SSCSubset}-1)$. If $S \subseteq \Field^{\SCVars+\SSCVars}$ is such that there exist $(c_{\vec{\alpha}})_{\vec{\alpha} \in \Field^{\SCVars}}$ and $(d_{\vec{\beta}})_{\vec{\beta} \in \Field^{\SCVars+\SSCVars}}$ such that
\begin{itemize}[nolistsep]

  \item for all $\StrongRandPoly \in \PolynomialRingIndOneXY{\Field}{\SCVars}{\VariableX}{\SSCVars}{\VariableY}{\SCDegree}{\SCDegree'}$ it holds that
$
\sum_{\vec{\alpha} \in \Field^{\SCVars}} c_{\vec{\alpha}} \sum_{\vec{y} \in \SSCSubset^{\SSCVars}} \StrongRandPoly(\vec{\alpha}, \vec{y})
= \sum_{\vec{q} \in S} d_{\vec{q}} \StrongRandPoly(\vec{q})
$
,and

  \item there exists $\StrongRandPoly' \in \PolynomialRingIndOneXY{\Field}{\SCVars}{\VariableX}{\SSCVars}{\VariableY}{\SCDegree}{\SCDegree'}$ such that $\sum_{\vec{\alpha} \in \Field^{\SCVars}} c_{\vec{\alpha}} \sum_{\vec{y} \in \SSCSubset^{\SSCVars}} \StrongRandPoly'(\vec{\alpha}, \vec{y}) \neq 0$,

\end{itemize}  
then $\SetCardinality{S} \geq \SetCardinality{\SSCSubset}^{\SSCVars}$.
\end{corollary}

Next, we give an equivalent formulation of \cref{cor:special-case-query} in terms of random variables that we use in later sections. (Essentially, the linear structure of the problem implies that \DoQuote{worst-case} statements are equivalent to \DoQuote{average-case} statements.)

\begin{corollary}[equivalent statement of \cref{cor:special-case-query}]
\label{cor:partial-sum-indep-vars}
Let $\Field$ be a finite field, $\SSCSubset$ be a subset of $\Field$, and $\SCDegree,\SCDegree' \in \Naturals$ with $\SCDegree' \geq 2(\SetCardinality{\SSCSubset}-1)$. Let $Q$ be a subset of $\Field^{\SCVars+\SSCVars}$ with $\SetCardinality{Q} < \SetCardinality{\SSCSubset}^{\SSCVars}$, and let $\StrongRandPoly$ be uniformly random in $ \PolynomialRingIndOneXY{\Field}{\SCVars}{\VariableX}{\SSCVars}{\VariableY}{\SCDegree}{\SCDegree'}$. Then, the ensembles $\big(\sum_{\vec{y} \in \SSCSubset^{\SSCVars}} Z(\vec{\alpha},\vec{y})\big)_{\vec{\alpha} \in \Field^{\SCVars}}$ and $\big(\StrongRandPoly(\vec{q})\big)_{\vec{q} \in Q}$ are independent.
\end{corollary}

\begin{proof}[Proof of \cref{cor:partial-sum-indep-vars}]
We will need a simple fact from linear algebra: that \DoQuote{linear independence equals statistical independence}. That is, if we sample an element from a vector space and examine some subsets of its entries, these distributions are independent if and only if there does not exist a linear dependence between the induced subspaces. The formal statement of the claim is as follows.
	
\begin{claim}
\label{claim:linear-algebra-fact}
Let $\Field$ be a finite field and $\Domain$ a finite set. Let $V \subseteq \Field^{\Domain}$ be an $\Field$-vector space, and let $\vec{v}$ be a random variable that is uniform over $V$. For any subdomains $S, S' \subseteq \Domain$, the restrictions $\Restrict{\vec{v}}{S}$ and $\Restrict{\vec{v}}{S'}$ are statistically dependent if and only if there exist constants $(c_{i})_{i \in S}$ and $(d_{i})_{i \in S'}$ such that:
\begin{itemize}[nolistsep]
  \item there exists $\vec{w} \in V$ such that $\sum_{i \in S} c_{i} w_{i} \neq 0$, and
  \item for all $\vec{w} \in V$, $\sum_{i \in S} c_{i} w_{i} = \sum_{i \in S'} d_{i} w_{i}$.
\end{itemize}
\end{claim}

\begin{proof}[Proof of \cref{claim:linear-algebra-fact}]
	For arbitrary $\vec{x} \in \Field^{S}, \vec{x}' \in \Field^{S'}$, we define the quantity
	\begin{equation*}
	p_{\vec{x},\vec{x}'} \DefineEqual \Pr_{\vec{v} \in V}
	\left[
	\Restrict{\vec{v}}{S} = \vec{x} \wedge \Restrict{\vec{v}}{S'} = \vec{x}'
	\right] \enspace.
	\end{equation*}
	Let $d \DefineEqual \dim(V)$, and let $B \in \Field^{D \times d}$ be a basis for $V$. Let $B_{S} \in \Field^{S \times d}$ be $B$ restricted to rows corresponding to elements of $S$, and let $B_{S'}$ be defined likewise. Finally, let $B_{S,S'} \in \Field^{(\SetCardinality{S}+\SetCardinality{S'})\times d}$ be the matrix whose rows are the rows of $B_{S}$, followed by the rows of $B_{S'}$. Then
	\begin{equation*}
	p_{\vec{x},\vec{x}'} = \Pr_{\vec{z} \in \Field^{d}}
	\left[
	B_{S,S'} \cdot \vec{z} = (\vec{x}, \vec{x}')
	\right] \enspace.
	\end{equation*}
	One can verify that, for any matrix $A \in \Field^{m \times n}$,
	\begin{equation*}
	\Pr_{\vec{z} \in \Field^{n}} [A\vec{z} = \vec{b}] = \begin{cases}
	\Field^{-\rank(A)} & \text{ if $\vec{b} \in \mathrm{colsp}(A)$, and} \\
	0 & \text{otherwise.}
	\end{cases}
	\end{equation*}
	Observe that $\mathrm{colsp}(B_{S,S'}) \subseteq \mathrm{colsp}(B_{S}) \times \mathrm{colsp}(B_{S'})$, and equality holds if and only if $\rank(B_{S,S'}) = \rank(B_{S}) + \rank(B_{S'})$. It follows that $p_{\vec{x},\vec{x}'} = \Pr_{\vec{v} \in V}[\Restrict{\vec{v}}{S} = \vec{x}] \cdot \Pr_{\vec{v} \in V}[\Restrict{\vec{v}}{S'} = \vec{x}']$ if and only if $\rank(B_{S,S'}) = \rank(B_{S}) + \rank(B_{S'})$. By the rank-nullity theorem and the construction of $B_{S,S'}$, this latter condition holds if and only if $\mathrm{nul}(B_{S,S'}^{T}) \subseteq \mathrm{nul}(B_{S}^{T}) \times \mathrm{nul}(B_{S'}^{T})$. To conclude the proof, it remains only to observe that the condition in the claim is equivalent to the existence of vectors $\vec{c} \in \Field^{S}$, $\vec{d} \in \Field^{S'}$ such that $\vec{c} \notin \mathrm{nul}(B_{S}^{T})$ but $(\vec{c},-\vec{d}) \in \mathrm{nul}(B_{S,S'}^{T})$.
\end{proof}

Now, observe that 
\begin{equation*}
\Big\{ \Big(\big(\StrongRandPoly(\vec{\gamma})\big)_{\vec{\gamma} \in \Field^{\SCVars+\SSCVars}}, \big(\sum_{\vec{y} \in \SSCSubset^{\SSCVars}} \StrongRandPoly(\vec{\alpha},\vec{y})\big)_{\vec{\alpha} \in \Field^{\SCVars}}\Big) : \StrongRandPoly \in \PolynomialRingIndOneXY{\Field}{\SCVars}{\VariableX}{\SSCVars}{\VariableY}{\SCDegree}{\SCDegree'} \Big\}
\end{equation*}
is an $\Field$-vector space with domain $\Field^{\SCVars+\SSCVars} \cup \Field^{\SCVars}$. Consider subdomains $\Field^{\SCVars}$ and $S$. Since $\SetCardinality{S} < \SetCardinality{\SSCSubset}^{\SSCVars}$, by \cref{lem:sum-query-lower-bound} there exist no constants $(c_{\vec{\alpha}})_{\alpha \in \Field^{\SCVars}}$, $(d_{\vec{\gamma}})_{\vec{\gamma} \in S}$ such that the conditions of the claim hold. This concludes the proof of \cref{cor:partial-sum-indep-vars}
\end{proof}

\doclearpage
\section{Zero knowledge sumcheck from algebraic query lower bounds}
\label{sec:strong-zk-sumcheck}

We leverage our lower bounds on the algebraic query complexity of polynomial summation (\cref{sec:algebraic-query-complexity}) to obtain an analogue of the sumcheck protocol with a strong zero knowledge guarantee, which we then use to obtain a zero knowledge low-degree IPCP for $\NEXP$ (\cref{sec:zk-nexp}).

The sumcheck protocol \cite{LundFKN92} is an Interactive Proof for claims of the form $\sum_{\vec{x} \in \SCSubset^{\SCVars}} \SCPoly(\vec{x}) = a$, where $\SCSubset$ is a subset of a finite field $\Field$, $\SCPoly$ is an $\SCVars$-variate polynomial over $\Field$ of individual degree at most $\SCDegree$, and $a$ is an element of $\Field$. The sumcheck protocol is \emph{not} zero knowledge (unless $\sharpP \subseteq \BPP$).

Prior work \cite{BenSassonCFGRS17} obtains a sumcheck protocol, in the IPCP model, with a certain (weak) zero knowledge guarantee. In that protocol, the prover first sends a proof oracle that consists of the evaluation of a random $\SCVars$-variate polynomial $\RandPoly$ of individual degree at most $\SCDegree$; after that, the prover and the verifier run the (standard) sumcheck protocol on a new polynomial obtained from $\SCPoly$ and $\RandPoly$. The purpose of $\RandPoly$ is to \DoQuote{mask} the partial sums, which are the intermediate values sent by the prover during the sumcheck protocol.

The zero knowledge guarantee in \cite{BenSassonCFGRS17} is the following: \emph{any verifier that makes $q$ queries to $\RandPoly$ learns at most $q$ evaluations of $\SCPoly$}. This guarantee suffices to obtain a zero knowledge protocol for $\sharpP$ (the application in \cite{BenSassonCFGRS17}), because the verifier can evaluate $\SCPoly$ efficiently at any point (as $\SCPoly$ is merely an arithmetization of a 3SAT formula).

We achieve a much stronger guarantee: \emph{any verifier that makes polynomially-many queries to $\RandPoly$ learns at most a single evaluation of $\SCPoly$} (that, moreover, lies within a chosen subset $\SoundnessSet^{\SCVars}$ of $\Field^{\SCVars}$). Our application requires this guarantee because we use the sumcheck simulator as a sub-simulator in a larger protocol, where $\SCPoly$ is a randomized low-degree extension of some function that is hard to compute for the verifier. The randomization introduces bounded independence, which makes a small number of queries easy to simulate (where \DoQuote{small} means somewhat less than the degree).

The main idea to achieve the above zero knowledge guarantee is the following. Rather than sending the masking polynomial $\RandPoly$ directly, the prover sends a (perfectly-hiding and statistically-binding) commitment to it in the form of a random $(\SCVars+\SSCVars)$-variate polynomial $\StrongRandPoly$. The \DoQuote{real} mask is recovered by summing out $k$ variables: $\RandPoly(\vec{\VariableX}) \DefineEqual \sum_{\vec{\beta} \in \SSCSubset^{\SSCVars}} \StrongRandPoly(\vec{\VariableX}, \vec{\beta})$. Our lower bounds on the algebraic query complexity of polynomial summation (\cref{sec:algebraic-query-complexity}) imply that any $q$ queries to $\StrongRandPoly$, with $q < \SetCardinality{\SSCSubset}^{\SSCVars}$, yield \emph{no information} about $\RandPoly$. The prover, however, can elect to decommit to $\RandPoly(\vec{c})$, for a single point $\vec{c} \in \SoundnessSet^{\SCVars}$ chosen by the verifier. This is achieved using the weak zero knowledge sumcheck protocol in \cite{BenSassonCFGRS17} as a subroutine: the prover sends $w \DefineEqual \RandPoly(\vec{c})$ and then proves that $w = \sum_{\vec{\beta} \in \SSCSubset^{\SSCVars}} \StrongRandPoly(\vec{c}, \vec{\beta})$.

The protocol thus proceeds as follows. Given a security parameter $\SubsetSize \in \Naturals$, the prover sends the evaluations of two polynomials $\StrongRandPoly \in \PolynomialRingIndOneXY{\Field}{\SCVars}{\VariableX}{\SSCVars}{\VariableY}{\SCDegree}{2\SubsetSize}$ and $\AuxRandPoly \in \PolynomialRingIndOne{\Field}{\SSCVars}{\VariableY}{2\SubsetSize}$ as proof oracles ($\AuxRandPoly$ is the mask for the subroutine in \cite{BenSassonCFGRS17}). The prover sends a field element $z$, which is (allegedly) the summation of $\StrongRandPoly$ over $\SCSubset^{\SCVars} \times \SSCSubset^{\SSCVars}$. The verifier replies with a random challenge $\rho \in \Field \setminus \{0\}$. The prover and the verifier then engage in the standard (not zero knowledge) sumcheck protocol on the claim \DoQuote{$\sum_{\vec{\alpha} \in \SCSubset^{\SCVars}} \rho \SCPoly(\vec{\alpha}) + \RandPoly(\vec{\alpha}) = \rho \SCSum + z$}. This reduces checking the correctness of this claim to checking a claim of the form \DoQuote{$\rho \SCPoly(\vec{c}) + \RandPoly(\vec{c}) = b$}, for some $\vec{c} \in \SoundnessSet^{\SCVars}$ and $b \in \Field$; the prover then decommits to $w \DefineEqual \RandPoly(\vec{c})$ as above. In sum, the verifier deduces that, with high probability, the claim \DoQuote{$\rho \SCPoly(\vec{c}) = b - w$} is true if and only if the original claim was.

If the verifier could evaluate $\SCPoly$, then the verifier could simply check the aforementioned claim and either accept or reject. However, we do not give the verifier access to $\SCPoly$ and, instead, we follow \cite{Meir13,GoldwasserKR15} and phrase sumcheck as a \emph{reduction} from a claim about a sum of a polynomial over a large product space to a claim about the evaluation of that polynomial at a single point. This view of the sumcheck protocol is useful later on when designing more complex protocols, which employ sumcheck as a sub-protocol. The completeness and soundness definitions, which we will formally define in \cref{sec:sumcheck-analysis}, are thus modified according to this viewpoint, where the verifier does \emph{not} have access to $\SCPoly$ and simply outputs the claim at the end of the protocol.

We state a simplified version of the main theorem of this section; the full version is given as \cref{lem:strong-zk-sumcheck}.

\begin{theorem}
\label{lem:strong-zk-sumcheck-short}
For every finite field $\Field$ and $\SCDegree,\SSCVars,\SubsetSize \in \Naturals$, $2\SubsetSize \leq \SCDegree$, there exists an $(\Field,\SCDegree,\SCVars+\SSCVars+1)$-low-degree IPCP system $\pair{\Prover}{\Verifier}$, for the sumcheck problem with respect to polynomials in $\PolynomialRingIndOne{\Field}{\SCVars}{\VariableX}{\SCDegree}$, which is zero knowledge against $\SubsetSize^{\SSCVars}-1$ queries, where the simulator makes a single query to the summand polynomial.
\end{theorem}

We prove \cref{lem:strong-zk-sumcheck-short} in the next subsections, by showing and analyzing a construction that implements the ideas we outlined above. We begin by stating the required preliminaries regarding sampling partial sums of random low-degree polynomials.

\subsection{Sampling partial sums of random low-degree polynomials}
\label{sec:partial-sums}

We recall an algorithm due to Ben-Sasson et al.~\cite{BenSassonCFGRS17} for adaptively sampling random low-degree multivariate polynomials from spaces with exponentially large dimension.

Let $\Field$ be a finite field, $\SCVars,\SCDegree$ positive integers, and $\SCSubset$ a subset of $\Field$. Recall that $\PolynomialRingIndOne{\Field}{\SCVars}{\VariableX}{\SCDegree}$ is the subspace of $\PolynomialRing{\Field}{\SCVars}{\VariableX}$ consisting of those polynomials with individual degrees at most $\SCDegree$. We denote by $\Field^{\leq\SCVars}$ the set of all vectors over $\Field$ of length at most $\SCVars$. Given $Q \in \PolynomialRingIndOne{\Field}{\SCVars}{\VariableX}{\SCDegree}$ and $\vec{\alpha} \in \Field^{\leq\SCVars}$, we define $Q(\vec{\alpha}) \DefineEqual \sum_{\vec{\gamma} \in \SCSubset^{\SCVars - \SetCardinality{\vec{\alpha}}}} Q(\vec{\alpha}, \vec{\gamma})$; that is, the answer to a query that specifies only a prefix of the variables is the sum of the values obtained by letting the remaining variables range over $\SCSubset$.

In \cref{sec:strong-zk-sumcheck} we rely on the fact, formally stated below and proved in \cite{BenSassonCFGRS17}, that one can efficiently sample the distribution $\RandPoly(\vec{\alpha})$, where $\RandPoly$ is uniformly random in $\PolynomialRingIndOne{\Field}{\SCVars}{\VariableX}{\SCDegree}$ and $\vec{\alpha} \in \Field^{\leq\SCVars}$ is fixed, \emph{even conditioned on any polynomial number of (consistent) values for $\RandPoly(\vec{\alpha}_{1}),\dots,\RandPoly(\vec{\alpha}_{\ListSize})$}, for any choice of $\vec{\alpha}_{1},\dots,\vec{\alpha}_{\ListSize} \in \Field^{\leq\SCVars}$. More precisely, the sampling algorithm runs in time that is only $\poly(\log \SetCardinality{\Field}, \SCVars, \SCDegree, \SetCardinality{\SCSubset}, \ListSize)$, which is much faster than the trivial running time of $\Omega(\SCDegree^{\SCVars})$ achieved by sampling $\RandPoly$ explicitly. This \DoQuote{succinct} sampling follows from the notion of \emph{succinct constraint detection} studied in \cite{BenSassonCFGRS17} for the case of partial sums of low-degree polynomials.

\begin{lemma}[\cite{BenSassonCFGRS17}]
\label{lem:efficient-poly-simulator}
There exists a probabilistic algorithm $\CodeSimAlgorithm$ such that, for every finite field $\Field$, positive integers $\SCVars,\SCDegree$, subset $\SCSubset$ of $\Field$, subset $S = \{(\alpha_{1},\beta_{1}), \dots, (\alpha_{\ListSize}, \beta_{\ListSize})\} \subseteq \Field^{\leq\SCVars} \times \Field$, and $(\alpha,\beta) \in \Field^{\leq\SCVars} \times \Field$,
\begin{equation*}
\Pr\Big[
\CodeSimAlgorithm(\Field,\SCVars,\SCDegree,\SCSubset,S,\alpha) = \beta
\Big]
=
\Pr_{\RandPoly \gets \PolynomialRingIndOne{\Field}{\SCVars}{\VariableX}{\SCDegree}}
\left[
\RandPoly(\alpha) = \beta
\pST
\begin{array}{c}
\RandPoly(\alpha_{1}) = \beta_{1} \\
\vdots \\
\RandPoly(\alpha_{\ListSize}) = \beta_{\ListSize}
\end{array}
\right]\enspace.
\end{equation*}
Moreover $\CodeSimAlgorithm$ runs in time $\SCVars(\SCDegree \ListSize \SetCardinality{\SCSubset} + \SCDegree^{3}\ListSize^{3}) \cdot \poly(\log \SetCardinality{\Field}) = \ListSize^{3} \cdot \poly(\SCVars, \SCDegree, \SetCardinality{\SCSubset}, \log \SetCardinality{\Field})$.
\end{lemma}

\subsection{Strong zero knowledge sumcheck}

We present our strong zero knowledge sumcheck protocol within the IPCP model. For brevity, throughout, we will refer to the weak zero knowledge IPCP for sumcheck in \cite{BenSassonCFGRS17} simply as the \DoQuote{weak-ZK sumcheck protocol}.

\begin{construction}
\label{con:strong-sumcheck-ip}
Fix a finite field $\Field$, and $\SCDegree, \SCVars, \SubsetSize \in \Naturals$ with $2\SubsetSize \leq \SCDegree$. Let $\SSCSubset$ be any subset of $\Field$ of size $\SubsetSize$. In the protocol $\pair{\Prover}{\Verifier}$:
\begin{itemize}

  \item $\Prover$ and $\Verifier$ receive an instance $(\SCSubset,\SCSum)$ as common input;
  
  \item $\Prover$ additionally receives a summand polynomial $\SCPoly \in \PolynomialRingIndOne{\Field}{\SCVars}{\VariableX}{\SCDegree}$ as an oracle.

\end{itemize}
The interaction between $\Prover$ and $\Verifier$ proceeds as follows:
\begin{enumerate}

  \item $\Prover$ draws uniformly random polynomials $\StrongRandPoly \in \PolynomialRingIndOneXY{\Field}{\SCVars}{\VariableX}{\SSCVars}{\VariableY}{\SCDegree}{2\SubsetSize}$ and $\AuxRandPoly \in \PolynomialRingIndOne{\Field}{\SSCVars}{\VariableY}{2\SubsetSize}$, and sends as an oracle the polynomial
	\begin{equation*}
		O(W,\vec{\VariableX},\vec{\VariableY}) \DefineEqual W \cdot \StrongRandPoly(\vec{\VariableX},\vec{\VariableY}) + (1-W) \cdot \AuxRandPoly(\vec{\VariableY}) \in \PolynomialRingIndOne{\Field}{\SCVars+\SSCVars+1}{\VariableX}{\SCDegree} \enspace;
	\end{equation*}
	note that $\StrongRandPoly$ can be recovered as $O(1,\cdot)$ and $\AuxRandPoly$ as $O(0,\vec{0},\cdot)$.
	
  \item \label{step:sum-challenge} $\Prover$ sends $z \DefineEqual \sum_{\vec{\alpha} \in \SCSubset^{\SCVars}}\sum_{\vec{\beta} \in \SSCSubset^{\SSCVars}} \StrongRandPoly(\vec{\alpha},\vec{\beta})$ to $\Verifier$.

  \item \label{step:challenge} $\Verifier$ draws a random element $\rho_{1}$ in $\Field \setminus \{0\}$ and sends it to $\Prover$.

  \item $\Prover$ and $\Verifier$ run the \emph{standard} sumcheck IP \cite{LundFKN92} on the statement \DoQuote{$\sum_{\vec{\alpha} \in \SCSubset^{\SCVars}} \MaskedPoly(\vec{\alpha})=\rho_{1}\SCSum+z$} where
\begin{equation*}
\MaskedPoly(\VariableX_{1},\dots,\VariableX_{\SCVars})
\DefineEqual
\rho_{1} \SCPoly(\VariableX_{1},\dots,\VariableX_{\SCVars})
+
\sum_{\vec{\beta} \in \SSCSubset^{\SSCVars}} \StrongRandPoly(\VariableX_{1},\dots,\VariableX_{\SCVars}, \vec{\beta})
\enspace,
\end{equation*}
with $\Prover$ playing the role of the prover and $\Verifier$ that of the verifier, and the following modification.

For $i=1,\dots,\SCVars$, in the $i$-th round, $\Verifier$ samples its random element $c_{i}$ from the set $\SoundnessSet$ rather than from all of $\Field$; if $\Prover$ ever receives $c_{i} \in \Field \setminus \SoundnessSet$, it immediately aborts. In particular, in the $\SCVars$-th round, $\Prover$ sends a polynomial $g_{\SCVars}(\VariableX_{\SCVars}) \DefineEqual \rho_{1} \SCPoly(c_{1}, \dots, c_{\SCVars-1}, \VariableX_{\SCVars}) + \sum_{\vec{\beta} \in \SSCSubset^{\SSCVars}} \StrongRandPoly(c_{1}, \dots, c_{\SCVars-1}, \VariableX_{\SCVars}, \vec{\beta})$ for some $c_{1}, \dots, c_{\SCVars-1} \in \SoundnessSet$.

  \item \label{step:send-cm} $\Verifier$ sends $c_{\SCVars} \in \SoundnessSet$ to $\Prover$.

  \item $\Prover$ sends the element $w \DefineEqual \sum_{\vec{\beta} \in \SSCSubset^{\SSCVars}} \StrongRandPoly(\vec{c}, \vec{\beta})$ to $\Verifier$, where $\vec{c} \DefineEqual (c_{1}, \dots, c_{\SCVars})$.

  \item $\Prover$ and $\Verifier$ engage in the weak-ZK sumcheck protocol with respect to the claim $\sum_{\vec{\beta} \in \SSCSubset^{\SSCVars}} \StrongRandPoly(\vec{c},\vec{\beta}) = w$, using $\AuxRandPoly$ as the oracle. If the verifier in that protocol rejects, so does $\Verifier$.

  \item $\Verifier$ outputs the claim \DoQuote{$\SCPoly(\vec{c}) = \frac{g_{\SCVars}(c_{\SCVars}) - w}{\rho_{1}}$}.

\end{enumerate}
\end{construction}

\begin{remark}
	Formally, the protocol in \cref{lem:strong-zk-sumcheck-short} is not presented as a proper low-degree IPCP, but rather as a reduction with respect to some fixed, yet \emph{inaccessible} low-degree polynomial. Nevertheless, this reduction perspective is consistent with our application, and indeed when in \cref{sec:zk-nexp} we use the protocol in \cref{lem:strong-zk-sumcheck-short} as a sub-procedure, we obtain a low-degree IPCP per our definition in \cref{sec:prelims}.
\end{remark}

\subsection{Analysis of the protocol}
\label{sec:sumcheck-analysis}

The following theorem, which is a more elaborate version of \cref{lem:strong-zk-sumcheck-short}, provides an analysis of \cref{con:strong-sumcheck-ip}. We stress that the protocol will satisfy a relaxed notion of soundness, similar to low-degree soundness, where the \DoQuote{no} instances are required to be low-degree polynomials. This suffices for our applications.

\begin{theorem}
\label{lem:strong-zk-sumcheck}
For every finite field $\Field$, $\SCDegree,\SSCVars,\SubsetSize \in \Naturals$, $2\SubsetSize \leq \SCDegree$, there exists an $(\Field,\SCDegree,\SCVars+\SSCVars+1)$-low-degree IPCP system $\pair{\Prover}{\Verifier}$ such that, for every $\SCPoly \in \PolynomialRingIndOne{\Field}{\SCVars}{\VariableX}{\SCDegree}$, the following holds.
\begin{itemize}

  \item \textsc{Completeness.}
  If $\sum_{\vec{\alpha} \in \SCSubset^{\SCVars}} F(\vec{\alpha}) = \SCSum$, then $\Verifier(\SCSubset,\SCSum)$, when interacting with $\Prover^{\SCPoly}(\SCSubset,\SCSum)$, outputs a true claim of the form \DoQuote{$\SCPoly(\vec{\gamma}) = a$} (with $\vec{\gamma} \in \Field^{\SCVars}$ and $a \in \Field$) with probability $1$.

  \item \textsc{Soundness.}
  If $\sum_{\vec{\alpha} \in \SCSubset^{\SCVars}} F(\vec{\alpha}) \neq \SCSum$, then for any malicious prover $\Malicious{\Prover}$ it holds that $\Verifier(\SCSubset,\SCSum)$, when interacting with $\Malicious{\Prover}$, outputs a true claim \DoQuote{$\SCPoly(\vec{\gamma}) = a$} (with $\vec{\gamma} \in \Field^{\SCVars}$ and $a \in \Field$) with probability at most $\frac{\SCVars\SCDegree}{\SetCardinality{\SoundnessSet}} + \frac{\SSCVars\SCDegree + 2}{\SetCardinality{\Field}-1}$.

  \item \textsc{Zero knowledge.}
  There exists a simulator $\Simulator$ such that if $\sum_{\vec{\alpha} \in \SCSubset^{\SCVars}} F(\vec{\alpha}) = \SCSum$, then for every $\SubsetSize^{\SSCVars}$-query malicious verifier $\Malicious{\Verifier}$, the following two distributions are equal
\begin{equation*}
\Simulator^{\Malicious{\Verifier},\SCPoly}(\SCSubset,\SCSum)
\quad\text{and}\quad
\IPCPView{\Prover^{\SCPoly}(\SCSubset,\SCSum)}{\Malicious{\Verifier}}
\enspace.
\end{equation*}
Moreover:
\begin{itemize}
  \item $\Simulator$ makes a single query to $\SCPoly$ at a point in $\SoundnessSet^{\SCVars}$;
  \item $\Simulator$ runs in time
\begin{equation*}
(\SCVars + \SSCVars)((\SCDegree+\SubsetSize)\QueryComplexity_{\Malicious{\Verifier}}\SetCardinality{\SCSubset} + (\SCDegree+\SubsetSize)^{3}\QueryComplexity_{\Malicious{\Verifier}}^{3}) \cdot \poly(\log\SetCardinality{\Field}) = \poly(\log \SetCardinality{\Field}, \SCDegree, \SCVars, \SubsetSize, \SSCVars, \SetCardinality{\SCSubset}) \cdot \QueryComplexity_{\Malicious{\Verifier}}^{3} \enspace,
\end{equation*}
where $\QueryComplexity_{\Malicious{\Verifier}}$ is $\Malicious{\Verifier}$'s query complexity;
  \item $\Simulator$'s behavior does not depend on $\SCSum$ until after the simulated $\Malicious{\Verifier}$ sends its first message.
\end{itemize}
\end{itemize}
\end{theorem}

\begin{remark}[space complexity]
With two-way access to the random tape, the prover can be made to run in space complexity $\poly(\log \SetCardinality{\Field}, \SCDegree, \SCVars, \SubsetSize, \SSCVars, \SetCardinality{\SCSubset})$.
\end{remark}

\begin{remark}[straightline simulators]
Inspection shows that our simulators, in fact, achieve the stronger notion of universal \emph{straightline simulators} \cite{FeigeS89,DworkS98}, in which the simulator do \emph{not} rewind the verifier.
\end{remark}

\begin{proof}
Completeness is immediate from the protocol description and the completeness property of the sumcheck sub-protocols it invokes. Soundness follows from the fact that, fixing $\SCPoly$ such that $\sum_{\vec{\alpha} \in \SCSubset^{\SCVars}} F(\vec{\alpha}) \neq \SCSum$, we can argue as follows:
\begin{itemize}

  \item For every polynomial $\StrongRandPoly \in \PolynomialRingIndOne{\Field}{\SCVars+\SSCVars}{\VariableX}{\SCDegree}$, with probability $1 - \frac{1}{\SetCardinality{\Field}-1}$ over the choice of $\rho_{1}$ it holds that $\sum_{\vec{\alpha} \in \SCSubset^{\SCVars}} \MaskedPoly(\vec{\alpha}) \neq \rho_{1}\SCSum + z$, i.e., the sumcheck claim is false.

  \item Therefore, by the soundness guarantee of the sumcheck protocol, with probability at least $1 - \SCVars \SCDegree/\SetCardinality{\SoundnessSet}$, either the verifier rejects or $\rho_{1} \SCPoly(\vec{c}) + \sum_{\vec{\beta} \in \SSCSubset^{\SSCVars}} \StrongRandPoly(\vec{c}, \vec{\beta}) \neq g_{\SCVars}(c_{\SCVars})$.

  \item Finally, we distinguish between two cases depending on $\Malicious{\Prover}$:
  \begin{itemize}[nolistsep]
    \item If $\Malicious{\Prover}$ sends $w \neq \sum_{\vec{\beta} \in \SSCSubset^{\SSCVars}} \StrongRandPoly(\vec{c}, \vec{\beta})$, then by the soundness guarantee of the weak-ZK sumcheck protocol, the verifier rejects with probability at least $1 - \frac{\SSCVars\cdot\SCDegree+1}{\SetCardinality{\Field}}$.
    \item If $\Malicious{\Prover}$ sends $w = \sum_{\vec{\beta} \in \SSCSubset^{\SSCVars}} \StrongRandPoly(\vec{c}, \vec{\beta})$, then $\SCPoly(\vec{c}) \neq \frac{g_{\SCVars}(c_{\SCVars}) - w}{\rho_{1}}$ with probability $1$.
  \end{itemize}
\end{itemize}
Taking a union bound on the above cases yields the claimed soundness error.

To show the perfect zero knowledge guarantee, we need to construct a suitably-efficient simulator that perfectly simulates the view of any malicious verifier $\Malicious{\Verifier}$. We first construct an \emph{inefficient} simulator $\SlowSimulator$ and prove that its output follows the desired distribution; afterwards, we explain how the simulator can be made efficient.

\begin{mdframed}
{\small
The simulator $\SlowSimulator$, given (straightline) access to $\Malicious{\Verifier}$ and oracle access to $\SCPoly$, works as follows:

\begin{enumerate}

  \item Draw $\Simulated{\StrongRandPoly} \in \PolynomialRingIndOneXY{\Field}{\SCVars}{\VariableX}{\SSCVars}{\VariableY}{\SCDegree}{2\SubsetSize}$. Run the weak-ZK sumcheck simulator $\Simulator'$.

  \item \label[step]{step:szksc-sim-oracles} Begin simulating $\Malicious{\Verifier}$. Its queries to $O$ are answered by making the appropriate queries to $\Simulated{\StrongRandPoly}$ and the simulated $\AuxRandPoly$ provided by $\Simulator'$.

  \item \label[step]{step:szksc-sim-sums} Send $\Simulated{z} \DefineEqual \sum_{\vec{\alpha} \in \SCSubset^{\SCVars}} \sum_{\vec{\beta} \in \SSCSubset^{\SSCVars}} \Simulated{\StrongRandPoly}(\vec{\alpha},\vec{\beta})$.

  \item \label[step]{step:szksc-first-sumcheck} Receive $\tilde{\rho}$. Draw $\Simulated{\MaskedPoly} \in \PolynomialRingIndOne{\Field}{\SCVars}{\VariableX}{\SCDegree}$ uniformly at random conditioned on $\sum_{\vec{\alpha} \in \SCSubset^{\SCVars}} \Simulated{\MaskedPoly}(\vec{\alpha}) = \tilde{\rho} \SCSum + \Simulated{z}$, then engage in the sumcheck protocol on the claim \DoQuote{$\sum_{\vec{\alpha} \in \SCSubset^{\SCVars}} \Simulated{\MaskedPoly}(\vec{\alpha}) = \tilde{\rho} \SCSum + \Simulated{z}$}. If in any round $\Malicious{\Verifier}$ sends $c_i \not\in \SoundnessSet$ as a challenge, abort.

  \item \label[step]{step:strong-zksc-w} Let $\vec{c} \in \SoundnessSet^{\SCVars}$ be the point chosen by $\Malicious{\Verifier}$ in the sumcheck protocol above. Query $\SCPoly(\vec{c})$, and set $\Simulated{w} \DefineEqual \Simulated{\MaskedPoly}(\vec{c}) - \tilde{\rho} \SCPoly(\vec{c})$; send this value to the verifier.

  \item \label[step]{step:strong-zksc-redraw} Draw $\Simulated{\StrongRandPoly}' \in \PolynomialRingIndOneXY{\Field}{\SCVars}{\VariableX}{\SSCVars}{\VariableY}{\SCDegree}{2\SubsetSize}$ uniformly at random conditioned on
  \begin{itemize}[nolistsep]
    \item $\sum_{\vec{\beta} \in \SSCSubset^{\SSCVars}} \Simulated{\StrongRandPoly}'(\vec{c}, \vec{\beta}) = \Simulated{w}$, and
    \item $\Simulated{\StrongRandPoly}'(\vec{\gamma}) = \Simulated{\StrongRandPoly}(\vec{\gamma})$ for all previous queries $\vec{\gamma}$ to $\StrongRandPoly$.
  \end{itemize} 
  From this point on, answer all queries to $\StrongRandPoly$ with $\Simulated{\StrongRandPoly}'$.

  \item \label[step]{step:strong-zksc-second-sumcheck} Use $\Simulator'$ to simulate the sumcheck protocol for the claim \DoQuote{$\sum_{\vec{\beta} \in \SSCSubset^{\SSCVars}} \Simulated{\StrongRandPoly}'(\vec{c},\vec{\beta}) = \Simulated{w}$}.

  \item Output the view of the simulated $\Malicious{\Verifier}$.
\end{enumerate}
}
\end{mdframed}

To prove that this simulator outputs the correct distribution, we consider the information that the verifier receives in each step and show that the corresponding random variable is distributed identically to the view of the verifier in the real protocol. It will be convenient to define $R(\vec{\VariableX}) = \sum_{\vec{\beta} \in \SSCSubset^{\SSCVars}} \StrongRandPoly(\vec{\VariableX},\vec{\beta})$, and $\Simulated{R}$ likewise. We proceed with a step-by-step analysis (the relevant steps are in \cref{step:szksc-sim-oracles,step:szksc-sim-sums,step:szksc-first-sumcheck,step:strong-zksc-w,step:strong-zksc-redraw,step:strong-zksc-second-sumcheck}).

In \cref{step:szksc-sim-oracles} and \cref{step:szksc-sim-sums}, the verifier has query access to a uniformly random polynomial $\StrongRandPoly$ and receives its summation over $\SCSubset^{\SCVars} \times \SSCSubset^{\SSCVars}$, exactly as in the real protocol. 

In \cref{step:szksc-first-sumcheck}, we simulate the (standard) sumcheck protocol on the polynomial $\Simulated{\MaskedPoly}$, which is chosen uniformly at random conditioned on $\sum_{\vec{\alpha} \in \SCSubset^{\SCVars}} \Simulated{\MaskedPoly}(\vec{\alpha}) = \rho_{1} a + \Simulated{z}$. A key observation is that this is the distribution of $\MaskedPoly$ in the real protocol. To see this, note that by \cref{cor:partial-sum-indep-vars}, we have that $\Malicious{\rho}$, being a function of fewer than $\SubsetSize^{\SSCVars}$ queries to $\StrongRandPoly$, is independent of $R$ given $\sum_{\vec{\alpha} \in \SCSubset^{\SCVars}} R(\vec{\alpha}) = z$. Then $\Malicious{\rho} \SCPoly$ is a random variable conditionally independent of $R$, and so $\MaskedPoly = R + \Malicious{\rho} \SCPoly$ is a uniformly random polynomial such that $\sum_{\vec{\alpha} \in \SCSubset^{\SCVars}} \MaskedPoly(\vec{\alpha}) = \Malicious{\rho} a + z$.\footnote{Note that if $\Malicious{\rho}$ were not independent of $R$ then this may not be true.}

In \cref{step:strong-zksc-w}, we send $\Simulated{w} \Simulated{\MaskedPoly}(\vec{c}) - \Malicious{\rho} \SCPoly(\vec{c})$ to the verifier. In the real protocol, we send $w \sum_{\vec{\beta} \in \SSCSubset^{\SSCVars}} \StrongRandPoly(\vec{c}, \vec{\beta}) = \MaskedPoly(\vec{c}) - \Malicious{\rho} \SCPoly(\vec{c})$, where the latter equality is by the definition of $\MaskedPoly$. Since $\MaskedPoly$ and $\Simulated{\MaskedPoly}$ are identically distributed, then $w$ and$ \Simulated{w}$ are also identically distributed.

In \cref{step:strong-zksc-redraw}, we replace $\Simulated{\StrongRandPoly}$ with a new oracle $\Simulated{\StrongRandPoly}'$ (that is consistent with $\Simulated{\StrongRandPoly}$ on all points in which it was queried), which is a commitment to $\Simulated{R}'$ such that $\Simulated{R}'(\vec{c}) = \Simulated{w}$. Consider any future query to $\StrongRandPoly$, which happens after this replacement. We show that this query is distributed exactly as in the original protocol. By \cref{cor:partial-sum-indep-vars}, the following holds for any $\vec{q} \in \Field^{\SCVars+\SSCVars}, a \in \Field$, where $U \subseteq \Field^{\SCVars+\SSCVars} \times \Field$ is the set of previous query-answer pairs.
\begin{equation*}
\Pr_{\Simulated{Z}'}\left[\Simulated{Z}'(\vec{q}) = a \middle\vert
\begin{array}{c}
\Simulated{Z}'(\vec{\gamma}) = b \quad \forall (\vec{\gamma}, b) \in U \\
\sum_{\vec{\beta} \in \SSCSubset^{\SSCVars}} \Simulated{\StrongRandPoly}'(\vec{c}, \vec{\beta}) = \Simulated{w}
\end{array}
\right]
=
\Pr_{Z}\left[Z(\vec{q}) = a \middle\vert
\begin{array}{c}
Z(\vec{\gamma}) = b \quad \forall (\vec{\gamma}, b) \in U \\
\sum_{\vec{\beta} \in \SSCSubset^{\SSCVars}} Z(\vec{\VariableX}, \vec{\beta}) \equiv Q(\vec{\VariableX}) - \Malicious{\rho} \SCPoly (\vec{\VariableX})
\end{array}
\right]
\end{equation*}
Observe that the left hand side describes the distribution of the answer to oracle query $\vec{q}$ provided by the simulator after we replace the $\StrongRandPoly$-oracle, and the right hand side describes the distribution of the answer to the same query in the real protocol.

In \cref{step:strong-zksc-second-sumcheck}, we make use of the weak-ZK simulator for the decommitment. Since, after the replacement of $\Simulated{\StrongRandPoly}$ by $\Simulated{\StrongRandPoly}'$, the statement we are proving is true, we can use its zero knowledge guarantee. This ensures that the only information the verifier gains is the value $R(\vec{c}) = w$, which we already simulate, and a number of evaluations of $\StrongRandPoly$ equal to the number of queries to $\AuxRandPoly$, which we can fold into the query bound. This concludes the argument for the correctness of the inefficient simulator $\SlowSimulator$.

To complete the proof of zero knowledge, we note that $\SlowSimulator$ can be transformed into an efficient simulator $\Simulator$ by using succinct constraint detection for the Reed--Muller code extended with partial sums \cite{BenSassonCFGRS17}: more precisely, we can use the algorithm of \cref{lem:efficient-poly-simulator} to answer both point and sum queries to $\StrongRandPoly$, $\AuxRandPoly$, and $\MaskedPoly$, in a stateful way, maintaining corresponding tables $\AnsTable{\Simulated{\StrongRandPoly}}$, $\AnsTable{\Simulated{\AuxRandPoly}}$, and $\AnsTable{\Simulated{\MaskedPoly}}$.
\end{proof}

\doclearpage
\section{Zero knowledge low-degree IPCP for $\NEXP$}
\label{sec:zk-nexp}

In this section we use the zero knowledge sumcheck protocol developed in \cref{sec:strong-zk-sumcheck} (along with the \cite{BenSassonCFGRS17} protocol) to build a zero knowledge low-degree IPCP for $\NEXP$, which is the key technical component in our proof of \cref{thm:main-zk}.

Our protocol is based on the IPCP for NEXP of \cite{BabaiFL91}. Recall that in this protocol, the prover first sends a low-degree extension of a $\NEXP$ witness, and then engages in the \cite{LundFKN92} sumcheck protocol on a polynomial related to the instance. To make this zero knowledge, the prover first takes a \emph{randomized} low-degree extension $R$ of the witness (which provides some bounded independence), and then sets the oracle to be an algebraic commitment to $R$. Namely, the prover draws a polynomial uniformly at random subject to the condition that \DoQuote{summing out} a few of its variables yields $R$, and places its evaluation in the oracle.

The prover and verifier then engage in the zero knowledge sumcheck detailed in \cref{sec:strong-zk-sumcheck} with respect to the \cite{BabaiFL91} polynomial. This ensures that the verifier learns nothing through the interaction except for a single evaluation of the summand polynomial, which corresponds to learning a constant number of evaluations of the randomized witness. Bounded independence ensures that these evaluations do not leak any information. The prover decommits these evaluations to the verifier, using the \DoQuote{weak} zero knowledge sumcheck protocol in \cite{BenSassonCFGRS17}.

Following \cite{BabaiFLS91}, the arithmetization encodes bit strings as elements in $\SCSubset^{m}$ for some $\SCSubset$ of size $\poly(\BitSize{B})$, rather than with $\SCSubset=\Bits$ as in \cite{BabaiFL91}, for improved efficiency.

\medskip
We start by defining the \emph{oracle 3-satisfiability problem}, which is the $\NEXP$-complete problem used in \cite{BabaiFL91} to construct two-prover interactive proofs for $\NEXP$.

\begin{definition}[$\OracleSATRelation$]
	\label{def:O3SAT}
	The \defemph{oracle 3-satisfiability relation}, denoted $\OracleSATRelation$, consists of all instance-witness pairs $(\Instance,\Witness)=\big((r, s, B), A\big)$, where $r, s$ are positive integers, $B \colon \Bits^{r + 3s + 3} \to \Bits$ is a boolean formula, and $A \colon \Bits^{s} \to \Bits$ is a function, that satisfy the following condition:
	\begin{equation*}
	\forall\,z \in \Bits^{r},\;
	\forall\,b_{1}, b_{2}, b_{3} \in \Bits^{s},\;
	B\big(z, b_{1}, b_{2}, b_{3}, A(b_{1}), A(b_{2}), A(b_{3})\big) = 1
	\enspace.
	\end{equation*}
\end{definition}

In the rest of this section, we prove the following theorem, which shows that every language in $\NEXP$ has a perfect zero knowledge low-degree IPCP with polynomial communication and query complexity.

\begin{theorem}[PZK low-degree IPCP for $\NEXP$]
	\label{thm:pzk-for-nexp}
	There exists $c \in \Naturals$ such that for any query bound function $\SCStrength(n)$, some integers $d(n),m(n) = O(n^{c} \log \SCStrength)$, and any sequence of fields $\F(n)$ that are extension fields of $\F_{2}$ with $|\F(n)| = \Omega((n^{c} \log \SCStrength)^{4})$, the $\NEXP$-complete relation $\OracleSATRelation$ has a public-coin, non-adaptive \LDIPCP{}, with parameters
	\begin{equation*}
		\LDIPCPparams
		{1/2}
		{O(n, \SCStrength)}
		{\poly(2^{n}, \SCStrength)}
		{\poly(n, \log \SCStrength)}
		{\poly(n, \log \SCStrength)}
		{\PolynomialRingIndOne{\Field}{\SCVars}{\VariableX}{\SCDegree}},
	\end{equation*}
which is zero knowledge with query bound $\SCStrength$.
\end{theorem}

\begin{proof}
	We begin with the arithmetization of the problem.
	
	\parhead{Arithmetization}
	Let $\LD{B} \colon \Field^{m} \to \Field$ be the \DoQuote{direct} arithmetization of the negation of $B$: rewrite $B$ by using ANDs and NOTs; negate its output; replace each $\mathrm{AND}(a,b)$ with $a \cdot b$ and $\mathrm{NOT}(a)$ with $1-a$. For every $\vec{x} \in \Bits^{r + 3s + 3}$, $\LD{B}(\vec{x}) = 0$ if $B(\vec{x})$ is true, and $\LD{B}(\vec{x}) = 1$ if $B(\vec{x})$ is false. Note that $\LD{B}$ is computable in time $\poly(\BitSize{B})$ and has total degree $O(\BitSize{B})$.
	
	Note that $(r, s, B) \in \OracleSATRelation$ if and only if there exists a multilinear function $\LD{A} \colon \Field^{s} \to \Field$ that is boolean on $\Bits^{s}$ such that $\LD{B}(\vec{z}, \vec{b}_{1}, \vec{b}_{2}, \vec{b}_{3}, \LD{A}(\vec{b}_{1}), \LD{A}(\vec{b}_{2}), \LD{A}(\vec{b}_{3})) = 0$, for all $\vec{z} \in \Bits^{r}$, $\vec{b}_{1}, \vec{b}_{2}, \vec{b}_{3} \in \Bits^{s}$.
	
	The requirement that $\LD{A}$ is boolean on $\Bits^{s}$ can be encoded by $2^{s}$ constraints: $\LD{A}(\vec{b})(1 - \LD{A}(\vec{b})) = 0$ for every $\vec{b} \in \Bits^{s}$.
	These constraints can be expressed as follows:
	\begin{align*}
	\left\{
	g_{1}(\vec{\alpha}) \DefineEqual \LD{B}(\vec{z}, \vec{b}_{1}, \vec{b}_{2}, \vec{b}_{3}, \LD{A}(\vec{b}_{1}), \LD{A}(\vec{b}_{2}), \LD{A}(\vec{b}_{3})) = 0 \right\}_{\vec{z} \in \Bits^{r},\, \vec{b}_{i} \in \Bits^{s}} \\
	\left\{
	g_{2}(\vec{\beta}) \DefineEqual \LD{A}(\vec{b})(1 - \LD{A}(\vec{b})) = 0 \right\}_{\vec{b} \in \Bits^{s}}
	\quad\quad\quad\quad\quad\quad\quad
	\end{align*}
	
	Let $F$ be the polynomial over $\Field$ given by
	\begin{equation*}
	F(\vec{\VariableX}, \vec{\VariableY})
	\DefineEqual
	\sum_{\vec{\alpha} \in \Bits^{r + 3s}}
	\left(
	g_{1}(\vec{\alpha}) \vec{\VariableX}^{\vec{\alpha}} + g_{2}(\vec{\alpha}_{[s]}) \vec{\VariableY}^{\vec{\alpha}}
	\right)
	\enspace,
	\end{equation*}
	where $\vec{\VariableX}^{\vec{\alpha}} \DefineEqual \VariableX_{1}^{\alpha_{1}} \cdots \VariableX_{\ell}^{\alpha_{\ell}}$ for $\vec{\alpha} \in \Bits^{\ell}$, and $\vec{\alpha}_{[s]}$ are the first $s$ coordinates in $\vec{\alpha}$.
	
	Note that $F \equiv 0$ if and only if all the above constraints hold. Since $F$ is a polynomial of total degree $r + 3s$, if $F \not\equiv 0$, then $F$ is zero on at most an $\frac{r + 3s}{\SetCardinality{\Field}}$ fraction of points in $\Field^{2(r + 3s)}$.
	
	For $\alpha_{i} \in \Bits$ it holds that $\VariableX_{i}^{\alpha_{i}} = 1 + (\VariableX_{i} - 1)\alpha_{i}$, so we can also write
\begin{align*}
   F(\vec{\VariableX}, \vec{\VariableY})
&= \sum_{\vec{\alpha} \in \Bits^{r + 3s}}
    \left(
	g_{1}(\vec{\alpha}) 
	\cdot \prod_{i=1}^{r+3s} (1 + (\VariableX_{i} - 1)\alpha_{i})
	+ 
	g_{2}(\vec{\alpha}_{[s]})
	\cdot \prod_{i=1}^{r+3s} (1 + (\VariableY_{i} - 1)\alpha_{i})
	\right) \\
&=: \sum_{\vec{\alpha} \in \Bits^{r + 3s}} f(\vec{\VariableX},\vec{\VariableY},\vec{\alpha})
\enspace.
\end{align*}
	
Let $\SCSubset$ be a subfield of $\Field$ of size $\poly(r + s + \log \SCStrength)$; define $m_{1} \DefineEqual r/\log \SetCardinality{\SCSubset}$ and $m_{2} \DefineEqual s/\log \SetCardinality{\SCSubset}$ (assuming without loss of generality that both are integers). For $i \in \{1, 2\}$, let $\gamma_{i} \colon \SCSubset^{m_{i}} \to \Bits^{m_{i} \log \SetCardinality{\SCSubset}}$ be the lexicographic order on $\SCSubset^{m_{i}}$. The low-degree extension $\LD{\gamma}_{i}$ of $\gamma_{i}$ is computable by an arithmetic circuit constructible in time $\poly(\SetCardinality{\SCSubset}, m_{i}, \log \SetCardinality{\Field})$ \cite[Claim 4.2]{GoldwasserKR15}. Let $\gamma \colon \SCSubset^{m_{1} + 3m_{2}} \to \Bits^{r + 3s}$ be such that $\gamma(\vec{\alpha}, \vec{\beta}_{1}, \vec{\beta}_{2}, \vec{\beta}_{3}) = (\gamma_{1}(\vec{\alpha}), \gamma_{2}(\vec{\beta}_{1}), \gamma_{2}(\vec{\beta}_{2}), \gamma_{2}(\vec{\beta}_{3}))$ for all $\vec{\alpha} \in \SCSubset^{m_{1}}$, $\vec{\beta}_{1}, \vec{\beta}_{2}, \vec{\beta}_{3} \in \SCSubset^{m_{2}}$; let $\LD{\gamma} \colon \Field^{m_{1} + 3m_{2}} \to \Field^{r + 3s}$ be its low-degree extension.
	
	We can use the above notation to write $F$ equivalently as
	\begin{align*}
	F(\vec{\VariableX}, \vec{\VariableY}) =
	\sum_{\substack{\vec{\alpha} \in \SCSubset^{m_{1}} \\
			\vec{\beta}_{1}, \vec{\beta}_{2}, \vec{\beta}_{3} \in \SCSubset^{m_{2}}}}
	&g_{1}(\LD{\gamma}(\vec{\alpha}, \vec{\beta}_{1}, \vec{\beta}_{2}, \vec{\beta}_{3}))
	\prod_{i=1}^{r+3s} (1 + (\VariableX_{i} - 1)\LD{\gamma}(\vec{\alpha}, \vec{\beta}_{1}, \vec{\beta}_{2}, \vec{\beta}_{3})_{i}) \\
	&+g_{2}(\LD{\gamma}_{2}(\vec{\beta}_{1}))
	\prod_{i=1}^{r+3s} (1 + (\VariableY_{i} - 1)\LD{\gamma}(\vec{\alpha}, \vec{\beta}_{1}, \vec{\beta}_{2}, \vec{\beta}_{3})_{i}) \enspace.
	\end{align*}
	
	We are now ready to specify the protocol. 
	
	\parhead{Low-degree IPCP for $\OracleSATRelation$}
	Let $\SSCVars \DefineEqual \lceil \log 100\SCStrength / \log \SetCardinality{\SCSubset} \rceil$. The interaction is as follows.
	\begin{enumerate}
		\item The prover draws a polynomial $\StrongRandPoly$ uniformly at random from $\PolynomialRingIndOneXY{\Field}{m_{2}}{\VariableX}{\SSCVars}{\VariableY}{\SetCardinality{\SCSubset}+2}{2\SetCardinality{\SCSubset}}$, subject to the condition that $\sum_{\vec{\beta} \in \SSCSubset^{\SSCVars}} \StrongRandPoly(\vec{\alpha},\vec{\beta}) = A(\gamma_{2}(\vec{\alpha}))$ for all $\vec{\alpha} \in \SCSubset^{m_{2}}$. It then generates an oracle $\Proof_{0}$ for the $\SetCardinality{\SCSubset}^{\SSCVars}$-strong zero knowledge sumcheck protocol (\cref{sec:strong-zk-sumcheck}) on input $(\Field, m_{1}+3m_{2}, \deg{f}, \SCSubset, 0)$ and oracles $\Proof_{1},\Proof_{2},\Proof_{3}$ for the invocation of the weak zero knowledge sumcheck protocol in \cite{BenSassonCFGRS17} on input $(\Field, \SSCVars, 2\SetCardinality{\SCSubset}, \SCSubset, \cdot)$. (In both zero knowledge sumchecks, the oracle message does not depend on the claim itself.) The prover sends an oracle which is the \DoQuote{bundling} of the evaluations of $\StrongRandPoly$ with $(\Proof_{0}, \Proof_{1}, \Proof_{2}, \Proof_{3})$.\footnote{%
By \DoQuote{bundling} we refer to a standard technique of sending a single low-degree polynomial which encodes a list of low-degree polynomials. More precisely, the bundling of $P_{1}(\vec{\VariableX}),\ldots,P_{\ell}(\vec{\VariableX})$ is the polynomial $P(W,\VariableX) := \sum_{\alpha \in S} \Lagrange{S}(W,\alpha) P_{\gamma(\alpha)}(\vec{\VariableX})$, for some $S \subseteq \Field$ such that $|S| = \ell$ and $\gamma : S \to \{1, \ldots, \ell\}$ an ordering of $S$. Observe that
\begin{inparaenum}
\item $P(i,\vec{\VariableX}) \equiv P_{i}(\vec{\VariableX})$ for all $i = 1, \ldots, k$;
\item $\deg{P} = \max\{\deg{P_{1}}, \ldots, \deg{P_{\ell}}, |S|-1\}$;
\item any query to $P$ can be answered by querying each $P_{i}$ at that point, and so the zero knowledge guarantee is unaffected except for reducing the query bound by a factor $\ell$.
\end{inparaenum}
}
		
		\item The verifier chooses $\vec{x}, \vec{y} \in \Field^{r+3s}$ uniformly at random and sends them to the prover. The prover and verifier engage in the zero knowledge sumcheck protocol of \cref{sec:strong-zk-sumcheck} with respect to the claim \DoQuote{$F(\vec{x}, \vec{y}) = 0$} with $\SoundnessSet = \Field \setminus \SCSubset$, using $\Proof_{1}$ as the oracle message. This reduces the claim to checking that $f(\vec{x},\vec{y},\vec{c},\vec{c}'_{1},\vec{c}'_{2},\vec{c}'_{3}) = a$ for uniformly random $\vec{c} \in (\Field \setminus \SCSubset)^{m_{1}}$, $\vec{c}'_{1}, \vec{c}'_{2}, \vec{c}'_{3} \in (\Field \setminus \SCSubset)^{m_{2}}$, and some $a \in \Field$ provided by the prover.
		
		\item The prover provides $h_{i} \DefineEqual A(\gamma_{2}(\vec{c}'_{i}))$ for each $i \in \{1,2,3\}$. The verifier substitutes these values into the expression for $f$ to check the above claims, and rejects if they do not hold.
		
		\item The prover and verifier engage in the zero knowledge sumcheck protocol in \cite{BenSassonCFGRS17} with respect to the claims \DoQuote{$\sum_{\vec{\beta} \in \SCSubset^{\SSCVars}} \StrongRandPoly(\vec{c}'_{i},\vec{\beta}) = h_{i}$}, for each $i \in \{1,2,3\}$, using $\Proof_{i}$ as the oracle message.
	\end{enumerate}
	
	\parhead{Completeness}
	If $((r, s, B), A) \in \OracleSATRelation$, then $F(\vec{\VariableX}, \vec{\VariableY})$ is the zero polynomial; hence $F(\vec{x}, \vec{y}) = 0$ for all $\vec{x}, \vec{y} \in \Field^{r + 3s}$. Completeness follows from the completeness of the zero knowledge sumcheck protocols.
	
	\parhead{Low-degree soundness}
	Suppose that $(r, s, B) \notin \Language(\OracleSATRelation)$, and let $(\Malicious{\StrongRandPoly}, \Malicious{\Proof}_{0}, \Malicious{\Proof}_{1}, \Malicious{\Proof}_{2}, \Malicious{\Proof}_{3})$ be the PCP oracle sent by a malicious prover. By the low-degree soundness condition, this is a collection of polynomials of individual degree at most $\SCDegree$. Let $\Malicious{A} \DefineEqual \sum_{\vec{\beta} \in \SCSubset^{\SSCVars}} \Malicious{\StrongRandPoly}(\vec{\VariableX},\vec{\beta})$, which we think of as playing the role of $\LD{A}(\gamma_{2}(\cdot))$ in $F$. Observe that $\Malicious{A}$ has individual degree at most $d \DefineEqual \SetCardinality{\SCSubset}+2$.
	
	If $(r, s, B) \notin \Language(\OracleSATRelation)$, then there is no choice of $\LD{A}$ such that $F(\vec{\VariableX}, \vec{\VariableY})$ is the zero polynomial. Thus, $F(\vec{x}, \vec{y}) = 0$ with probability at most $(r + 3s)/\SetCardinality{\Field}$ over the choice of $\vec{x}, \vec{y}$. By the soundness of the zero knowledge sumcheck protocol (\cref{lem:strong-zk-sumcheck}), the verifier outputs a false claim \DoQuote{$f(\vec{x},\vec{y},\vec{\alpha}) = a$} with probability at least $1 - O((m_{1}+3m_{2}+\SSCVars)d)/(\SetCardinality{\Field} - \SetCardinality{\SCSubset}))$. If substituting $h_{i}$ for $\LD{A}(\gamma_{2}(\vec{c}'_{i}))$ in $f$ does not yield $a$, then the verifier rejects. Otherwise, it must be the case that for at least one $i \in \{1,2,3\}$, $\Malicious{A}(\vec{c}'_{i}) \neq h_{i}$. By the soundness of the sumcheck protocol in \cite{BenSassonCFGRS17}, the verifier rejects with probability at least $1 - O(\frac{\SSCVars d}{\SetCardinality{\Field}})$. Taking a union bound, the verifier rejects with probability at least $1 - O((m_{1}+3m_{2}+\SSCVars)d/\SetCardinality{\Field}) = 1 - O((r+s+\log \SCStrength) d/\SetCardinality{\Field})$.
	
	\parhead{Zero knowledge}
	Perfect zero knowledge is achieved via the following (straightline) simulator.
\begin{mdframed}
	{\small
		\begin{enumerate}[nolistsep]
			
			\item Draw a uniformly random polynomial $\Simulated{\StrongRandPoly} \in \PolynomialRingIndOneXY{\Field}{m_{2}}{\VariableX}{\SSCVars}{\VariableY}{\SetCardinality{\SCSubset}+2}{2\SetCardinality{\SCSubset}}$. 
			
			\item Invoke the $\SetCardinality{\SCSubset}^{k}$-strong ZK sumcheck simulator on input $(\Field, m_{1}+3m_{2}, \deg{f}, \SCSubset, 0)$, and use it to answer queries to $\Proof_{0}$ throughout. In parallel, run three copies of the simulator for the weak ZK sumcheck in \cite{BenSassonCFGRS17}, with respect to input $(\Field, k, 2\SetCardinality{\SCSubset}, \SCSubset, \cdot)$, and use them to answer queries to $\Proof_{1}, \Proof_{2}, \Proof_{3}$ respectively. (Recall that the behavior of each simulator does not depend on the claim being proven until after the first simulated message, so we can choose these later.)
				
			\item Receive $\vec{x}, \vec{y} \in \Field^{r+3s}$ from $\Malicious{\Verifier}$. 
			
			\item Simulate the strong ZK sumcheck protocol on the claim \DoQuote{$F(\vec{x},\vec{y}) = 0$}. The subsimulator will query $f$ at a single location $\vec{c} \in (\Field - \SCSubset)^{r+3s}$. Reply with the value $f(\vec{x},\vec{y},\vec{c})$, for $\vec{c} = (\vec{c}_{0}, \vec{c}_{1}, \vec{c}_{2}, \vec{c}_{3}) \in (\Field \setminus \SCSubset)^{r+3s}$. Computing this requires the values $\LD{A}(\LD{\gamma}(\vec{c}_{i}))$ for $i \in \{1,2,3\}$; we substitute each of these with $\Simulated{h}^{i} \in \Field$ drawn uniformly at random (except that if $\vec{c}_{i} = \vec{c}_{j}$ for $i \neq j$, then fix $\Simulated{h}^{i} = \Simulated{h}^{j}$).
			
			\item For $i \in \{1,2,3\}$, simulate the weak ZK sumcheck protocol in \cite{BenSassonCFGRS17} with respect to the claim \DoQuote{$\sum_{\vec{\beta} \in \SCSubset^{\SSCVars}} \StrongRandPoly(\vec{\alpha},\vec{\beta}) = \Simulated{h}^{i}$}. Whenever the subsimulator queries $\StrongRandPoly$, answer using $\Simulated{\StrongRandPoly}$.
			
		\end{enumerate}
	}
\end{mdframed}

The verifier's view consists of its interaction with $\Prover$ during the four sumcheck protocols it invokes and its queries to the oracle. The strong zero knowledge sumcheck subsimulator in \cref{sec:strong-zk-sumcheck} guarantees that the queries to $\Proof_{0}$ and the first sumcheck are perfectly simulated given a single query to $f$ at the point $\vec{c} \in (\Field \setminus \SCSubset)^{r+3s}$ chosen by $\Malicious{\Verifier}$. Since $\LD{A}'(\vec{\VariableX}) = \sum_{\vec{\beta} \in \SCSubset^{\SSCVars}} \StrongRandPoly(\vec{\VariableX},\vec{\beta}) \in \PolynomialRingIndOne{\Field}{\SCVars}{\VariableX}{\SetCardinality{\SCSubset}+2}$, the evaluation of $\LD{A}$ at any $3$ points outside of $\SCSubset^{\SCVars}$ does not determine its value at any point in $\SCSubset^{\SCVars}$. In particular, this means that the values of the $h_{i}$'s sent by the prover in the original protocol are independently uniformly random in $\Field$ (except if $\vec{c}_{i} = \vec{c}_{j}$ for $i \neq j$ as above). Thus the $\Simulated{h}^{i}$'s are identically distributed to the $h_{i}$'s, and therefore both the prover message and the simulator's query are perfectly simulated.

The sumcheck simulator in \cite{BenSassonCFGRS17} ensures that the view of the verifier in the rest of the sumchecks is perfectly simulated given $q_{\Malicious{\Verifier}}$ queries to $\StrongRandPoly$, where $q_{\Malicious{\Verifier}}$ is the number of queries the verifier makes across all $\Proof_{i}$, $i \in \{1,2,3\}$. Hence, the number of \DoQuote{queries} the simulator makes to $\Simulated{\StrongRandPoly}$ is strictly less than $100\SCStrength$ (because $\Malicious{\Verifier}$ is $\SCStrength$-query). By \cref{cor:partial-sum-indep-vars}, any set of strictly less than $100\SCStrength$ queries to $\StrongRandPoly$ is independent of $\LD{A}'$, and so the answers are identically distributed to the answers to those queries if they were made to a uniformly random polynomial, which is the distribution of $\Simulated{\StrongRandPoly}$.

Clearly, drawing a uniformly random polynomial in $\Simulated{\StrongRandPoly} \in \PolynomialRingIndOneXY{\Field}{m_{2}}{\VariableX}{\SSCVars}{\VariableY}{\SetCardinality{\SCSubset}+2}{2\SetCardinality{\SCSubset}}$ cannot be done in polynomial time. However, we can instead use the algorithm of \cref{lem:efficient-poly-simulator} to draw $\StrongRandPoly$ (a straightforward modification allows us to handle different degrees in $\vec{\VariableX}, \vec{\VariableY}$; alternatively, we could simply set the degree bound for both to be $2\SetCardinality{\SCSubset}$). The running time of the simulator is then $\poly(\log \SetCardinality{\Field}, m_{1}, m_{2}, k, \SetCardinality{\SCSubset})$.
\end{proof}

\doclearpage
\bookmarksetup{startatroot}
\addtocontents{toc}{\protect\vspace*{\baselineskip}\protect}
\appendix
\section{Reducing query complexity while preserving zero knowledge}
\label{sec:single-unif-q}

We prove \cref{prop:round-reduction} by showing that any low-degree IPCP can be transformed into a low-degree IPCP that makes a single uniform query, at only a small cost in parameters, while \emph{preserving zero knowledge}.

Let $\SCVars, \SCDegree \in \Naturals$, and let $\Field$ be a finite field of size $\SetCardinality{\Field} > (\SCVars \SCDegree/\SoundnessError)^C$. Let $(\Prover,\Verifier)$ be an $\RoundComplexity$-round \LDIPCP{} for a language $\Language$. Denote its oracle by $\Oracle$, query complexity by $\QueryComplexity$, PCP length by $\ProofLength$, communication complexity by $\CommunicationComplexity$, and soundness error by $\SoundnessError = 1/2$.

We transform $(\Prover,\Verifier)$ into a low-degree IPCP $(\Prover',\Verifier')$ for $\Language$ with parameters
	\begin{equation*}
		\LDIPCPparams
		{\SoundnessError' = \SoundnessError + \frac{\SCDegree \QueryComplexity}{|\Field|-\QueryComplexity}}
		{\RoundComplexity' = \RoundComplexity+1}
		{\ProofLength' = \ProofLength}
		{\CommunicationComplexity' = \CommunicationComplexity + \poly(\SCDegree, \QueryComplexity, \SCVars)}
		{\QueryComplexity' = 1}
		{\PolynomialRingIndOne{\Field}{\SCVars}{\VariableX}{\SCDegree}}
	\enspace,
	\end{equation*}
where the new honest verifier's single query is uniformly distributed. Furthermore, if $(\Prover,\Verifier)$ is (perfect) zero knowledge with query bound $\QueryBound$, then $(\Prover',\Verifier')$ is (perfect) zero knowledge with query bound $\QueryBound - (\SCDegree \QueryComplexity + 1)$.

We reduce the query complexity of the IPCP verifier $\Verifier$ from $\QueryComplexity$ to $1$ by using the standard approach of leveraging algebraic structure and additional interaction with the prover, while making sure that the query reduction preserves zero knowledge and that the (single) query that $\Verifier$ makes is uniformly distributed. Specifically, if the verifier wants to query the oracle $\Oracle$ at every point in a set $A$, the verifier asks the prover to provide the restriction of $R$ to a curve that contains all points in $A$. Since we wish to make a single uniform query, rather than asking for the curve of minimal degree (which is unique), we choose the curve at random from all such curves of degree at most $\SetCardinality{A}$. This technique is used by \cite{KalaiR08} to show the same result for general public-coin IPCPs (with standard soundness).\footnote{We remark that the transformation in \cite{KalaiR08} also implicitly assumes that the IPCP is public coin.} Since we require only low-degree soundness, we can dramatically simplify their proof. Consider the following protocol.

\begin{construction}
Let $(\Prover,\Verifier)$ be a $\QueryComplexity$-query \LDIPCP{} for $\Language$ in which the (honest) oracle is some $\Oracle \in \PolynomialRingIndOne{\Field}{\SCVars}{\VariableX}{\SCDegree}$. We construct an \LDIPCP{} $(\Prover',\Verifier')$ for $\Language$ in which $\Verifier'$ makes a single uniformly distributed query to $\Oracle$. We may assume that $\SetCardinality{\Field} > \QueryComplexity$, otherwise the stated soundness guarantee is trivial.
\begin{enumerate}[nolistsep]

  \item \emph{Random curve.}
$\Verifier'$ chooses a random $\vec{r} \in \Field^\SCVars$ and a random $t \in \Field \setminus S$, for some $S \subseteq \Field$ with $|S| = \QueryComplexity$. $\Verifier'$ computes a curve $\gamma \colon \Field \to \Field^\SCVars$ of degree $\QueryComplexity$ such that $\{\gamma(s)\}_{s \in S} = A$ and $\gamma(t) = \vec{r}$ and sends it to $\Prover'$. $\Prover'$ replies with the coefficients of the polynomial $\rho \colon \F \to \F$, of degree at most $\SCDegree \QueryComplexity$, that (allegedly) is the restriction of $\RandPoly$ to $\gamma$.

  \item \label{step:kr_consistency}
  \emph{Consistency.}
$\Verifier$ queries $\Oracle$ at $\vec{r}$ and receives an answer $a$; it rejects if $a \neq \rho(t)$. (Recall that $\vec{r}=\gamma(t)$.)
			
  \item \emph{Emulating the multi-query verifier.}
  $\Verifier'$ rules according to the decision predicate of $\Verifier$ with respect to the transcript of the \DoQuote{interaction phase} and the answers to the query set $A$, as indicated by the curve $\gamma$.

\end{enumerate}
\end{construction}
One can verify that the complexity of the protocol is as stated. Since the evaluation of the curve $\gamma$ (which has degree $\QueryComplexity$) at a random $t \in \Field \setminus S$ is a random variable that is uniformity distributed over $\Field^\SCVars$, the single query that the verifier makes is uniformly distributed.
Completeness is immediate by construction. For soundness, since $\rho$ is a univariate polynomial of degree at most $\SCDegree \QueryComplexity$, and $\Prover'$ does not know $t$, if the check $\rho(t) = \Oracle(\vec{r})$ (in Step \ref{step:kr_consistency}) 
passes with probability greater than $\frac{\SCDegree \QueryComplexity}{|\Field|-\QueryComplexity}$, then all the answers to the query set $A$, as indicated by $\gamma$, are consistent with $\Oracle$.
Perfect zero knowledge is preserved because the polynomial $\rho$ can be computed efficiently by making $\SCDegree \QueryComplexity + 1$ queries to $\Oracle$, and so the additional information provided by the prover could have been computed by the malicious verifier itself.


\doclearpage
\section{From PCP to \MIPStar{} via a black box transformation}
\label{sec:blackbox}

We show that any (non-adaptive) PCP, and more generally any IPCP, can be transformed into an \MIPStar{} in a black box way. While a proof of this fact is implicit in \cite{Vidick16,NatarajanV17}, the machinery developed in \cref{part:I} allows us to elucidate its structure and give a compellingly short proof of it.

We first define what we mean by \emph{black box}. Informally, we call a transformation black box if it does not depend on the language being decided.\footnote{This rules out degenerate \DoQuote{transformations} that ignore the given PCP or IPCP $(\Prover,\Verifier)$ for the language $\Language$, and simply output an \MIPStar{} for $\Language$ unrelated to $(\Prover,\Verifier)$. (Such degenerate transformations are trivially implied by the inclusion $\NEXP \subseteq \MIPS$.)} The following definition formalizes this notion.

\begin{definition}
\label{def:blackboxtrans}
A transformation $\Transformation$ maps IPCP to \MIPStar{} if, given as input an IPCP $(\Prover,\Verifier)$ for a language $\Language$, outputs an \MIPStar{} for the language $\Language$. Such a transformation is \emph{black box} if the verifier in the resulting \MIPStar{} can be expressed as an algorithm with access only to the queries and messages of $\Verifier(x)$ but no access to the input $x$, apart from its length.
\end{definition}
\noindent We stress that \cref{def:blackboxtrans} also applies to PCPs (by viewing them as $0$-round IPCPs).

The transformation is in two stages. First, we convert the IPCP into a low-degree IPCP by encoding the oracle as a low-degree polynomial. Second, we apply the transformation in \cref{lem:lifting} (which includes invoking the query reduction in \cref{prop:round-reduction}) to convert the low-degree IPCP into an \MIPStar{}. Besides being black box in the formal sense, the resulting \MIPStar{} verifier is simple to describe: it is the original verifier, composed with an interactive query-reduction protocol and a low-degree test for entangled provers.

In sum, the above yields the following corollary.

\begin{corollary}
\label{cor:ipcp-to-mips}
There exists a black box transformation that maps any $\RoundComplexity$-round IPCP for a language $\Language$ to a $2$-prover $(\RoundComplexity+1)$-round \MIPStar{} for $\Language$.
\end{corollary}

\begin{proof}
Let $\pair{\Prover}{\Verifier}$ be an $\RoundComplexity$-round IPCP for $\Language$; denote its PCP oracle by $\Oracle$, query complexity by $\QueryComplexity$, and proof length by $\ProofLength$. Let $\SCVars,\SCDegree \in \Naturals$ be such that $\ProofLength \leq \SCDegree^{\SCVars}$, and let $\Field$ be a finite field with $\SetCardinality{\Field} > \max\Set{(2\SCVars \SCDegree)^C, 5\SCDegree\QueryComplexity}$ (where $C$ is the constant from \cref{thm:quantum_low_degree_test}).
 Encode the oracle $\Oracle$ of the IPCP as a polynomial $\LD{\Oracle} \in \PolynomialRingIndOne{\Field}{\SCDegree}{\VariableX}{\SCVars}$ by computing the low-degree extension of $\Oracle$ (see \cref{sec:notation}); the new oracle is the evaluation of $\LD{\Oracle}$ over $\Field^{\SCVars}$. Since this encoding is systematic, the verifier $\Verifier$ can directly query $\LD{\Oracle}$ at the positions that correspond to its query set. Completeness and soundness are clearly preserved, as are the query and communication complexities. Proof length is increased from $\ProofLength$ to $|\Field|^{\SCVars}$. The resulting IPCP satisfies the conditions of a low-degree IPCP, and so we can apply the exact argument as in \cref{sec:putting-it-all-together} to obtain the desired \MIPStar{}. Straightforward inspection shows that all of the applied transformations are indeed black box.
\end{proof}

\begin{remark}
Zero knowledge is \emph{not} preserved by this transformation because taking the low-degree extension of the oracle may allow the verifier to learn global information that cannot, in general, be simulated via a small number of queries. (For example, a single point of the low-degree extension may amount to a summation over exponentially many points; see \cref{sec:query-complexity-appendix}.)
\end{remark}

If we apply \cref{cor:ipcp-to-mips} to any PCP (i.e., any $0$-round IPCP) for $\NEXP$ (for example, the one in \cite{BabaiFLS91}), we immediately recover the following result from \cite{NatarajanV17}, which shows that every language in $\NEXP$ has an \MIPStar{} with optimal round complexity and number of provers.

\begin{corollary}
Every language in $\NEXP$ has a 1-round 2-prover \MIPStar{}.
\end{corollary}

\doclearpage
\section{Algebraic query complexity upper bounds}
\label{sec:query-complexity-appendix}

We show that in certain cases the degree constraints in \cref{lem:sum-query-lower-bound} are tight.

\subsection{Multilinear polynomials}
\label{sec:multilinear-polynomials}

The first result is for the case of multivariate polynomials over any finite field, where $\SCSubset \subseteq \Field$ is arbitrary. The proof is a simple extension of a proof due to \cite{JumaKRS09} for the case $\SCSubset = \Bits$.

\begin{theorem}[multilinear polynomials]
\label{thm:multilinear-upper-bound}
Let $\Field$ be a finite field, $\SCSubset$ a subset of $\Field$, and $\gamma \DefineEqual \sum_{\alpha \in \SCSubset} \alpha$. For every $\PolyA \in \PolynomialRingIndOne{\Field}{\SCVars}{\VariableX}{1}$ (i.e., for every $\SCVars$-variate multilinear polynomial $\PolyA$) it holds that
	\begin{equation*}
	\sum_{\vec{\alpha} \in \SCSubset^{\SCVars}} \PolyA(\vec{\alpha})
	=
	\begin{cases}
	\PolyA\big(\frac{\gamma}{\SetCardinality{\SCSubset}},\dots,\frac{\gamma}{\SetCardinality{\SCSubset}}\big) \cdot \SetCardinality{\SCSubset}^{\SCVars} & \text{ if } \Characteristic{\Field} \nmid \SetCardinality{\SCSubset} \\
	\kappa \cdot \gamma^{\SCVars} & \text{ if } \Characteristic{\Field} \mid \SetCardinality{\SCSubset}
	\end{cases}
	\enspace,
	\end{equation*}
	where $\kappa$ is the coefficient of $\VariableX_{1} \cdots \VariableX_{\SCVars}$ in $\PolyA$.
\end{theorem}

\begin{proof}
	First suppose that $\Characteristic{\Field}$ does not divide $\SetCardinality{\SCSubset}$. Let $\vec{\alpha}$ be uniformly random in $\SCSubset^{\SCVars}$; in particular, $\alpha_{i}$ and $\alpha_{j}$ are independent for $i \neq j$. For every monomial $m(\vec{\VariableX}) = \VariableX_{1}^{e_{1}} \cdots \VariableX_{\SCVars}^{e_{\SCVars}}$ with $e_{1},\dots,e_{\SCVars} \in \Bits$,
	\begin{equation*}
	\Expectation[M(\vec{\alpha})]
	= \Expectation[\alpha_{1}^{e_{1}} \cdots \alpha_{\SCVars}^{e_{\SCVars}}]
	= \Expectation[\alpha_{1}^{e_{1}}] \cdots \Expectation[\alpha_{\SCVars}^{e_{\SCVars}}]
	= \Expectation[\alpha_{1}]^{e_{1}} \cdots \Expectation[\alpha_{\SCVars}]^{e_{\SCVars}}
	= M(\Expectation[\alpha_{1}], \dots, \Expectation[\alpha_{\SCVars}])
	\enspace.
	\end{equation*}
	Since $\PolyA$ is a linear combination of monomials, $\Expectation[\PolyA(\vec{\alpha})] = \PolyA(\Expectation[\vec{\alpha}])$. Each $\alpha_{i}$ is uniformly random in $\SCSubset$, so $\Expectation[\alpha_{i}] = \frac{1}{\SetCardinality{\SCSubset}} \sum_{\alpha \in \SCSubset} \alpha = \frac{\gamma}{\SetCardinality{\SCSubset}}$, and thus $\PolyA(\Expectation[\vec{\alpha}]) = \PolyA(\frac{\gamma}{\SetCardinality{\SCSubset}},\dots,\frac{\gamma}{\SetCardinality{\SCSubset}})$, which implies that $\Expectation[\PolyA(\vec{\alpha})] = \PolyA(\frac{\gamma}{\SetCardinality{\SCSubset}},\dots,\frac{\gamma}{\SetCardinality{\SCSubset}})$. To deduce the claimed relation, it suffices to note that $\Expectation[\PolyA(\vec{\alpha})] = \frac{1}{\SetCardinality{\SCSubset}^{\SCVars}}\sum_{\vec{\alpha} \in \SCSubset^{\SCVars}} \PolyA(\vec{\alpha})$.
	
	Next suppose that $\Characteristic{\Field}$ divides $\SetCardinality{\SCSubset}$. For every monomial $m(\vec{\VariableX}) = \VariableX_{1}^{e_{1}} \cdots \VariableX_{\SCVars}^{e_{\SCVars}}$ with $e_{1},\dots,e_{\SCVars} \in \Bits$:
	\begin{itemize}
		
		\item if there exists $j \in [\SCVars]$ such that $e_j = 0$ then
		\begin{equation*}
		\sum_{\vec{\alpha} \in \SCSubset^{\SCVars}} M(\vec{\alpha})
		= \SetCardinality{\SCSubset}
		\sum_{\alpha_{1},\dots,\alpha_{j-1},\alpha_{j+1},\dots,\alpha_{\SCVars} \in \SCSubset}
		\alpha_{1}^{e_{1}} \cdots \alpha_{j-1}^{e_{j-1}} \alpha_{j+1}^{e_{j+1}} \cdots \alpha_{\SCVars}^{e_{\SCVars}}
		= 0
		\enspace.
		\end{equation*}
		
		\item if instead $e_{1} = \cdots = e_{\SCVars} = 1$ then
		\begin{equation*}
		\sum_{\vec{\alpha} \in \SCSubset^{\SCVars}} M(\vec{\alpha})
		= \sum_{\vec{\alpha} \in \SCSubset^{\SCVars}}\prod_{i=1}^{\SCVars} \alpha_{i}
		= \prod_{i=1}^{\SCVars} \sum_{\alpha_{i} \in \SCSubset} \alpha_{i}
		= \left(\sum_{\alpha \in \SCSubset} \alpha \right)^{\SCVars}
		\enspace. \qedhere
		\end{equation*}
	\end{itemize} 
\end{proof}

The following corollary shows that for prime fields of odd size, the value of $\sum_{\vec{\alpha} \in \SCSubset^{\SCVars}} \PolyA(\vec{\alpha})$ can be computed efficiently for any $\SCSubset \subseteq \Field$ using at most a single query to $\PolyA$.

\begin{corollary}
\label{cor:prime-field-multilinear}
	Let $\Field$ be a prime field of odd size, $\SCSubset$ a subset of $\Field$, and $\gamma \DefineEqual \sum_{\alpha \in \SCSubset} \alpha$. For every $\PolyA \in \PolynomialRingIndOne{\Field}{\SCVars}{\VariableX}{1}$ (i.e., for every $\SCVars$-variate multilinear polynomial $\PolyA$) it holds that
	\begin{equation*}
	\sum_{\vec{\alpha} \in \SCSubset^{\SCVars}} \PolyA(\vec{\alpha})
	=
	\begin{cases}
	\PolyA\big(\frac{\gamma}{\SetCardinality{\SCSubset}},\dots,\frac{\gamma}{\SetCardinality{\SCSubset}}\big) \cdot \SetCardinality{\SCSubset}^{\SCVars} & \text{ if } \Characteristic{\Field} \nmid \SetCardinality{\SCSubset} \\
	0 & \text{ if } \Characteristic{\Field} \mid \SetCardinality{\SCSubset}
	\end{cases}
	\enspace.
	\end{equation*}
\end{corollary}

\begin{proof}
\cref{thm:multilinear-upper-bound} implies both cases. If $\Characteristic{\Field}$ does not divide $\SetCardinality{\SCSubset}$, then the claimed value is as in the theorem. If instead $\Characteristic{\Field}$ divides $\SetCardinality{\SCSubset}$, then it must be the case that $\SCSubset = \Field$, since $p \DefineEqual \Characteristic{\Field}$ equals $\SetCardinality{\Field}$; in this case, $\gamma = \sum_{\alpha \in \SCSubset} \alpha = (p-1)p/2$, which is divisible by $p$ since $2$ must divide $p-1$ (as $p$ is odd).
\end{proof}

\subsection{Subsets with group structure}
\label{sec:case-of-groups}

We show that if $\SCSubset$ is assumed to have some group structure, then few queries may suffice even for polynomials of degree greater than one. In particular, \cref{lem:multiplicative-group} shows that if $\SCSubset$ is a multiplicative subgroup of $\Field$ and $\SCDegree \leq \SetCardinality{\SCSubset}$, then one query suffices; \cref{lem:additive-group} shows that if $\SCSubset$ is an additive subgroup of $\Field$, then the answer depends on a polynomial related to $\SCSubset$.

\begin{lemma}[multiplicative groups]
\label{lem:multiplicative-group}
Let $\Field$ be a field, $\RMSubDomain$ a finite multiplicative subgroup of $\Field$, and $\RMVars,\RMDegree$ positive integers with $\RMDegree < \SetCardinality{\RMSubDomain}$. For every $\PolyA \in \PolynomialRingIndOne{\Field}{\RMVars}{\VariableX}{\RMDegree}$,
	\begin{equation*}
	\sum_{\vec{\alpha} \in \RMSubDomain^{\RMVars}} \PolyA(\vec{\alpha})
	= \PolyA(0,\dots,0) \cdot \SetCardinality{\RMSubDomain}^{\RMVars}
	\enspace.
	\end{equation*}
\end{lemma}

\begin{remark}
	\label{rem:necessity-of-degree-bound}
	The hypothesis that $\RMDegree < \SetCardinality{\RMSubDomain}$ is necessary for the lemma, as we now explain. Choose $\RMSubDomain = \Multiplicative{\SubField}$, where $\SubField$ is a proper subfield of $\Field$, $\RMVars=1$, and $\RMDegree=\SetCardinality{\RMSubDomain}$. Consider the polynomial $\VariableX^{\SetCardinality{\RMSubDomain}}$, which has degree at least $\RMDegree$: $\VariableX^{\SetCardinality{\RMSubDomain}}$ vanishes on $0$; however, $\VariableX^{\SetCardinality{\RMSubDomain}}$ evaluates to $1$ everywhere on $\RMSubDomain$ so that its sum over $\RMSubDomain$ equals $\SetCardinality{\RMSubDomain} \neq 0$. (Note that if $\RMSubDomain$ is a multiplicative subgroup of $\Field$ then $\Characteristic{\Field} \nmid \SetCardinality{\RMSubDomain}$ because $\SetCardinality{\RMSubDomain}$ equals $\Characteristic{\Field}^{k}-1$ for some positive integer $k$.)
\end{remark}

\begin{proof}
	The proof is by induction on the number of variables $\RMVars$. The base case is when $\RMVars=1$, which we argue as follows. The group $\RMSubDomain$ is cyclic, because it is a (finite) multiplicative subgroup of a field; so let $\omega$ generate $\RMSubDomain$. Writing $\PolyA(\VariableX_{1})=\sum_{j=0}^{\RMDegree} \beta_{j} \VariableX_{1}^{j}$ for some $\beta_{0},\dots,\beta_{\RMDegree} \in \Field$, we have
	\begin{equation*}
	\sum_{\alpha_{1} \in \RMSubDomain} \PolyA(\alpha_{1})
	= \sum_{i=0}^{\SetCardinality{\RMSubDomain}-1} \PolyA(\omega^i)
	= \sum_{i=0}^{\SetCardinality{\RMSubDomain}-1} \sum_{j=0}^{\RMDegree} \beta_{j} \omega^{ij}
	= \sum_{j=0}^{\RMDegree} \beta_{j} \sum_{i=0}^{\SetCardinality{\RMSubDomain}-1} (\omega^{j})^{i}
	= \beta_{0} \SetCardinality{\RMSubDomain}
	= f(0) \SetCardinality{\RMSubDomain}
	\enspace,
	\end{equation*}
	which proves the base case. The second-to-last equality follows from the fact that for every $\gamma \in \RMSubDomain$,
	\begin{equation*}
	\sum_{i=0}^{\SetCardinality{\RMSubDomain}-1} \gamma^{i}
	=
	\begin{cases}
	\SetCardinality{\RMSubDomain} & \text{if $\gamma=1$} \\
	\frac{\gamma^{\SetCardinality{\RMSubDomain}}-1}{\gamma-1}=0 & \text{if $\gamma \neq 1$}
	\end{cases}
	\enspace.
	\end{equation*}
	
	For the inductive step, assume the statement for any number of variables less than $\RMVars$; we now prove that it holds for $\RMVars$ variables as well. Let $\PolyA_{\alpha}$ denote $\PolyA$ with the variable $\VariableX_{1}$ fixed to $\alpha$. Next, apply the inductive assumption below in the second equality (with $\RMVars-1$ variables) and last one (with $1$ variable), to obtain
	\begin{align*}
	\sum_{\vec{\alpha} \in \RMSubDomain^{\RMVars}} \PolyA(\alpha_{1},\dots,\alpha_{\RMVars})
	&= \sum_{\alpha_{1}\in \RMSubDomain} \sum_{(\alpha_{2},\dots,\alpha_{\RMVars})\in \RMSubDomain^{\RMVars-1}} \PolyA_{\alpha_{1}}(\alpha_{2},\dots,\alpha_{\RMVars}) \\
	&= \SetCardinality{\RMSubDomain}^{\RMVars-1} \sum_{\alpha_{1} \in \RMSubDomain} \PolyA_{\alpha_{1}}(0^{\RMVars-1}) \\
	&= \SetCardinality{\RMSubDomain}^{\RMVars-1} \sum_{\alpha_{1} \in \RMSubDomain} \PolyA(\alpha_{1},0,\dots,0) \\
	&= \SetCardinality{\RMSubDomain}^{\RMVars} \PolyA(0,\dots,0)
	\enspace,
	\end{align*}
	as claimed.
\end{proof}

\begin{lemma}[additive groups]
\label{lem:additive-group}
Let $\Field$ be a field, $\RMSubDomain$ a finite additive subgroup of $\Field$, and $\RMVars,\RMDegree$ positive integers with $\RMDegree < \SetCardinality{\RMSubDomain}$. For every $\vec{v} \in \Field^{\RMVars}$, $\PolyA \in \PolynomialRingIndOne{\Field}{\RMVars}{\VariableX}{\RMDegree}$,
	\begin{equation*}
	\sum_{\vec{\alpha} \in \RMSubDomain^{\RMVars}} \PolyA(\vec{\alpha} + \vec{v})
	= \kappa \cdot a_{0}^{\RMVars} \enspace,
	\end{equation*}
	where $\kappa$ is the coefficient of $\VariableX_{1}^{\SetCardinality{\RMSubDomain}-1} \cdots \VariableX_{\RMVars}^{\SetCardinality{\RMSubDomain}-1}$ in $\PolyA$, and $a_{0}$ is the (formal) linear term of the subspace polynomial $\prod_{h \in \RMSubDomain} (\VariableX - h)$. In particular, if $\PolyA$ has total degree strictly less than $\RMVars (\SetCardinality{\RMSubDomain}-1)$, then the above sum evaluates to $0$.
\end{lemma}
\begin{proof}
	Without loss of generality, let $\RMDegree \DefineEqual \SetCardinality{\RMSubDomain} - 1$. The proof is by induction on the number of variables $\RMVars$. When $\RMVars = 1$, we have that $\PolyA(\VariableX) = \sum_{j=0}^{\RMDegree} \beta_{j} \VariableX^{j}$ for some $\beta_{0}, \ldots, \beta_{\RMDegree} \in \Field$. Then
	\begin{equation*}
	\sum_{\alpha \in \RMSubDomain} \PolyA(\alpha + v)
	= \sum_{\alpha \in \RMSubDomain} \sum_{j=0}^{\RMDegree} \beta_{j} (\alpha+v)^{j}
	= \sum_{j=0}^{\RMDegree} \beta_{j} \sum_{\alpha \in \RMSubDomain} (\alpha+v)^{j}
	= \beta_{\RMDegree} a_{0}
	\end{equation*}
	where the final equality follows by \cite[(Proof of) Theorem 1]{ByottC99}, and the fact that $d = \SetCardinality{\RMSubDomain} - 1$.
	
	For the inductive step, assume the statement for $\RMVars-1$ variables; we now prove that it holds for $\RMVars$ variables as well. Let $\PolyA_{\alpha}$ denote $\PolyA$ with the variable $\VariableX_{1}$ fixed to $\alpha$; we have $\PolyA_{\alpha}(\VariableX_{2}, \dots, \VariableX_{\RMVars}) = \sum_{\vec{e} \in \{0,\dots,\RMDegree\}^{\RMVars}} \beta_{\vec{e}} \cdot \alpha^{e_{1}} \VariableX_{2}^{e_{2}} \dots \VariableX_{\RMVars}^{e_{\RMVars}}$. Next, apply the inductive hypothesis below in the second equality (with $\RMVars-1$ variables) to obtain
	\begin{equation*}
	\sum_{\vec{\alpha} \in \RMSubDomain^{\RMVars}} \PolyA(\vec{\alpha} + \vec{v})
	= \sum_{\alpha_{1}\in \RMSubDomain} \sum_{(\alpha_{2},\dots,\alpha_{\RMVars})\in \RMSubDomain^{\RMVars-1}} \PolyA_{\alpha_{1} + v_{1}}(\alpha_{2} + v_{2},\dots,\alpha_{\RMVars} + v_{\RMVars})
	= \sum_{\alpha_{1}\in \RMSubDomain} a_{0}^{\RMVars-1} \kappa(\alpha_{1}+v_{1}) \enspace,
	\end{equation*}
	where $\kappa(\VariableX_{1}) \DefineEqual \sum_{j=0}^{\RMDegree} \beta_{(j, \RMDegree, \dots, \RMDegree)} \VariableX_{1}^{j}$. Applying the hypothesis again for $1$ variable yields
	\begin{equation*}
	\sum_{\alpha_{1}\in \RMSubDomain} a_{0}^{\RMVars-1} \kappa(\alpha_{1} + v_{1}) = a_{0}^{\RMVars} \cdot \beta_{(\RMDegree, \dots, \RMDegree)} \enspace,
	\end{equation*}
	and the claim follows.
\end{proof}

\clearpage
\section*{Acknowledgments}
We are grateful to Thomas Vidick for multiple technical and conceptual suggestions that greatly improved our results and their presentation, as well as for allowing us to include his proof of \cref{thm:quantum_low_degree_test}. We also thank Zeph Landau, Chinmay Nirkhe, and Igor Shinkar for helpful discussions.

\printbibliography
\end{document}